\documentclass[a4paper,fleqn]{cas-dc}
\usepackage[numbers,sort&compress]{natbib}
\usepackage{import}
\usepackage{hyperref}
\usepackage{longtable}
\usepackage{amsmath}
\usepackage{amsfonts}
\usepackage{enumerate}
\usepackage{epstopdf}
\usepackage{subcaption}
\usepackage{algpseudocode}
\usepackage{lscape}
\usepackage{amssymb}
\usepackage[linesnumbered,ruled,vlined]{algorithm2e}
\usepackage{url}
\usepackage{amsthm}
\usepackage{array,graphicx}
\usepackage{setspace}
\usepackage{geometry}
\usepackage{multirow}
\usepackage{appendix}
\usepackage{mathtools}
\usepackage{wrapfig}
\usepackage{sidecap}
\usepackage{textcomp}
\usepackage{calc}
\usepackage{pgfplots}
\pgfplotsset{compat=1.16}
\usepgfplotslibrary{statistics}
    
\usepackage{colortbl}
\usepackage{hhline}
\usepackage{makecell}
\usepackage{rotating}
\usepackage[utf8]{inputenc}
\usepackage[english]{babel}

\newtheorem{corollary}{Corollary}

\usepackage{supertabular,booktabs}

\def\tsc#1{\csdef{#1}{\textsc{\lowercase{#1}}\xspace}}
\tsc{WGM}
\tsc{QE}
\tsc{EP}
\tsc{PMS}
\tsc{BEC}
\tsc{DE}

\begin{document}
\let\WriteBookmarks\relax
\def\floatpagepagefraction{1}
\def\textpagefraction{.001}
\shorttitle{HIEL approach for CPDP and its Analysis on the Project Attributes}

\title [mode = title]{How Far Does the Predictive Decision Impact the Software Project? The Cost, Service Time, and Failure Analysis from a Cross-Project Defect Prediction Model}
\author[1]{Umamaheswara Sharma B}
\cormark[1]
\fnmark[1]
\ead[url]{uma.phd@student.nitw.ac.in}

\author[1]{Ravichandra Sadam}
\fnmark[2]
\ead[url]{ravic@nitw.ac.in}
\address[1]{Department of Computer Science and Engineering, National Institute of Technology, Warangal}

\cortext[cor1]{Corresponding author}
\fntext[fn1]{Research Scholar, Department of Computer Science and Engineering, National Institute of Technology, Warangal, Telangana, India, 506001}

\fntext[fn2]{Associate Professor, Department of Computer Science and Engineering, National Institute of Technology, Warangal, Telangana, India, 506001}

\begin{abstract}
\textbf{\textbf{Context:}} Cross-project defect prediction (CPDP) models are being developed to optimize the testing resources.\\
\textbf{Objectives:} Proposing {an ensemble} classification framework for CPDP as many existing models are lacking with better performances and analysing the main objectives of CPDP from the outcomes of the proposed classification framework.\\
\textbf{Method:} For the classification task, we propose a bootstrap aggregation based hybrid-inducer ensemble learning (HIEL) technique that uses probabilistic weighted majority voting (PWMV) {strategy}. To know the impact of HIEL on the software project, we propose three {project-specific} performance measures such as percent of perfect cleans (PPC), percent of non-perfect cleans (PNPC), and false omission rate (FOR) from the predictions to calculate the amount of saved cost, remaining service time, and percent of the failures in the target project.\\
\textbf{Results:} {On many target projects from PROMISE, NASA, and AEEEM repositories, the proposed model outperformed recent works such as TDS, TCA+, HYDRA, TPTL, and CODEP in terms of F-measure. In terms of AUC, the TCA+ and HYDRA models stand as strong competitors to the HIEL model.}\\
\textbf{Conclusion:} For better predictions, we recommend ensemble learning approaches for the CPDP models. And, to estimate the benefits from the CPDP models, we recommend the above {project-specific} performance measures.
\end{abstract}
\begin{keywords}
Cross-Project Defect Prediction \sep Ensemble Learning \sep Machine Learning \sep Prediction Analysis \sep Probabilistic Weighted Majority Voting
\end{keywords}

\maketitle

\section{Introduction}
\label{introduction}
Building reliable software is the ultimate goal for the project management team \cite{lyu1996handbook}. Due to the enormous increase in building large-scale complex software systems, with the limited available resources such as manpower and cost, achieving the project goals within the deadlines is a challenging task \cite{Song2011, Wong2016}. In addition to that, the test team has to put more effort {into} achieving the correct behaviour of the developed software by removing/modifying the defective modules. To mitigate such challenges, automation tools such as software defect prediction (SDP) models are being developed \cite{lessmann2008benchmarking, Monperrus:2018, wu2018cross}. The motivation behind building SDP models is that it reduces the workload on the test team by reducing the testing time of software modules and, consequently, it reduces the total project's budget \cite{Fenton1999, Challagulla2005}.

Because of such advantages, research on building SDP models is attracted by both academia and {the} software industry. The main objective of SDP models is to predict the defect proneness of a module in newly developed software. These SDP models are {being built} based on machine learning (ML) algorithms that use metric suits as the independent variables (features). Later, using the trained SDP models, a newly developed module is classified into either \textit{defective} or \textit{clean} classes.

Recent studies have introduced two major categories of SDP, such as with-in project defect prediction (WPDP) and cross-project defect prediction (CPDP). The WPDP models utilise the module information from the released versions of a project to train the prediction model. The defect proneness of the modules in a currently developing version of that project is then found using the trained WPDP model \cite{turhan2009relative, bhutamapuram2021project, wu2018cross}. Lack of sufficient defect data is the major constraint in developing WPDP models \cite{zimmermann2009cross}. To overcome the problem of insufficient data, the CPDP models utilise the defect data of the various released projects (and {their} releases) to train the ML model.

Finding the defect proneness of the modules in a cross-project domain is a challenging task because, training a model on the developed project(s) may not be generalisable well on the target project \cite{xia2016hydra}. For this, many transfer learning approaches were proposed {in order} to provide solution by {transferring} the common knowledge from source project(s) to the target projects \cite{he2012investigation, ma2012transfer, nam2013transfer, zhang2015empirical, ni2017cluster}. In addition, ensemble learning models {have also been} proposed in the literature \cite{panichella2014cross, xia2016hydra} since these models provide solutions based on the group of classifiers rather than utilising single classifiers. However, Herbold et al \cite{herbold2017comparative} {claim that} none of the proposed approaches {in the literature consistently achieve significant results.}

To mitigate the risk of poor classification, in this work, we propose a hybrid-inducer ensemble learning (HIEL) technique based on the bootstrap aggregation methodology. This classification framework is constructed based on a diversity generation mechanism. The advantage of using {the} diversity generation mechanism is that it better captures the generalisable properties of the test instances across the different domains by relaying on the group of classifiers where, the solution is achieved based on the {average of} the predicted outcomes \cite{polikar2006ensemble}. Hence, instead of utilising one classifier, HIEL utilises a group of various inducers such as Logistic Regression (LR), Support Vector Machine (SVM), Decision Trees (DT), Na\"ve Bayes (NB), \textit{k}-Nearest Neighbours (\textit{k}-NN), and Artificial Neural Network (NN), during the training phase. Here, each inducer takes a random sample from the original training set, which is collected from the augmented source projects and builds the ensemble of classifiers. Later, the predictions of each classifier {are} fed as input to the combiner model called {the} probabilistic weighted majority voting (PWMV) approach to observe the final decision. The PWMV works on the basis of {a} weighted majority voting approach where, at each trial, the weights are updated based on the parameter $\beta$. {In this process, whenever} the expert (classifier) makes a mistake on the test example, gets penalised weight. Here, the weights are assigned to the experts as probabilities. Similarly, if the expert makes the correct decision, then that expert will be considered as the most probable expert in the near future. In this way, the PWMV predicts the class label for the test instance. The detailed working procedure of the proposed HIEL model using PWMV is given in Section \ref{HIEL}.

To find the superiority of the proposed HIEL approach, {this work includes a comparison of} recent implementations such as Transfer Component Analysis Plus (TCA+) \cite{nam2013transfer}, Training Data Selection (TDS) \cite{herbold2013training}, HYbrid moDel Reconstruction Approach (HYDRA) \cite{xia2016hydra}, COmbined DEfect Predictor (CODEP) \cite{panichella2014cross}, and Two-Phase Transfer Learning (TPTL) \cite{liu2019two} models. For the empirical analysis, we have used {38} releases of {12} software projects from the publicly available PROMISE repository \cite{promise2005}, {12 releases of 6 projects from the NASA MDP repository \cite{shepperd2013data}, and 5 projects from the AEEEM repository \cite{d2012evaluating}}. The results show that the proposed HIEL model with PWMV outperforms the above models on the majority of {the} target projects and is a strong competitor to the HYDRA model. The detailed discussion of the performance of the HIEL using PWMV is given in Section \ref{results}.

{In addition to} proposing a classification model, there is a need to analyse the impact of the prediction result that shows on the project attributes such as the amount of saved budget, the percent of remaining service time, and the percent of failures that may occur in a software with the introduction of the CPDP model in the real-time scenario. Making use of these measures is essential because this is the true goal of SDP research. To the best of our knowledge, this work is {intended} to provide such an analysis of the proposed CPDP model, as such analysis was not provided in the literature.

For this objective, we have proposed three prediction project-specific performance measures such as Percent of Perfect Cleans (PPC), Percent of Non-Perfect Cleans (PNPC), and {False Omission Rate (FOR)}. The {PPC} measure is used to analyse {the total savings in the total allocated} budget in the developing project, where as {PNPC} is used to estimate the {amount of time required to} service (based on the remaining number of code edits) {the code}. The {measure-}{FOR} calculates the percent of failures that may occur in the software{, after deploying the project}. The experimental results show that, {the projects from the repositories such as PROMISE, NASA, and AEEEM have 49.70\%, 68.70\%, and 58.98\% of budget savings, respectively. In this regard, the testers have to put their efforts on the code for only 50.30\%, 31.30\%, and 41.02\% of the total time to remove the total defect content (including the observed failures) in the projects of PROMISE, NASA, and AEEEM, respectively.} We hope this experiment will pave the path for future CPDP {(in general, the SDP)} research to more focus on developing the prediction models that provide analysis on the benefits of the project from the predictions.

With this work, we have made the following principal contributions to the cross-project defect prediction study:
\begin{enumerate}
\item Proposed a novel hybrid-inducer ensemble learning (HIEL) technique that uses bagging methodology to train the CPDP model. For this, a combiner approach called probabilistic weighted majority voting (PWMV) is employed to mitigate the risk of poor classification and to achieve the generalisable properties of defective instances that work across the target projects.
\item Provided a tight upper bound on the number of mistakes made by the best expert (classifier) in the PWMV combiner approach.
\item {Proposed three} project specific {measures} such as the amount of saved budget, amount of remaining service time for the tester, and the percent of failures that may occur in the project. To the best of our knowledge, estimating such project-specific measures from the predictions is new to this SDP (in any sub-category) task.
\end{enumerate}
\textbf{\textit{Paper Organization}}: Section \ref{RelatedWork} presents the literature review on classification models developed for CPDP. In section \ref{HIEL}, the proposed method for training is presented along with the PWMV combiner approach. The elements required to develop and evaluate the proposed model are given in Section \ref{empiSetup}. This section also defines the project-specific performance measures. The results of the proposed model and the detailed discussion are presented in section \ref{results}. Threats to the validity of the proposed work are presented in section \ref{threats}. Section \ref{conclusion} concludes the work and discusses the future work.
\section{Background and Related Work}
\label{RelatedWork}
Section \ref{CPDP_Review} reviews the classification models for CPDP. Section \ref{CPDP_Ensemble_Review} discusses the current status of the literature on using ensemble models for CPDP. Section \ref{DiversityGeneration} discusses the different strategies of generating diverse classifiers in the ensemble learning methodology. 
\subsection{Cross Project Defect Prediction Studies}
\label{CPDP_Review}
{Predicting the defect proneness of the source-code components is still evolving, and proposing a successful prediction model is still a challenging task \cite{herbold2017comparative}.}

As discussed in section \ref{introduction}, the CPDP models are majorly suffering from the data distribution among the source and target projects. As a consequence, these models are yielding poor predictive performance. To conquer the challenge of data distribution among source and target projects, many transfer learning models have been proposed \cite{ma2012transfer, nam2013transfer, herbold2013training, liu2019two} {in the literature}.

In \cite{ma2012transfer}, Ying Ma et al. proposed a Transfer Na\"ive Bayes (TNB) for the CPDP task. They exploited the sample weighting algorithm on the data across all the source projects to form a training dataset. {From the experimentation} on NASA and SOFTLAB datasets, they concluded that, when there is not enough training data to train a strong classifier, knowledge acquired from various source projects based on the feature-level information may help to obtain better performances. {For the same problem}, in \cite{nam2013transfer}, Nam et al. proposed a transfer component analysis plus (TCA+) model based on their prior work called transfer component analysis (TCA) \cite{pan2010domain}. The TCA tries to find the latent feature space between source and target projects by minimising the data distributions.  From the experimentation on PROMISE projects, they observed the improved accuracy {of} the target projects upon selecting a suitable source project.

On the similar problem, Herbold in \cite{herbold2013training} proposed two distance based similarity measures to select the training data across the source projects. To form the training dataset, \textit{k}-nearest neighbour algorithm was employed. From the experimental evaluation of TDS on the PROMISE projects, they concluded that improved performance can be obtained through selecting appropriate training data. {In \cite{xu2019cross}, Zhou Xu et al. proposed a method called balanced distribution adaptation (BDA), targeted to show the impact of conditional distribution differences in selecting the source and target data.}

Recently, Liu et al. in \cite{liu2019two} proposed a two-phase transfer learning (TPTL) model to overcome the problem that is {posed by the} TCA+ model. That is, the performance of TCA+ may vary on various target projects {based} on the selection of different source projects. For this, they proposed {the} TPTL model that works in two phases. In the first phase, a source project estimator (SPE) was proposed to select any two source projects for the target project where the two source projects form {the} highest distribution similarity. Later, TCA+ was employed {to build} two classification models on the selected source projects. {TPTL demonstrated improved performance on the PROMISE projects when compared to state-of-the-art models.}

{Selecting the appropriate simplified metric suits may also impact the performance of the defect prediction models. A detailed empirical analysis was presented in \cite{he2015empirical} by Peng He et al. that discussed the practical recommendations for selecting the training data, subset of metrics, and classification models. Hosseini et al. in \cite{hosseini2018benchmark}, also aimed to demonstrate the use of different feature sets for successful performances.}

However, achieving generalised predictive performance across the target projects is still a challenging task \cite{xia2016hydra, polikar2006ensemble}. This is because {the} use of specific classifiers is causing such low classification performances \cite{polikar2006ensemble}. To mitigate such challenges, recently, ensemble learning models {have been} proposed \cite{panichella2014cross, zhang2015empirical, laradji2015software, xia2016hydra} to overcome the poor predictive performances that these models are exhibiting on some target projects.
\subsection{Ensemble Learning Approaches for CPDP}
\label{CPDP_Ensemble_Review}
{In contrast to discussing the distributional characteristics of the data}, proposing ensemble learning models {has emerged as a tool} for CPDP in recent years. In \cite{panichella2014cross}, Panichella et al. proposed a CODEP (COmbined DEfect Predictor) approach that combines the outputs of six classifiers. Their research suggests using ensemble models for CPDP problems because many base-line classifiers perform poorly. Zhang et al. in \cite{zhang2015empirical} also suggested the benefits of utilising the ensemble frameworks for the CPDP problem. In \cite{laradji2015software}, Laradji et al. targeted to show the benefits of combining feature selection and ensemble learning. Their approach concluded that using fewer features with ensemble learning can help to mitigate issues like data imbalance and feature redundancy.

An extensive empirical study was conducted by Xin Xia et al in \cite{xia2016hydra} to address the CPDP problem. {Their approach (HYbrid moDel Reconstruction Approach (HYDRA)) is a combination of genetic algorithm and AdaBoost ensemble learning methodology, to create a bag of classifiers. This experiment makes use of an enormous number of classifiers in the process of obtaining the final decision. The HYDRA model consists of two steps: model building step and prediction step. Similar to our work, they have taken the common metrics from the source projects data and target projects data, to build the classification model. The model building step is used to generate the numerous classifiers. The model building step contains two phases such as genetic algorithm (GA) phase and ensemble learning (EL) phase. In the GA phase, for each source project, and training target data (which is the 5\% of the labeled target data), they build a classifier, and in total they build N+1 classifiers (assuming there are N number of source projects). Here, the extra one classifier is built on the training target data alone. Next, HYDRA uses GA to search for the best composition of these classifiers. In the EL phase, they built multiple GA classifiers, by running the GA phase multiple times (k times), and composing these GA classifiers according to their training error rate. In the prediction step, a new unlabeled instance from the target data will be given as input to the HYDRA model to observe the prediction.} This experiment indicates the use of an enormous number of classifiers in the process of obtaining a better final decision. 

However, a survey in \cite{herbold2017comparative} by Herbold et al. {indicates} the requirement of proposing a better classification model for this application. To alleviate unstable classification performances, Breiman and Polikar in \cite{breiman1996bagging} and \cite{polikar2006ensemble}, respectively, suggest {using an} ensemble of diverse classifiers in the decision making process. In addition, their study reveals several interesting conclusions, such as: 1) the ensemble of diverse classifiers can address the instability of the individual weak learners; 2) it can provide a generalisable solution for the classification task. Hence, our study {targeted to implement the} diversity generation mechanism as the classification technique to approach {the} CPDP problem.
\subsection{Diversity Generation Mechanism}
\label{DiversityGeneration}
Ensemble learning methods draw near-accurate final decision for a test example on the basis of utilising a huge number of diverse classifiers. According to the definition of ensemble learning, as the number of diverse classifiers increases infinitely, then the probability of correct decision on the test example approaches 1 \cite{breiman1996bagging, polikar2006ensemble}. Hence, to obtain the {most} accurate prediction, it is required to generate as many diverse classifiers as possible.

In order to increase the number of decision makers, different diversity generation strategies such as \textit{manipulating the training sample, feature subspace sampling, inclusion of hybrid inducer models, parameter optimization}, and \textit{changing the output representation} are utilised \cite{rokach2010pattern}. Note that, usage of such components of taxonomy may not be mutually exclusive. That is, according to the applicability, one or more strategies can be combined to generate the diverse classifiers.

By considering the training time of {an} ensemble classifier, we make use of two diverse generation strategies, such as sampling with replacement on the original training data and hybrid-inducer strategies, to generate the diverse classifiers. The strategies in the construction of the proposed ensemble learning methodology is given below.
\subsubsection{The Hybrid-Inducer System}
{In the hybrid inducer strategy, multiple base inducers are used to generate the diversity in the ensemble approach. Note that, a classification algorithm is also referred to as an \textit{inducer} and an instance of an inducer for a set of training data is called a \textit{classifier}. Ideally, this hybrid-inducer strategy would always perform as well as the best of its ingredients. The advantage of this strategy is that, it can solve the dilemma which arises from the famous ``\textit{no free lunch theorem}''. This theorem implies that a certain inducer is considered successful only insofar as its bias matches with the characteristics of the application domain \cite{Brazdil1994}. Therefore, for a given machine learning application, the practitioner needs to decide which inducer needs to be used. But the hybrid-inducer strategy obviates the need to try each one and simplifies the entire process. That is, using the hybrid-inducer strategy,} each inducer obtains a bias (either explicit or implicit) that causes it to prefer some generalisations over others. Hence, this hybrid inducer strategy always performs better among the other strategies \cite{rokach2010pattern}. To generate numerous diverse classifiers, in this work, we use six base inducers such as LR, \textit{k}-NN, NB, SVM, NN, and DT. The description of these inducers is discussed in section \ref{BenchmarkMLs}.
\subsubsection{The Bootstrapping}
This is the general approach to create a bag of classifiers. {Using this strategy, each inducer is trained using a different variation, subset, or sample of the original training set. This strategy is effective for the inducer, which has a relatively large variance-error. That is, a small variation in the training set may cause a major change in the performance of the utilised classifier. In this approach, the training samples are generated (from the original training set) majorly using either with replacement or with-out replacement strategies. If the training set contain massive data-points, training the classifier become a bottleneck. Hence, in that case, the common strategy is to partition the entire population (original training data) into disjoint sets and train a classifier on each partitioned data \cite{rokach2010pattern}. If the original training set contain a limited data-points, then the common approach is to sample the training set with the replacement. In this work, we sample the original set of population data (source projects data) with replacement as many source projects contain a few data-points. To ensure a sufficient number of training instances in each of the drawn sample, it is a common practice to set the size of the sample as the size of the original training set \cite{rokach2010pattern}.}

The ensemble of classifiers is generated based on the parameter called \textit{sample number}. {A classifier is generated for each sample} based on the utilised inducer. To select the value of the sample number, we followed the result of a study by Opitz in \cite{opitz1999popular}. The study suggests {using} 10 samples to generate diverse classifiers to achieve better predictions on the test data. {This sample number is treated as the stopping criterion in generating the diverse classifiers.}
\section{The Hybrid-Inducer Ensemble Learning (HIEL) Technique}
\label{HIEL}
{The first part of this section presents the proposed diversity generation mechanism, and the second part of this section presents the combiner scheme called probabilistic weighted majority voting (PWMV).}
\subsection{Training the HIEL Model}
\label{TrainingHIEL}
The training procedure of {the} HIEL model is given in {the} algorithm \ref{HIEL-Algorithm}. The prerequisites for this model are: training data (source project's information), test data (target project's data), inducers, and a constant to form the bootstrap samples. Before training the HIEL model, training data $\mathcal{S}$ is collected by augmenting all the source projects' (including {their} releases) data. Based on the assumption that each project is developed in similar environments, by matching the metrics in each project, we augment the module's information across all the source projects. Here, each metric is treated as an independent variable in training the model. Similarly, we prepare the test data (validation set) $\mathcal{V}$ from the target project to find the defect proneness of the newly developed modules. Later, to generate {an} ensemble of diverse classifiers, two diversity generation schemes, such as hybrid inducer system, and bootstrap aggregation, are used, and the description of these methods is given in the section \ref{DiversityGeneration}.

As discussed in section \ref{DiversityGeneration}, HIEL model employs sampling with replacement to collect the training data $\mathcal{S}_t^j$ of size \textit{k} from the original training set $\mathcal{S}$. Hence, after sampling with replacement, the module's information is included in the sample at-least zero times.

The algorithm \ref{HIEL-Algorithm} iterates over |$\mathcal{I}$| inducers and, for each inducer, \textit{T} classifiers are generated based on the bootstrap samples. That is, in $t^{th}$ iteration $t \in \{1,2 \dots, T\}$, a random sample $\mathcal{S}_t^j$ $j \in \{1,2 \dots, |\mathcal{I}|\}$ of size \textit{k}, \textit{k}=$|\mathcal{S}|$, is drawn with the replacement from $\mathcal{S}$. Now, each inducer $\mathcal{I}_j$ takes the drawn sample $\mathcal{S}_t^j$ and builds a classifier $\mathcal{C}_j[t]$ in $t^{th}$ iteration. Then, using $\mathcal{C}_j[t]$, predict the class labels for the test examples in a set $\mathcal{V}$ into \textit{defect} (1) or \textit{clean} (0), and store the predicted class labels in $\mathcal{P}_j[t]$. Since, the algorithm takes |$\mathcal{I}$| inducers and \textit{T} samples for each inducer, as a result, after completion of all the iterations, a total of |$\mathcal{I}|\times$\textit{T} classifiers were built. In the end, the predicted class labels $\mathcal{P}_j[t]$, are given as input to the probabilistic weighted majority voting combiner model to get the final prediction for the test examples $\mathcal{V}$. The figure \ref{flowchart-HIEL} represents the proposed training framework of the HIEL model and, the classification procedure using PWMV.
\begin{algorithm}
\SetAlgoLined
\SetKwInput{KwInput}{Input}                
\SetKwInput{KwOutput}{Output}              

\KwInput{\textit{Inducers} $\mathcal{I}$ = \{$\mathcal{I}_1$, $\mathcal{I}_2$, $\cdots$, $\mathcal{I}_i$\},\\
\textit{Training set size k =} $|\mathcal{S}|$,\\
\textit{Number of iterations} = \textit{T} = 10}
\KwOutput{Classifiers, $\mathcal{C}$ = \{$\mathcal{C}_j[t], j \in \{1, 2, \dots, i\}, t \in \{1,2, \dots, T\}$\}\\
Predictions, $\mathcal{P}$ = \{$\mathcal{P}_j[t], j \in \{1, 2, \dots, i\}, t \in \{1,2, \dots, T\}$\}}
\KwData{Population Data $\mathcal{P}d$ = \textit{Training Set} ($\mathcal{S}$) $\cup$ \textit{Validation Set} ($\mathcal{V}$),\\
\textit{Training Set} $\mathcal{S}$ = \{Defects data collected from source projects\}, \\
\textit{Validation Set} $\mathcal{V}$ = \{Defects data of the target project\} }
\SetKwFunction{FMain}{HIEL}
\SetKwProg{Fn}{Function}{:}{\KwRet $\mathcal{P}$}
 \Fn{\FMain{$\mathcal{P}d$, $\mathcal{S}$, $\mathcal{V}$, $\mathcal{I}$, T, \textit{k}}}{
 j = 1\;
 \While{j $\leq|\mathcal{I}|$}{ 
 	t = 1\;
 	\While{t $ \leq $ N}{
 	  $\mathcal{S}_t^j$ = Sample $\mathcal{S}$ with \textit{k} observations randomly with replacement\;
	  $\mathcal{Y}_t^j$ = Class labels of test set $\mathcal{V}$\;
    	  $\mathcal{C}_j[t]$ = Build classifier using $\mathcal{I}_j$ on $\mathcal{S}_t^j$\;
    	  $\mathcal{P}_j[t]$ = Predict the class labels into \textit{defect} or \textit{clean} for test set $\mathcal{V}$ using the classifier $\mathcal{C}_j[t]$\;
	  t++\;
	 }
   j++\;
   }
 }
\caption{HIEL-Diversity Generation Phase}
\label{HIEL-Algorithm}
\end{algorithm}
\begin{figure*}[ht]
\centering
\includegraphics[height=9cm, width = \textwidth]{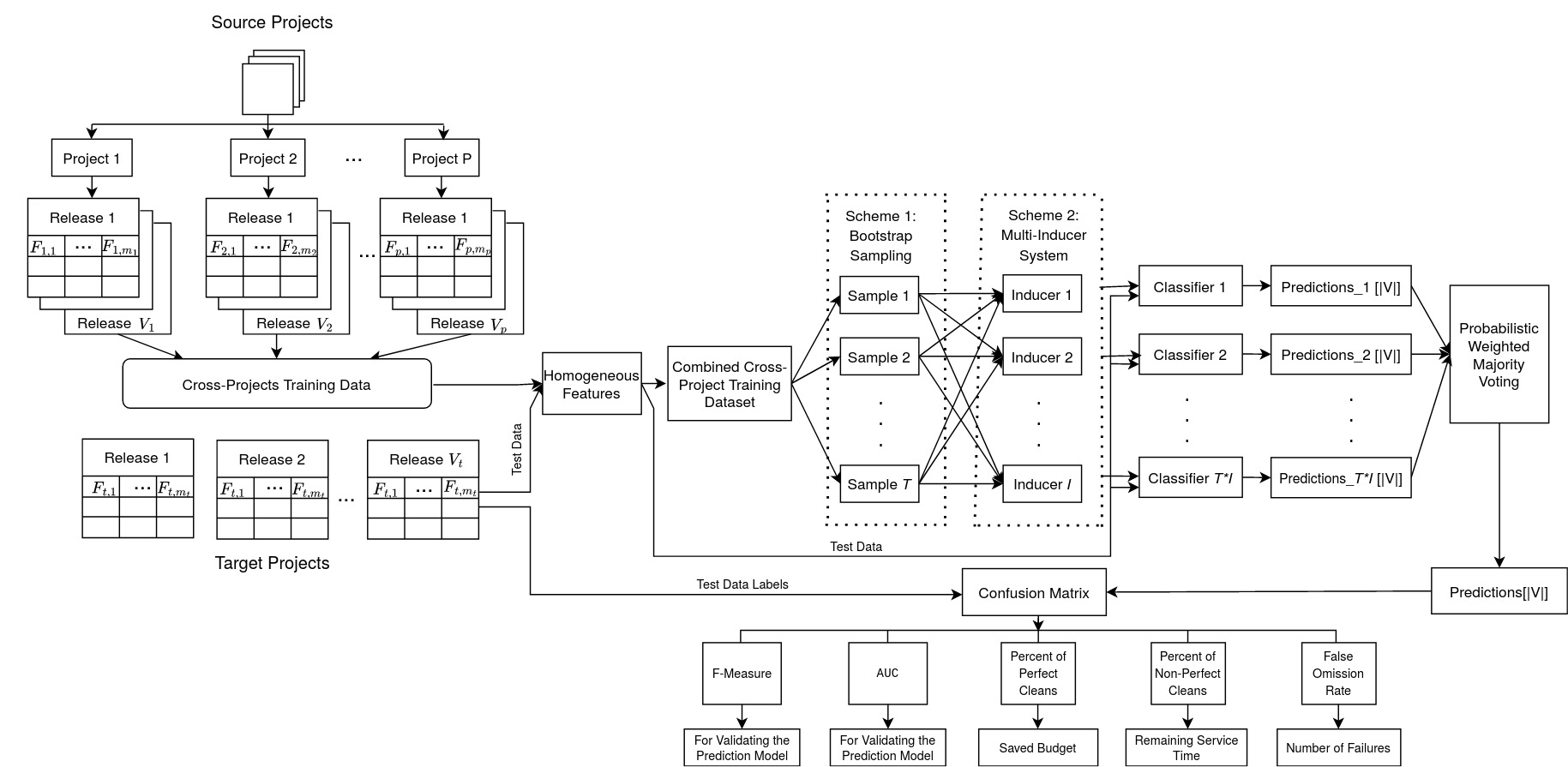}
\caption{Work flow of the classification procedure using proposed HIEL method}
\label{flowchart-HIEL}
\end{figure*}
\subsection{Probabilistic Weighted Majority Voting for Classification}
\label{pwmv}
After obtaining the predictions from the HIEL model on the test data, the combiner method called probabilistic weighted majority voting (PWMV) is employed to predict the final class label for the test instance. The intuition {behind} the PWMV combiner method is developed based on the weighted majority voting (WMV) strategy \cite{polikar2006ensemble}. {The weighted majority voting algorithm maintains a list of weights (let us assume the weights as, $w_1, w_2, \cdots, w_n$), for a set of $n$ experts, and computes the final prediction based on the weighted majority of the expert opinions. Without loss of generality, where each single expert is assumed as a trained classifier. Now, in WMV, uniform weights are distributed among all the experts. For simplicity, let the weight given to all the experts be considered {to be} 1. Given a sequence of experts (classifiers), at each test example, if the \textit{expert} makes the wrong prediction, gets the punishment of half of its weight. Now, to classify the text example, add the weights of all the experts in the two categories (\textit{defective} (1) or \textit{clean} (0)), and compare the resulting weights of these two categories; then, classify the test example to the class with the highest weight.}

{According to Blum \cite{blum1998line},} the regular WMV method guarantees an upper bound of $2.14(\log_2 n+\epsilon)$ mistakes during several trials. Where, $\epsilon$ is the number of mistakes made by the best expert, so far, and \textit{n} is the total number of experts involved in the final decision. This (the upper bound, $2.14(\log_2 n+\epsilon)$) is still problematic because, the best expert still makes the mistakes of more than 20\% of the total trials. As this (depending on $\epsilon$) is the limitation on achieving better results, in this case, a probabilistic version of the weighted majority voting is used to reduce the upper bound on the total mistakes made by the best expert \cite{blum1998line}. Here, in probabilistic weighted majority voting (PWMV), the weights of individual experts are updated using the parameter $\beta, 0 < \beta < 1$. Where $\beta$ is the amount of penalty that is given to each mistaken expert in the process of making the decision. The intuition behind the PWMV algorithm is that this dilutes (reduces) the worst possible mistakes made by the best expert. This is because, instead of considering the weights as probabilities, each outcome for the test example is predicted with a probability proportionate to its weight. That is, as the weight of the expert reduces, then the probability of the expert predicting the class label also reduces. The general algorithmic procedure to classify new instances using PWMV is given in the algorithm \ref{rwmvalgorithm}.

\begin{algorithm}
\SetAlgoLined
 Initialize the weights of all the experts $w = \{w_{11}, \cdots, w_{1T}, w_{21}, \cdots, w_{2T}, {\cdots,} w_{|\mathcal{I}|1}, \cdots, w_{|\mathcal{I}|T}\}$  to 1\;
 Observe the predictions from all the experts $\hat{y} = \{\hat{y}_{11}, \hat{y}_{12}, \cdots, \hat{y}_{|\mathcal{I}|T}\}$ for the test example $x_i$\;
 Initialize $\beta$\;
 Calculate \textit{W} = $\sum_{i=1}^{|\mathcal{I}|} \sum_{j=1}^{T} w_{ij}$\;
 \While{Test Set is not Empty}{
  \For{All the experts $\{c_1, c_2, \cdots, c_{|\mathcal{I}|T}{\}}$}{
  \eIf{The prediction $\hat{y_i}$ matches with the actual value,$y_i$}{
   Predict the class label for $x_i$ as $\hat{y_i}$ with the maximum probability $\frac{w_i}{W}$\;
   Do not modify the weight of the expert $c_i$\;
   }{
   Penalise mistaken expert by multiplying its weight with $\beta$. (That is $w_i = w_i*\beta$)\;
   Update the probability of the expert with $\frac{w_i}{W}$ \;
   }
  }
  Update the weight value W with the resulting weights ($W = \sum_i w_i$)\;
 }
 \caption{Probabilistic Weighted Majority Voting Algorithm}
 \label{rwmvalgorithm}
\end{algorithm}
Using the analysis on minimising the upper bound on the total mistakes made by the best expert \cite{blum1998line}, the following corollary defines the maximum upper bound on the total mistakes made by the best expert using PWMV.\\ 
\begin{corollary}
\label{corollary1}
Given the weights $w = \{w_{11}, w_{12}, \cdots, w_{1T},\\ w_{21},w_{22}, \cdots, w_{2T}, {\cdots,}  w_{|\mathcal{I}|1}, w_{|\mathcal{I}|2}, \cdots, w_{|\mathcal{I}|T}\}$ of experts and the set of predictions $\hat{y} = \{\hat{y}_{11}, \hat{y}_{12}, \cdots, \hat{y}_{|\mathcal{I}|T}\}$ on the test instance, ${x}_{l}$, the upper bound on the number of mistakes $\mathrm{M}$ made by the PWMV satisfies:\\
\[
\mathrm{M} \leq \frac{\log |\mathcal{I}| + \log T - \epsilon*\log \beta}{1-\beta}
\]
Where $|\mathcal{I}|$ is the expert index, $\beta$ is the penalty given to each mistaken expert and $\epsilon$ is the number of mistakes so far made by the best expert chosen for the HIEL model.
\end{corollary}
\begin{proof}
Assign the weights of each expert $w = \{w_{11}, w_{12}, \cdots,\\ w_{1T}, w_{21},w_{22}, \cdots, w_{2T}, {\cdots,}  w_{|\mathcal{I}|1}, w_{|\mathcal{I}|2}, \cdots, w_{|\mathcal{I}|T}\}$ to 1. Observe the predictions $\hat{y} = \{\hat{y}_{11}, \hat{y}_{12}, \cdots, \hat{y}_{|\mathcal{I}|T}\}$ from all the hybrid experts on ${x}_{l}$. From the weights \textit{w}, calculate;

\begin{equation}
W = \sum_{i=1}^{|\mathcal{I}|} \sum_{j=1}^{T} w_{ij}.
\end{equation}

Let $F_i$  is the total weight on the wrong answers {on} the $i^{th}$ trial. Suppose, \textit{t} trials have been experienced. And, let $\mathrm{M}$ be the expected number of mistakes made so far. It is defined as $\mathrm{M}$ = $\sum_{i=1}^{t} F_i$.

On the $i^{th}$ example, the updated weight $W^{'}$ is observed as:
\begin{equation}
W' = W*(1-(1-\beta)F_i)
\end{equation}

Where $\beta$ is the penalty given to the each mistaken expert. Now, similarly, for all the experts from HIEL, the updated final weight is observed as:
\begin{equation}
W' = |\mathcal{I}|*T*\prod_{i=1}^{t}(1-(1-\beta)F_i)
\end{equation}

Now, let $\epsilon $ be the number of mistakes made by the best expert so far. Again, using the fact that, $W^{'}$ is as large as the weight of the best expert, then we have:
\begin{equation}
|\mathcal{I}|*T*\prod_{i=1}^{t}(1-(1-\beta)F_i) \geq \beta^\epsilon
\end{equation}

Apply \textit{log} on both sides, then;
\begin{equation}
\log(|\mathcal{I}|*T) + \sum_{i=1}^{t} \log (1-(1-\beta)F_i) \geq \epsilon\log\beta
\end{equation}
\begin{equation}
\implies \epsilon\log(\frac{1}{\beta}) \geq - \log (|\mathcal{I}|*T) - \sum_{i=1}^{t} \log(1-(1-\beta)F_i)
\end{equation}

We know, $\log(1-x) = -x - \frac{x^2}{2}-\frac{x^3}{3}- \cdots$

$\therefore \log(1-x) \leq -x,  \forall x$.

Then,
\begin{equation}
\implies \epsilon\log(\frac{1}{\beta}) \geq - \log (|\mathcal{I}|*T) + (1-\beta)\sum_{i=1}^{t} F_i
\end{equation}
\begin{equation}
\implies \epsilon\log(\frac{1}{\beta}) \geq - (\log |\mathcal{I}| + \log T) + (1-\beta)*M
\end{equation}
where, $\mathrm{M}$ = $\sum_{i=1}^{t} F_i$.
\begin{equation}
\implies (1-\beta)*M \leq \epsilon\log(\frac{1}{\beta}) + \log |\mathcal{I}| + \log T
\end{equation}
\begin{equation}
\implies M \leq \frac{\epsilon\log(\frac{1}{\beta}) + \log |\mathcal{I}| + \log T}{(1-\beta)}
\end{equation}
This satisfies the above corollary.
\end{proof}
\section{Empirical Setup}
\label{empiSetup}
The description of the used defect datasets {is presented} in Section \ref{studiedDatasets}. The default parameter setting for the used base-line inducers {is provided in} Section \ref{BenchmarkMLs}. The model specific performance evaluation measures {are discussed} in Section \ref{evalMeasures}. In addition to that, section \ref{studentsTtest} discusses the {well known significance tests such as} \textit{one-sample Wilcoxon signed rank test} and \textit{Cliff’s delta effect size test}.
\subsection{Defect Datasets}
\label{studiedDatasets}
To construct and validate the prediction model, this work utilises publicly available defect datasets from {the repositories such as} PROMISE \cite{promise2005}, NASA Metrics Data Program (MDP) \cite{shepperd2013data} and AEEEM \cite{d2012evaluating}. For this empirical analysis, we have used a total of {38, 12 and 5 projects from the PROMISE, NASA, and AEEEM repositories, respectively}. Where, each PROMISE project is having 24 metrics (features) to represent the {software} module. {For the NASA projects, we have selected 21 common metrics since each project differs with respect to the variable number of metrics. Similarly, for the AEEEM projects, we have used 17 metrics to build the model. Amongst the available set of metrics, the variables that indicate the severity level of the software module are excluded from the AEEEM projects.} The description of each project (and its releases), such as the number of modules present in that released version, {the total lines of code in the project,} the number of defects and the percent of defects in the released version are presented in the table \ref{DefectsData}.

For the source projects, we have augmented the defect data of all the projects except {all releases of the target project. In this regard, each release of the target project is treated as the test data.}
\begin{table*}[]
\centering
\caption{An overview of utilised projects}
\label{DefectsData}
\begin{tabular}{lrrrr|lrrrr}
\hline
\textbf{Project} & \multicolumn{1}{c}{\textbf{Modules}} & \textbf{LoC} & \multicolumn{1}{c}{\textbf{Defects}} & \multicolumn{1}{c|}{\textbf{\%Defects}} & \textbf{Project} & \multicolumn{1}{c}{\textbf{Modules}} & \textbf{LoC} & \multicolumn{1}{c}{\textbf{Defects}} & \multicolumn{1}{c}{\textbf{\%Defects}} \\ \hline \hline
\multicolumn{10}{c}{\textbf{PROMISE Projects}} \\ \hline \hline
Ant-1.3 & 125 & 37,699 & 20 & 16.00 & Lucene-2.4 & 340 & 102,859 & 203 & 59.71 \\
Ant-1.4 & 178 & 54,195 & 40 & 22.47 & Poi-1.5 & 237 & 55,428 & 141 & 59.49 \\
Ant-1.5 & 293 & 87,047 & 32 & 10.92 & Poi-2.0 & 314 & 93,171 & 37 & 11.78 \\
Ant-1.6 & 351 & 113,246 & 92 & 26.21 & Poi-2.5 & 385 & 119,731 & 248 & 64.42 \\
Ant-1.7 & 745 & 208,653 & 166 & 22.28 & Poi-3.0 & 442 & 129,327 & 281 & 63.57 \\
{\color[HTML]{000000} Camel-1.0} & {\color[HTML]{000000} 339} & {\color[HTML]{000000} 33,721} & {\color[HTML]{000000} 13} & {\color[HTML]{000000} 03.83} & Redaktor & 176 & 59,280 & 27 & 15.35 \\
{\color[HTML]{000000} Camel-1.2} & {\color[HTML]{000000} 608} & {\color[HTML]{000000} 66,302} & {\color[HTML]{000000} 216} & {\color[HTML]{000000} 35.53} & Synapse-1.0 & 157 & 28,806 & 16 & 10.19 \\
{\color[HTML]{000000} Camel-1.4} & {\color[HTML]{000000} 872} & {\color[HTML]{000000} 98,080} & {\color[HTML]{000000} 145} & {\color[HTML]{000000} 16.63} & Synapse-1.1 & 222 & 42,302 & 60 & 27.03 \\
{\color[HTML]{000000} Camel-1.6} & {\color[HTML]{000000} 965} & {\color[HTML]{000000} 113,055} & {\color[HTML]{000000} 188} & {\color[HTML]{000000} 19.48} & Synapse-1.2 & 256 & 53,500 & 86 & 33.59 \\
{\color[HTML]{000000} Jedit-3.2} & {\color[HTML]{000000} 272} & {\color[HTML]{000000} 128,883} & {\color[HTML]{000000} 90} & {\color[HTML]{000000} 33.09} & Tomcat & 858 & 300,674 & 77 & 08.97 \\
{\color[HTML]{000000} Jedit-4.0} & {\color[HTML]{000000} 306} & {\color[HTML]{000000} 144,803} & {\color[HTML]{000000} 75} & {\color[HTML]{000000} 24.51} & Velocity-1.4 & 196 & 51,713 & 147 & 75.00 \\
{\color[HTML]{000000} Jedit-4.1} & {\color[HTML]{000000} 312} & {\color[HTML]{000000} 153,087} & {\color[HTML]{000000} 79} & {\color[HTML]{000000} 25.32} & Velocity-1.6 & 229 & 57,012 & 78 & 34.06 \\
{\color[HTML]{000000} Jedit-4.2} & {\color[HTML]{000000} 367} & {\color[HTML]{000000} 170,683} & {\color[HTML]{000000} 48} & {\color[HTML]{000000} 13.08} & Xalan-2.4 & 723 & 225,088 & 110 & 15.21 \\
{\color[HTML]{000000} Jedit-4.3} & {\color[HTML]{000000} 492} & {\color[HTML]{000000} 202,363} & {\color[HTML]{000000} 11} & {\color[HTML]{000000} 02.24} & Xalan-2.5 & 803 & 304,864 & 387 & 48.19 \\
Log4j-1.0 & 135 & 21,549 & 34 & 25.19 & Xalan-2.6 & 885 & 411,737 & 411 & 46.44 \\
Log4j-1.1 & 109 & 19,938 & 37 & 33.95 & Xalan-2.7 & 909 & 428,555 & 898 & 98.79 \\
Log4j-1.2 & 205 & 38,191 & 189 & 92.20 & Xerces-1.2 & 440 & 159,254 & 71 & 16.14 \\
\multicolumn{1}{l}{Lucene-2.0} & 195 & 50,596 & 91 & 46.67 & Xerces-1.3 & 453 & 167,095 & 69 & 15.23 \\
Lucene-2.2 & 247 & 63,571 & 144 & 58.3 & Xerces-1.4 & 588 & 141,180 & 437 & 74.32 \\ \hline \hline
\multicolumn{10}{c}{{\color[HTML]{000000} \textbf{NASA Projects}}} \\ \hline \hline
{\color[HTML]{000000} CM1} & {\color[HTML]{000000} 344} & {\color[HTML]{000000} 15,486} & {\color[HTML]{000000} 42} & {\color[HTML]{000000} 12.21} & {\color[HTML]{000000} MW1} & {\color[HTML]{000000} 264} & {\color[HTML]{000000} 6,905} & {\color[HTML]{000000} 27} & {\color[HTML]{000000} 10.23} \\
{\color[HTML]{000000} JM1} & {\color[HTML]{000000} 9,953} & {\color[HTML]{000000} 376,794} & {\color[HTML]{000000} 1,759} & {\color[HTML]{000000} 18.24} & {\color[HTML]{000000} PC1} & {\color[HTML]{000000} 759} & {\color[HTML]{000000} 23,020} & {\color[HTML]{000000} 61} & {\color[HTML]{000000} 08.04} \\
{\color[HTML]{000000} KC1} & {\color[HTML]{000000} 2,096} & {\color[HTML]{000000} 42,706} & {\color[HTML]{000000} 325} & {\color[HTML]{000000} 15.51} & {\color[HTML]{000000} PC2} & {\color[HTML]{000000} 1,585} & {\color[HTML]{000000} 17,834} & {\color[HTML]{000000} 16} & {\color[HTML]{000000} 01.01} \\
{\color[HTML]{000000} KC3} & {\color[HTML]{000000} 200} & {\color[HTML]{000000} 6,399} & {\color[HTML]{000000} 36} & {\color[HTML]{000000} 18.00} & {\color[HTML]{000000} PC3} & {\color[HTML]{000000} 1,125} & {\color[HTML]{000000} 33,016} & {\color[HTML]{000000} 140} & {\color[HTML]{000000} 12.44} \\
{\color[HTML]{000000} MC1} & {\color[HTML]{000000} 9,277} & {\color[HTML]{000000} 66,583} & {\color[HTML]{000000} 68} & {\color[HTML]{000000} 00.73} & {\color[HTML]{000000} PC4} & {\color[HTML]{000000} 1,399} & {\color[HTML]{000000} 30,055} & {\color[HTML]{000000} 178} & {\color[HTML]{000000} 12.72} \\
{\color[HTML]{000000} MC2} & {\color[HTML]{000000} 127} & {\color[HTML]{000000} 5,503} & {\color[HTML]{000000} 44} & {\color[HTML]{000000} 34.65} & {\color[HTML]{000000} PC5} & {\color[HTML]{000000} 17,186} & {\color[HTML]{000000} 161,695} & {\color[HTML]{000000} 516} & {\color[HTML]{000000} 03.01} \\ \hline \hline
\multicolumn{10}{c}{{\color[HTML]{000000} \textbf{AEEEM Projects}}} \\ \hline \hline
{\color[HTML]{000000} Eclipse} & {\color[HTML]{000000} 997} & {\color[HTML]{000000} 363,633} & {\color[HTML]{000000} 206} & {\color[HTML]{000000} 20.66} & {\color[HTML]{000000} Mylyn} & {\color[HTML]{000000} 1,862} & {\color[HTML]{000000} 135,334} & {\color[HTML]{000000} 245} & {\color[HTML]{000000} 13.16} \\
{\color[HTML]{000000} Equinox} & {\color[HTML]{000000} 324} & {\color[HTML]{000000} 40,250} & {\color[HTML]{000000} 129} & {\color[HTML]{000000} 39.82} & {\color[HTML]{000000} PDE} & {\color[HTML]{000000} 1,497} & {\color[HTML]{000000} 123,017} & {\color[HTML]{000000} 209} & {\color[HTML]{000000} 16.96} \\
{\color[HTML]{000000} Lucene} & {\color[HTML]{000000} 691} & {\color[HTML]{000000} 43,732} & {\color[HTML]{000000} 64} & {\color[HTML]{000000} 09.26} & {\color[HTML]{000000} } & {\color[HTML]{000000} } & {\color[HTML]{000000} } & {\color[HTML]{000000} } & {\color[HTML]{000000} } \\ \hline
\end{tabular}
\end{table*}

\subsection{Benchmark Machine Learning Classifiers}
\label{BenchmarkMLs}
The proposed HIEL model uses six base inducers such as logistic regression, \textit{k}-nearest neighbours, support vector machine, Na\"ive Bayes, neural networks, and decision trees in the training phase.

{In} the logistic regression (LR) model, a \textit{logit} model is used to form a relation between the metric suits and the defect attribute. In the case of \textit{k}-nearest neighbour (\textit{k}-NN) model, the value of \textit{k} is chosen based on 10-fold cross validation. As a result, after validating the model with different values of \textit{k}, we have taken the value of \textit{k} as 11. For the support vector machine (SVM) model, we have used a linear kernel to train the model. In the case of Na\"ive Bayes (NB), to avoid the zero-probability problem, we set the Laplace smoothing parameter ($alpha$) to 1 \cite{rish2001empirical}. For the Neural Networks (NN) model, we have used a 2-hidden layered resilient backpropagation algorithm with the weight backtracking mechanism. For the stopping criteria, a 0.5 threshold is used in the partial derivatives of the error function. Later, we have limited the maximum steps in the training model to the value of 1e+5. In the case of decision trees, a general implementation of a classification and regression tree model \cite{breiman1984classification} is employed in this approach.
\subsection{Evaluation Measures}
\label{evalMeasures}
\begin{table}
\begin{center}
\caption{The confusion matrix represents the actual and predicted defect labels of the target project modules}
\label{confusionMatrix}
\begin{tabular}{|c|c|c|c|}
\hline
\multicolumn{2}{|c|}{\multirow{2}{*}{}} & \multicolumn{2}{c|}{Actual values} \\ \cline{3-4} 
\multicolumn{2}{|c|}{} & Defective & Clean \\ \hline
\multirow{2}{*}{Predicted values} & Defective & TP & FP \\ \cline{2-4} 
 & Clean & FN & TN \\ \hline
\end{tabular}
\end{center}
\end{table}

Performance measures are the key attributes {for analysing} the benefits {of} the prediction model. {Because the CPDP (in general, the SDP) models are designed to optimise the testing resources, providing benefits from the predictions that are understandable to the project manager is essential and has been ignored in defect prediction studies.} For this, in this work, three prediction model specific performance measures such as {Percent of Perfect Cleans (PPC), Percent of Non-Perfect Cleans (PNPC)} and {False Omission Rate (FOR)} are introduced for the first time in the SDP (in this case, CPDP) {problem context}. The measures {PPC, PNPC}, and {FOR} are calculated to estimate the saved amount of {budget} from the total allocated budget, to estimate the remaining service time, and to estimate the percent of failures, respectively,in the developing project. In addition to the above measures, we have used {other performance measures such as {AUC} and {F-measure}} for the comparative analysis.

{Similar to the traditional measures, the proposed} performance measures are {also} calculated based on the confusion matrix. The confusion matrix {for this binary classification is defined based on the} number of actual and predicted values of the test-set instances. In table \ref{confusionMatrix}, {the number of} true positives (TP) represents the number of defective modules which are predicted {into its correct class}, {the number of} true negatives (TN) represents the number of clean modules which are predicted {into its correct class}, {the number of} false positives (FP) represents the number of clean modules which are predicted {as being from the defective class}, and the number of false negatives (FN) represents the number of defective modules which are predicted as {being from the clean class}. {The following subsections describe in-detail about all the performance measures.}
\subsubsection{{Percent of Perfect Cleans (PPC)}}
\label{PPC}
{The {Percent of Perfect Cleans (PPC)} measure helps in deriving the amount of saved budget in the target project.} The {PPC} is defined as the ratio of true negatives over the total number of test observations. This is given as:
\begin{equation}
\label{equ-ppc}
     \text{PPC} = \frac{\text{True Negatives}}{\text{Total Test Instances}} = {\frac{|TN|}{|n_t|}}
\end{equation}

{Where, $|n_t|$ and |\textit{TN}| denote the size of the test set and {the} size of the true negatives, respectively.} Note that, the false negatives {are also representative of} the predicted clean modules. {But, even if the defective module is predicted as clean, then during {the} operational phase, the end user may experience an inconvenience in the software. As a result, the testing team must look for any hidden defects in that software module.} \cite{lyu1996handbook}. That is, if the end-user triggers these unidentified defect{ive modules}, then these defects will be active ({such defects are} also called as \textit{dormant defects} or \textit{soft defects}) and produce an error; when the erroneous instructions affect the delivered service, we observe the failure in the system \cite{lyu1996handbook}. Hence, even though if the defective modules are wrongly recognised as clean, then repairing of such modules by the testers is {an inevitable} task.

{In the other terms, all the false negative instances have ground truth labels as defective. Therefore, irrespective of the decision from the model, such modules need to be inspected to remove such defects, even after deploying the project. Hence, equation \ref{equ-ppc} does not include false negative instances in the numerator.}

Now from equation \ref{equ-ppc}, we define the {consequent measure called} percent of saved {lines of code}. For this, we use the lines of code (\textit{LoC}) as {an additional} attribute {to the module in the true negatives. Making use of LoC as an additional attribute (to the respective true negative instance) is possible in the CPDP scenario because every test module has a known size metric. Hence, by using the information of the LoC of the true negative module, it is possible to derive the total saved LoCs in the project.} Using the equation \ref{equ-ppc}, the percent of saved lines of code is derived as:
\begin{equation}
\label{equ-savedLoC}
    \text{Percent of Saved LoC} =\frac{\mathlarger{\sum}_{ {i\in TN}}  {SL}(LoC_i)}{\mathlarger{\sum}_{ {i\in n_t}}  {SL}(LoC_i)}
\end{equation}

Where, {$LoC_i$ represents the total LoC in the module \textit{i} and, $SL(LoC_i)$ represents the saved lines of code from the module \textit{i}}. Now, let us assume that a unit amount of cost is required to test each line of code. That is, we assume the cost spent to test each line of the code is uniform. {Now, the equation \ref{equ-savedLoC} can be rewritten as:}
\begin{equation}
\label{equ-PercentsavedCost}
    \text{Percent of Saved Budget} = \frac{\mathlarger{\sum}_{ {i\in TN}}  {SB}(LoC_i)}{\mathlarger{\sum}_{ {i\in n_t}}  {SB}(LoC_i)}
\end{equation}

{Here, $SB(LoC_i)$ represents the budget savings from the module \textit{i}. Now, the numerator of} the equation \ref{equ-PercentsavedCost} {represents the} saved {budget in the project and, it is given} as:
\begin{equation}
\label{equ-savedCost}
    {\text{Saved Budget} = \mathlarger{\sum}_{i\in TN} SB(LoC_i)}
\end{equation}

{Since PPC is the base measure for the consequent measures such as PSB and saved budget, we are treating PPC as the main measure in this approach. However, the consequent measures provide more information on the predictions.}
\subsubsection{{Percent of Non-Perfect Cleans (PNPC)}}
\label{NPC}
To {estimate} the amount of work that still remains for the tester after the prediction, we {use the measure called} {Percent of Non-Perfect Cleans (PNPC)}. The {PNPC} is defined as the ratio of sum of true positives, false positives and false negatives over the total number of instances {tested}. This is {given} as:
\begin{equation}
\label{equ-pnpc}
   {PNPC} = {\frac{|n_t|-|TN|}{|n_t|}}
\end{equation}

{While PPC and its supplementary measures provide information about the savings in the total budget, the PNPC and its supplementary measures provide information about the pending work to remove the total defects in the target project. As discussed in section \ref{PPC}, since the false negative instances also require repair, we have included information about these false negatives in the numerator.}

{Now, to} estimate the amount of work still remaining for the tester after the prediction, we use {the} same LoC as the attribute where we assume {a} unit amount of service time is required to {inspect each} line of code {in such modules}. Using equation \ref{equ-pnpc}, the percent of remaining lines of code {(which require modifications)} is derived as:
\begin{equation}
\label{equ-remainingLoC}
    \text{Percent of Remaining LoCs} = \frac{\mathlarger{\sum}_{ {i\in n_t-TN}}  {RL}(LoC_i)}{\mathlarger{\sum}_{ {i\in n_t}}  {RL}(LoC_i)}
\end{equation}

{Where $RL(LoC_i)$ represents the remaining lines of code of the module \textit{i}, for which the tester has to inspect. Equation \ref{equ-remainingLoC} is also a representative of the percent of remaining edits because we assume a unit amount of time is required to inspect each line of code.  Hence, the equation \ref{equ-remainingLoC} can be rewritten as:}
\begin{equation}
\label{equ-percentRemainingEdits}
    \text{Percent of Remaining Edits} = \frac{\mathlarger{\sum}_{ {i\in n_t-TN}}  {RE}(LoC_i)}{\mathlarger{\sum}_{ {i\in n_t}}  {RE}(LoC_i)}
\end{equation}

{Where $RE(LoC_i)$ represents the edits remaining in the module \textit{i}. Now, the numerator of equation \ref{equ-percentRemainingEdits} represents the remaining edits in the project, and it is given as:}
\begin{equation}
\label{equ-remainingEdits}
    {\text{Remaining Edits} = \mathlarger{\sum}_{i\in n_t-TN} RE(LoC_i)}
\end{equation}

{The equation \ref{equ-remainingEdits} is also an called as remaining service time because, a unit time is required to inspect each line of code.} From equation \ref{equ-remainingEdits}, we calculate the supplemental value called the \textit{project hours}, {which estimates the total time required to service the remaining edits, in project hours}. Assume the test team can modify $\Delta$ lines of code for every {one} hour of time. Then, the number of project hours required to modify the modules is {calculated} as:
\begin{equation}
\label{ProjectHours}
    \text{Project Hours} = \Bigg(\frac{\mathlarger{\sum}_{{i\in n_t-TN}} {RE}(LoC_i)}{\Delta} \hspace{0.1cm}\Bigg) hours
\end{equation}

{Note that, without the size metric (LoC), it is not possible to calculate the percent of saved budget, remaining service time, and other related measures from the prediction model.}

{In addition to the project hours metric, we derive another supplemental measure called the editing rate. This measure calculates} the number of lines of code that need to be reviewed to observe one defect. This is calculated based on the ratio of remaining edits and the {total defects in the project}:
\begin{equation}
\label{equ-editing rate}
    \text{Editing Rate} = \frac{\mathlarger{\mathlarger{\sum}_{{i\in n_t-TN}} {RE}(LoC_i)}}{{|TP|+|FN|}}
\end{equation}

{Note that, it is obvious that, if the project management team does not use the SDP models, then the testers have to look into each software module to discover the defects \cite{pressman2005software}. The original editing rate (if the SDP is not in practical use) is given as:}
\begin{equation}
\label{equ-original editing rate}
    \text{{Original Editing Rate =}} \frac{\mathlarger{{\mathlarger{\sum}_{i\in n_t}} {LoC_i}}}{{|TP|+|FN|}}
\end{equation}

{The difference between the original editing rate and the editing rate is the decreased editing rate by the use of the CPDP (in general, the SDP) model.}
\subsubsection{{False Omission Rate (FOR)}}
\label{FOR}
{The} {false omission rate (FOR)} is the ratio of false negatives over the number of predicted cleans. This measure is used to estimate the percent of failures that may occur in the target software. The {FOR} is expressed as:
\begin{equation}
\label{equ-FOR}
   \text{False Omission Rate (FOR)} = \frac{{|}FN{|}}{{|}TN{|}+{|}FN{|}}
\end{equation}

Here, the number of false negatives provides information about the dormant defective modules. As discussed in section \ref{PPC}, upon triggering such defective modules in the {operational} phase {then the end-user} may experience failures {in the software} \cite{lyu1996handbook}. {Hence, minimising the occurrence of failures is also the main objective of the prediction model.} From \cite{lyu1996handbook}, we assume each defective module may cause {the} failure of the software system. If the tester believes this prediction model, then, in a newly developed software, we may expect the number of failures {to be} the ratio of {FOR}. {To make the system with fewer failures (in the ideal case, a failure free system), the defect prediction model should minimize the percentage of false alarms (FOR) in the system.}
\subsubsection{{F-Measure}}
The {F-measure} is {derived} from the harmonic mean of {precision} and {recall}. {F-measure} is used to know the number of instances that the prediction models {are classified} correctly (precisely) and to know the robustness of the classifier (it does not miss {a} significant number of test instances). This measure is calculated as:
\begin{equation}
    {F-measure} = 2.\frac{{precision*recall}}{{precision+recall}}
\end{equation}
{\subsubsection{{AUC}}
\label{AUC}
The Area Under Curve ({AUC}) is a probability curve for a binary classification problem. The {AUC} plots a curve between the true positive rate and the false positive rate at various threshold values, it essentially, in this case, separates the defective modules from the clean modules. Higher values of {AUC} indicate the success of the classifier in distinguishing between defective and clean instances. For instance, at {AUC} = 1, the classifier can successfully distinguish the two classes correctly and, at {AUC} = 0, the classifier predicts all the clean modules as defective and all the defective modules as clean. When the value of {AUC} is in-between 0.5 and 1, then there is a high chance that the classifier performs better on the test data.}
\subsection{Statistical Significance Test}
\label{studentsTtest}
To find whether the performance difference between the HIEL and the other models is statistically significant, two non-parametric significance tests, such as the \textit{one-sample Wilcoxon signed rank test} and \textit{Cliff's delta effect size} tests, were conducted. These significance {tests are conducted using the measures such as {AUC} and} {F1-score (F-measure)}. The \textit{one-sample Wilcoxon signed rank test} is utilised in this work alternative to the \textit{one sample t-test} because, the defect data may not be always be normally distributed. The null and alternative hypothesis for \textit{one-sample Wilcoxon signed rank test} is defined as:\\
$H_0$: The average {performance} of the other models is equal to the {performance} of the HIEL.\\
$H_1$: The average {performance} of the other models differs from the {performance} of the HIEL.\\
For this, the customary threshold value is taken as 0.05. The null hypothesis is rejected when the \textit{p}-value is less than 0.05 significance level \cite{demvsar2006statistical}.

To know the amount of difference between the HIEL and the other models, we use another non-parametric effect size measure called Cliff's delta. This method is considered for the additional analysis of the other hypothesis tests. This measure provides four levels of effectiveness of the HIEL model on the other models over the target projects. These levels are given in table \ref{CliffsDelta}. {The larger} value of Cliff's delta indicates the greater effect between the models.
\begin{table}[ht]
\centering
\caption{Cliff's delta effect size levels \cite{cliff1993dominance}}
\label{CliffsDelta}
\begin{tabular}{lcl}
\hline
S.No & |$\delta$| & Effectiveness Category \\ \hline \hline
1 & 0.000 $\leq |\delta| <$ 0.147 & Negligible \\
2 & 0.147 $\leq |\delta| <$ 0.330 & Small \\
3 & 0.330 $\leq |\delta| <$ 0.474 & Medium \\
4 & 0.474 $\leq |\delta| \leq$ 1.000 & Strong \\ \hline
\end{tabular}
\end{table}

\begin{table}[!t]
\centering
{
\caption{Utilised system characteristics}
\label{SystemCharacteristics}
\setlength{\tabcolsep}{1pt}
\begin{tabular}{p{20pt}|p{60pt}|p{160pt}}
\hline
\multicolumn{3}{c}{\textbf{Systems requirements}} \\ \hline \hline
\textbf{S.No} & \multicolumn{1}{c|}{\textbf{Resource}} & \multicolumn{1}{c}{\textbf{Description}} \\ \hline
1 & \texttt{RStudio-4.1.2} & Integrated Development Environment \\ \hline
2 & \href{https://github.com/ekamnit/CPDP-HIEL}{Replication} & Resource for the {experiments} \\ \hline \hline
\multicolumn{3}{c}{\textbf{\texttt{R} packages and its functions}} \\ \hline \hline
\textbf{S.No} & \multicolumn{1}{c|}{\textbf{Package}} & \multicolumn{1}{c}{\textbf{Function}} \\ \hline
1 &\texttt{stats}& \texttt{glm} (\textit{Logistic Regression})  \\ \hline
2 & \texttt{e1071} & \texttt{naiveBayes} (\textit{Na\"ive Bayes}) \\ \hline
3 & \texttt{e1071} & \texttt{svm} (\textit{Support Vector Machine}) \\ \hline
4 & \texttt{class} & \texttt{knn} (\textit{k-Nearest Neighbours})\\ \hline
5 & \texttt{rpart} & \texttt{rpart} (\textit{Recursive Partitioning and Regression Trees})\\ \hline
6 & \texttt{neuralnet} & \texttt{neuralnet} (\textit{Neural Networks}) \\ \hline
\end{tabular}%
}
\end{table}
\subsection{Developing Environment}
\label{developEnv}
To develop and validate the proposed model, we have used an open-source IDE \texttt{R-4.1.2}. The utilised packages and their corresponding functions are listed in table \ref{SystemCharacteristics}. {The description of the source code and the utilised datasets for the proposed approach is provided in the link} in table \ref{SystemCharacteristics}.
\section{Results and Discussion}
\label{results}
This section presents the results of the conducted experiments on the proposed HIEL model. The section \ref{Betas and Diverse Classifiers} presents the experimental evaluation to choose the value of $\beta$. Section \ref{Comparative Analysis Using TM} presents the comparative analysis of the proposed HIEL model with the other works such as TDS, TCA+, HYDRA, TPTL, and CODEP on the 10 target projects. In section \ref{Cost,Service Time and Failure Analysis}, we examined the amount of budget saved, the remaining service time, the percentage of failures, and other {supplementary measures} derived from the prediction results. The detailed analysis is given below.
\subsection{Experimentation on Choosing $\beta$}
\label{Betas and Diverse Classifiers}
\begin{figure*}[ht]
 \subfloat[Performances of PWMV on \textit{Ant-1.7}]{
  \label{MistakesFigure1}
	\begin{minipage}[c][1\width]{
	   0.3\textwidth}
	   \centering
	   \includegraphics[width=1.1\textwidth]{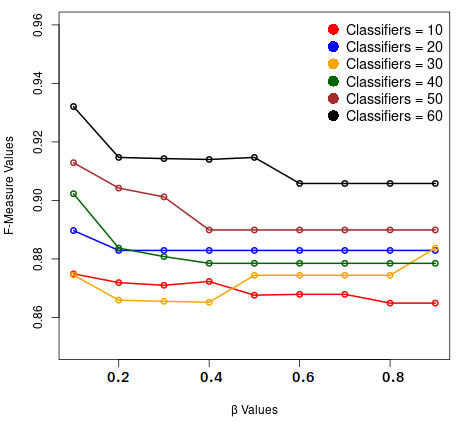}
	\end{minipage}}
 \hfill 	
 \subfloat[Performances of PWMV on \textit{JEdit-4.3}]{
  \label{MistakesFigure2}
	\begin{minipage}[c][1\width]{
	   0.3\textwidth}
	   \centering
	   \includegraphics[width=1.1\textwidth]{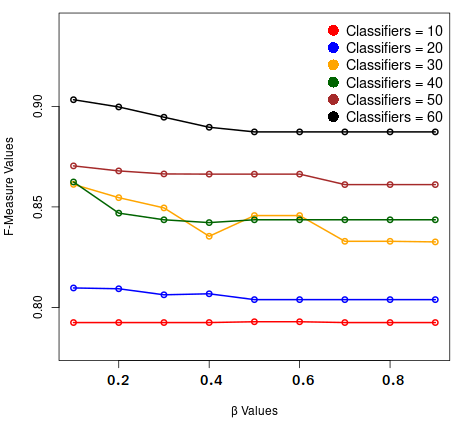}
	\end{minipage}}
 \hfill	
 \subfloat[Performances of PWMV on \textit{Redaktor}]{
  \label{MistakesFigure3}
	\begin{minipage}[c][1\width]{
	   0.3\textwidth}
	   \centering
	   \includegraphics[width=1.1\textwidth]{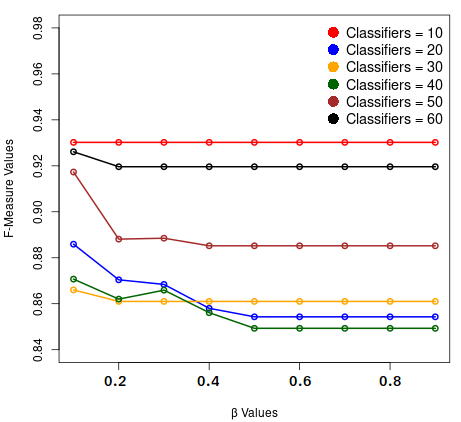}
	\end{minipage}}
 \hfill	
 \subfloat[Performances of PWMV on \textit{Synapse-1.0}]{
  \label{MistakesFigure4}
	\begin{minipage}[c][1\width]{
	   0.3\textwidth}
	   \centering
	   \includegraphics[width=1.1\textwidth]{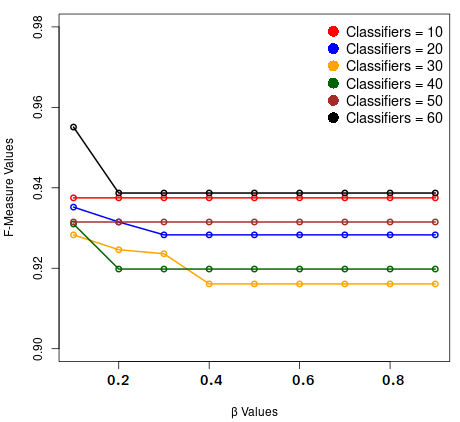}
	\end{minipage}}
 \hfill 	
 \subfloat[Performances of PWMV on \textit{Tomcat}]{
  \label{MistakesFigure5}
	\begin{minipage}[c][1\width]{
	   0.3\textwidth}
	   \centering
	   \includegraphics[width=1.1\textwidth]{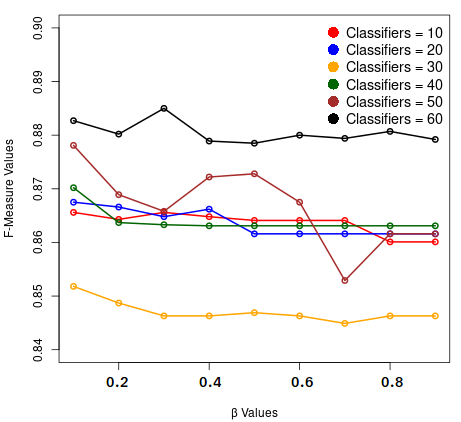}
	\end{minipage}}
 \hfill	
 \subfloat[Performances of PWMV on \textit{Velocity-1.6}]{
  \label{MistakesFigure6}
	\begin{minipage}[c][1\width]{
	   0.3\textwidth}
	   \centering
	   \includegraphics[width=1.1\textwidth]{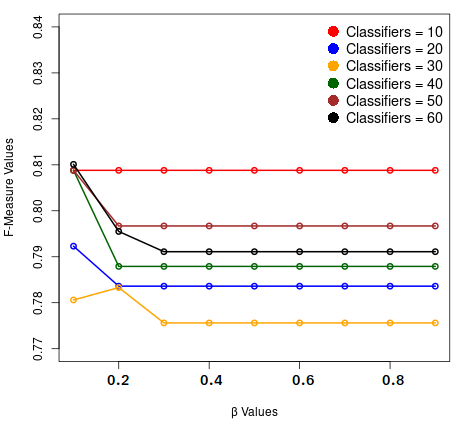}
	\end{minipage}}
 \hfill	
 \subfloat[Performances of PWMV on \textit{Xalan-2.4}]{
  \label{MistakesFigure7}
	\begin{minipage}[c][1\width]{
	   0.3\textwidth}
	   \centering
	   \includegraphics[width=1.1\textwidth]{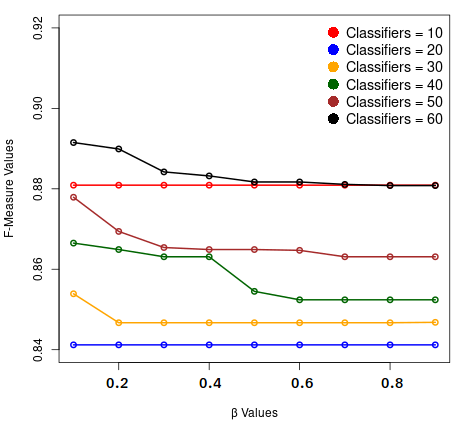}
	\end{minipage}}
 \hfill 	
 \subfloat[Performances of PWMV on \textit{Xerces-1.3}]{
  \label{MistakesFigure8}
	\begin{minipage}[c][1\width]{
	   0.3\textwidth}
	   \centering
	   \includegraphics[width=1.1\textwidth]{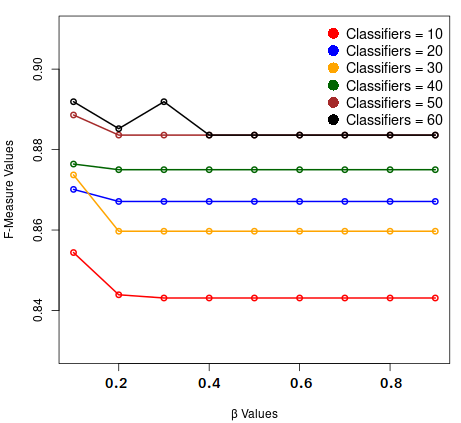}
	\end{minipage}}
 \hfill	
 \subfloat[Performances of PWMV on \textit{JM1}]{
  \label{MistakesFigure9}
	\begin{minipage}[c][1\width]{
	   0.3\textwidth}
	   \centering
	   \includegraphics[width=1.1\textwidth]{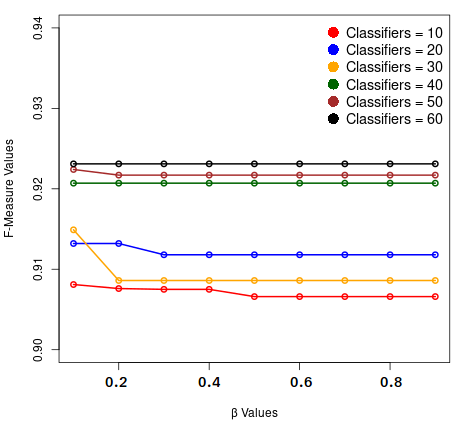}
	\end{minipage}}
 \hfill	
 \subfloat[Performances of PWMV on \textit{MW1}]{
  \label{MistakesFigure10}
	\begin{minipage}[c][1\width]{
	   0.3\textwidth}
	   \centering
	   \includegraphics[width=1.1\textwidth]{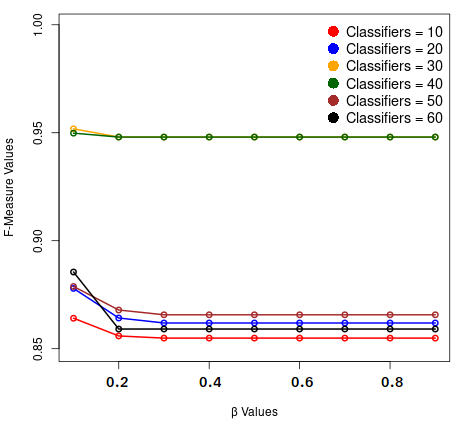}
	\end{minipage}}
 \hfill 	
 \subfloat[Performances of PWMV on \textit{Eclipse}]{
  \label{MistakesFigure11}
	\begin{minipage}[c][1\width]{
	   0.3\textwidth}
	   \centering
	   \includegraphics[width=1.1\textwidth]{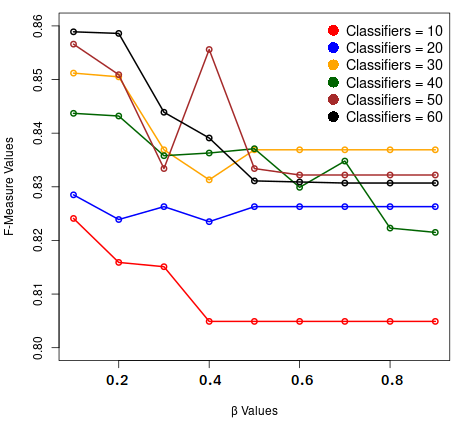}
	\end{minipage}}
 \hfill	
 \subfloat[Performances of PWMV on \textit{Mylyn}]{
  \label{MistakesFigure12}
	\begin{minipage}[c][1\width]{
	   0.3\textwidth}
	   \centering
	   \includegraphics[width=1.1\textwidth]{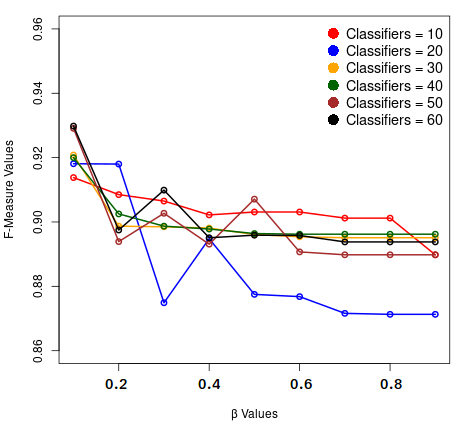}
	\end{minipage}}
\caption{Variation in the {F-measure} values after experimenting  PWMV at different values of $\beta$ and at using different number of diverse classifiers on the target projects.}
\label{BetasAndDiverseClassifiers}
\end{figure*}
As discussed in corollary \ref{corollary1}, choosing the right value of $\beta$ in PWMV helps to improve the final performance of the HIEL model. Hence, to select the value of $\beta$, we have conducted an experiment on different sets of diverse classifiers based on {F-measure} on each target project. {Nonetheless, the experimentation can also be conducted using other metrics.} We then utilised different sets of diverse classifiers to provide generalizability in showing the performances. Since each target project has a different distribution of the data, we have conducted this experiment on each new version of the target project to select the value of $\beta$. {Figure \ref{BetasAndDiverseClassifiers} represents the change in the performances of HIEL when utilising different values of $\beta$ at different sets of utilised diverse classifiers on the randomly selected target projects from the repositories PROMISE, NASA, and AEEEM, respectively.}

It is observed from the Figure \ref{BetasAndDiverseClassifiers} that the performance of the HIEL model is high on the majority of the target projects at the $\beta$ value of {0.1}, except on the projects {\textit{Ant-1.7}}, \textit{Tomcat}, and {\textit{Velocity-1.6}}. On the projects {\textit{Ant-1.7}}, \textit{Tomcat}, and {\textit{Velocity-1.6}}, the performance of the HIEL using PWMV is high at values of {0.9, 0.3, and 0.2}, respectively. {These performances are observed when utilising 30, 60, and 30 diverse classifiers on the above projects: \textit{Ant-1.7}, \textit{Tomcat} and \textit{Velocity-1.6}}, respectively. But on the target projects  {\textit{Tomcat} and \textit{Velocity-1.6}}, the performance of HIEL at {$\beta=0.1$} is just next to the performance of HIEL at {$\beta=0.3$} and at {$\beta=0.2$} respectively. {In the project \textit{Ant-1.7}, the performance of HIEL at $\beta=0.1$ stands in sixth place when we order the values of {F-measure} in decreasing order.}

Precisely, when we consider each project, for example, in the case of \textit{Ant-1.7}, out of six sets of diverse classifiers, the {F-measure} values {are high} at the {$\beta=0.1$}, except when HIEL is tested for the set of {30} diverse classifiers. In the cases of {\textit{JEdit-4.3, Redaktor, Synapse-1.0, Tomcat, Xalan-2.4, Xerces-1.3, JM1, MW1, Eclipse}, and \textit{Mlyn}}, the {F-measure} values {are high at} $\beta=0.1$ when using all sets of diverse classifiers. In the case of {\textit{Velocity-1.6}}, the {F-measure} values of HIEL are incremented at $\beta=0.2$ when using the set of {30} diverse classifiers. On the other set of diverse classifiers {(the case of \textit{Velocity-1.6})}, the {F-measure} values of HIEL were recorded {high} at $\beta=0.1$.

From figures \ref{BetasAndDiverseClassifiers} and \ref{MistakesFigure}, {it is observed that}, on {the} majority of the projects, the performances of HIEL using PWMV are recorded high at $\beta = 0.1, 0.2,$ and $0.3$, {when utilising various sets of diverse classifiers}. For the remaining values of $\beta$, the performance of HIEL {gets} decremented except in very few cases. Amongst the values of $\beta = 0.1, 0.2,$ and $0.3$, in the majority cases (on majority projects), the HIEL recorded its high {F-measure} values when using the value of $\beta=$ {0.1}. {Because, in the majority of cases, the {F-measure} values are high at $\beta=0.1$, in section \ref{Comparative Analysis Using TM}, the comparative analysis is conducted using the proposed classification framework at $\beta=0.1$.}
\subsection{Comparative Analysis {Using Traditional Measures}}
\label{Comparative Analysis Using TM}

\begin{table*}[ht]
\centering
\caption{Comparison of the proposed HIEL model with the other published models in terms of {F-measure} on PROMISE projects. {The W/T/L indicates the proposed model won, tied, or lost to the other models}}
\label{FMeasure-PROMISE}
\begin{tabular}{clcccccc}
\hline
{\color[HTML]{000000} \textbf{S.No}} & {\color[HTML]{000000} \textbf{Project}} & {\color[HTML]{000000} \textbf{CODEP}} & {\color[HTML]{000000} \textbf{TCA+}} & {\color[HTML]{000000} \textbf{HYDRA}} & {\color[HTML]{000000} \textbf{TPTL}} & {\color[HTML]{000000} \textbf{TDS}} & {\color[HTML]{000000} \textbf{HIEL}} \\ \hline
\hline
{\color[HTML]{000000} \textbf{1}} & {\color[HTML]{000000} \textbf{Ant-1.3}} & {\color[HTML]{000000} 0.2857} & {\color[HTML]{000000} 0.3689} & {\color[HTML]{000000} 0.468} & {\color[HTML]{000000} 0.456} & {\color[HTML]{000000} 0.2857} & {\color[HTML]{000000} \textbf{0.9321}} \\
{\color[HTML]{000000} \textbf{2}} & {\color[HTML]{000000} \textbf{Ant-1.4}} & {\color[HTML]{000000} 0.2222} & {\color[HTML]{000000} 0.3816} & {\color[HTML]{000000} 0.914} & {\color[HTML]{000000} 0.377} & {\color[HTML]{000000} 0.2476} & {\color[HTML]{000000} \textbf{0.8704}} \\
{\color[HTML]{000000} \textbf{3}} & {\color[HTML]{000000} \textbf{Ant-1.5}} & {\color[HTML]{000000} 0.2970} & {\color[HTML]{000000} 0.2956} & {\color[HTML]{000000} 0.691} & {\color[HTML]{000000} 0.237} & {\color[HTML]{000000} 0.2319} & {\color[HTML]{000000} \textbf{0.9149}} \\
{\color[HTML]{000000} \textbf{4}} & {\color[HTML]{000000} \textbf{Ant-1.6}} & {\color[HTML]{000000} 0.5161} & {\color[HTML]{000000} 0.5419} & {\color[HTML]{000000} 0.807} & {\color[HTML]{000000} 0.595} & {\color[HTML]{000000} 0.4153} & {\color[HTML]{000000} \textbf{0.8313}} \\
{\color[HTML]{000000} \textbf{5}} & {\color[HTML]{000000} \textbf{Ant-1.7}} & {\color[HTML]{000000} 0.4805} & {\color[HTML]{000000} 0.4896} & {\color[HTML]{000000} 0.295} & {\color[HTML]{000000} 0.455} & {\color[HTML]{000000} 0.3750} & {\color[HTML]{000000} \textbf{0.8489}} \\
{\color[HTML]{000000} \textbf{6}} & {\color[HTML]{000000} \textbf{Camel-1.0}} & {\color[HTML]{000000} 0.0870} & {\color[HTML]{000000} 0.1194} & {\color[HTML]{000000} 0.529} & {\color[HTML]{000000} 0.093} & {\color[HTML]{000000} 0.0541} & {\color[HTML]{000000} \textbf{0.9755}} \\
{\color[HTML]{000000} \textbf{7}} & {\color[HTML]{000000} \textbf{Camel-1.2}} & {\color[HTML]{000000} 0.1978} & {\color[HTML]{000000} 0.4189} & {\color[HTML]{000000} 0.19} & {\color[HTML]{000000} 0.502} & {\color[HTML]{000000} 0.4828} & {\color[HTML]{000000} \textbf{0.7837}} \\
{\color[HTML]{000000} \textbf{8}} & {\color[HTML]{000000} \textbf{Camel-1.4}} & {\color[HTML]{000000} 0.2083} & {\color[HTML]{000000} 0.3750} & {\color[HTML]{000000} 0.503} & {\color[HTML]{000000} 0.339} & {\color[HTML]{000000} 0.2747} & {\color[HTML]{000000} \textbf{0.9127}} \\
{\color[HTML]{000000} \textbf{9}} & {\color[HTML]{000000} \textbf{Camel-1.6}} & {\color[HTML]{000000} 0.1965} & {\color[HTML]{000000} 0.3514} & {\color[HTML]{000000} \textbf{0.991}} & {\color[HTML]{000000} 0.356} & {\color[HTML]{000000} 0.2436} & {\color[HTML]{000000} 0.8884} \\
{\color[HTML]{000000} \textbf{10}} & {\color[HTML]{000000} \textbf{Jedit-3.2}} & {\color[HTML]{000000} 0.3862} & {\color[HTML]{000000} 0.6324} & {\color[HTML]{000000} \textbf{0.903}} & {\color[HTML]{000000} 0.536} & {\color[HTML]{000000} 0.4828} & {\color[HTML]{000000} 0.7407} \\
{\color[HTML]{000000} \textbf{11}} & {\color[HTML]{000000} \textbf{Jedit-4.0}} & {\color[HTML]{000000} 0.4255} & {\color[HTML]{000000} 0.5194} & {\color[HTML]{000000} 0.6291} & {\color[HTML]{000000} 0.447} & {\color[HTML]{000000} 0.5054} & {\color[HTML]{000000} \textbf{0.8412}} \\
{\color[HTML]{000000} \textbf{12}} & {\color[HTML]{000000} \textbf{Jedit-4.1}} & {\color[HTML]{000000} 0.4968} & {\color[HTML]{000000} 0.5333} & {\color[HTML]{000000} 0.551} & {\color[HTML]{000000} 0.522} & {\color[HTML]{000000} 0.4205} & {\color[HTML]{000000} \textbf{0.7967}} \\
{\color[HTML]{000000} \textbf{13}} & {\color[HTML]{000000} \textbf{Jedit-4.2}} & {\color[HTML]{000000} 0.4295} & {\color[HTML]{000000} 0.3321} & {\color[HTML]{000000} 0.492} & {\color[HTML]{000000} 0.37} & {\color[HTML]{000000} 0.2172} & {\color[HTML]{000000} \textbf{0.8402}} \\
{\color[HTML]{000000} \textbf{14}} & {\color[HTML]{000000} \textbf{Jedit-4.3}} & {\color[HTML]{000000} 0.0508} & {\color[HTML]{000000} 0.0461} & {\color[HTML]{000000} 0.324} & {\color[HTML]{000000} 0.05} & {\color[HTML]{000000} 0.0494} & {\color[HTML]{000000} \textbf{0.9033}} \\
{\color[HTML]{000000} \textbf{15}} & {\color[HTML]{000000} \textbf{Log4j-1.0}} & {\color[HTML]{000000} 0.2273} & {\color[HTML]{000000} 0.5000} & {\color[HTML]{000000} 0.413} & {\color[HTML]{000000} 0.637} & {\color[HTML]{000000} 0.4356} & {\color[HTML]{000000} \textbf{0.8793}} \\
{\color[HTML]{000000} \textbf{16}} & {\color[HTML]{000000} \textbf{Log4j-1.1}} & {\color[HTML]{000000} 0.2979} & {\color[HTML]{000000} 0.6957} & {\color[HTML]{000000} 0.538} & {\color[HTML]{000000} 0.699} & {\color[HTML]{000000} 0.4146} & {\color[HTML]{000000} \textbf{0.8421}} \\
{\color[HTML]{000000} \textbf{17}} & {\color[HTML]{000000} \textbf{Log4j-1.2}} & {\color[HTML]{000000} 0.2569} & {\color[HTML]{000000} 0.6529} & {\color[HTML]{000000} \textbf{0.914}} & {\color[HTML]{000000} 0.606} & {\color[HTML]{000000} 0.5387} & {\color[HTML]{000000} 0.1436} \\
{\color[HTML]{000000} \textbf{18}} & {\color[HTML]{000000} \textbf{Lucene-2.0}} & {\color[HTML]{000000} 0.2569} & {\color[HTML]{000000} \textbf{0.6593}} & {\color[HTML]{000000} 0.648} & {\color[HTML]{000000} 0.569} & {\color[HTML]{000000} 0.6267} & {\color[HTML]{000000} 0.7259} \\
{\color[HTML]{000000} \textbf{19}} & {\color[HTML]{000000} \textbf{Lucene-2.2}} & {\color[HTML]{000000} 0.1707} & {\color[HTML]{000000} 0.6202} & {\color[HTML]{000000} \textbf{0.657}} & {\color[HTML]{000000} 0.515} & {\color[HTML]{000000} 0.5447} & {\color[HTML]{000000} 0.5925} \\
{\color[HTML]{000000} \textbf{20}} & {\color[HTML]{000000} \textbf{Lucene-2.4}} & {\color[HTML]{000000} 0.2823} & {\color[HTML]{000000} 0.6593} & {\color[HTML]{000000} \textbf{0.691}} & {\color[HTML]{000000} 0.458} & {\color[HTML]{000000} 0.3282} & {\color[HTML]{000000} 0.6138} \\
{\color[HTML]{000000} \textbf{21}} & {\color[HTML]{000000} \textbf{Poi-1.5}} & {\color[HTML]{000000} 0.2395} & {\color[HTML]{000000} 0.7390} & {\color[HTML]{000000} \textbf{0.742}} & {\color[HTML]{000000} 0.713} & {\color[HTML]{000000} 0.4259} & {\color[HTML]{000000} 0.5895} \\
{\color[HTML]{000000} \textbf{22}} & {\color[HTML]{000000} \textbf{Poi-2.0}} & {\color[HTML]{000000} 0.3077} & {\color[HTML]{000000} 0.2439} & {\color[HTML]{000000} 0.283} & {\color[HTML]{000000} 0.218} & {\color[HTML]{000000} 0.2446} & {\color[HTML]{000000} \textbf{0.9029}} \\
{\color[HTML]{000000} \textbf{23}} & {\color[HTML]{000000} \textbf{Poi-2.5}} & {\color[HTML]{000000} 0.2838} & {\color[HTML]{000000} 0.7723} & {\color[HTML]{000000} \textbf{0.78}} & {\color[HTML]{000000} 0.728} & {\color[HTML]{000000} 0.7261} & {\color[HTML]{000000} 0.5443} \\
{\color[HTML]{000000} \textbf{24}} & {\color[HTML]{000000} \textbf{Poi-3.0}} & {\color[HTML]{000000} 0.3718} & {\color[HTML]{000000} \textbf{0.8303}} & {\color[HTML]{000000} 0.807} & {\color[HTML]{000000} 0.787} & {\color[HTML]{000000} 0.7159} & {\color[HTML]{000000} 0.5537} \\
{\color[HTML]{000000} \textbf{25}} & {\color[HTML]{000000} \textbf{Redaktor}} & {\color[HTML]{000000} 0.2273} & {\color[HTML]{000000} 0.2344} & {\color[HTML]{000000} 0.295} & {\color[HTML]{000000} 0.353} & {\color[HTML]{000000} 0.2381} & {\color[HTML]{000000} \textbf{0.9172}} \\
{\color[HTML]{000000} \textbf{26}} & {\color[HTML]{000000} \textbf{Synapse-1.0}} & {\color[HTML]{000000} 0.4865} & {\color[HTML]{000000} 0.2542} & {\color[HTML]{000000} 0.252} & {\color[HTML]{000000} 0.253} & {\color[HTML]{000000} 0.1064} & {\color[HTML]{000000} \textbf{0.9514}} \\
{\color[HTML]{000000} \textbf{27}} & {\color[HTML]{000000} \textbf{Synapse-1.1}} & {\color[HTML]{000000} 0.3855} & {\color[HTML]{000000} 0.4583} & {\color[HTML]{000000} 0.494} & {\color[HTML]{000000} 0.475} & {\color[HTML]{000000} 0.4405} & {\color[HTML]{000000} \textbf{0.8611}} \\
{\color[HTML]{000000} \textbf{28}} & {\color[HTML]{000000} \textbf{Synapse-1.2}} & {\color[HTML]{000000} 0.3276} & {\color[HTML]{000000} 0.5654} & {\color[HTML]{000000} 0.529} & {\color[HTML]{000000} 0.571} & {\color[HTML]{000000} 0.4664} & {\color[HTML]{000000} \textbf{0.8102}} \\
{\color[HTML]{000000} \textbf{29}} & {\color[HTML]{000000} \textbf{Tomcat}} & {\color[HTML]{000000} 0.3907} & {\color[HTML]{000000} 0.2756} & {\color[HTML]{000000} 0.19} & {\color[HTML]{000000} 0.287} & {\color[HTML]{000000} 0.2626} & {\color[HTML]{000000} \textbf{0.8876}} \\
{\color[HTML]{000000} \textbf{30}} & {\color[HTML]{000000} \textbf{Velocity-1.4}} & {\color[HTML]{000000} 0.2143} & {\color[HTML]{000000} 0.5726} & {\color[HTML]{000000} \textbf{0.793}} & {\color[HTML]{000000} 0.734} & {\color[HTML]{000000} 0.3961} & {\color[HTML]{000000} 0.4257} \\
{\color[HTML]{000000} \textbf{31}} & {\color[HTML]{000000} \textbf{Velocity-1.6}} & {\color[HTML]{000000} 0.3333} & {\color[HTML]{000000} 0.5116} & {\color[HTML]{000000} 0.503} & {\color[HTML]{000000} 0.568} & {\color[HTML]{000000} 0.3587} & {\color[HTML]{000000} \textbf{0.8101}} \\
{\color[HTML]{000000} \textbf{32}} & {\color[HTML]{000000} \textbf{Xalan-2.4}} & {\color[HTML]{000000} 0.3436} & {\color[HTML]{000000} 0.3924} & {\color[HTML]{000000} 0.315} & {\color[HTML]{000000} 0.403} & {\color[HTML]{000000} 0.3279} & {\color[HTML]{000000} \textbf{0.8915}} \\
{\color[HTML]{000000} \textbf{33}} & {\color[HTML]{000000} \textbf{Xalan-2.5}} & {\color[HTML]{000000} 0.3366} & {\color[HTML]{000000} 0.5438} & {\color[HTML]{000000} 0.593} & {\color[HTML]{000000} 0.533} & {\color[HTML]{000000} 0.3944} & {\color[HTML]{000000} \textbf{0.6893}} \\
{\color[HTML]{000000} \textbf{34}} & {\color[HTML]{000000} \textbf{Xalan-2.6}} & {\color[HTML]{000000} 0.3415} & {\color[HTML]{000000} 0.5210} & {\color[HTML]{000000} 0.656} & {\color[HTML]{000000} 0.512} & {\color[HTML]{000000} 0.4611} & {\color[HTML]{000000} \textbf{0.7021}} \\
{\color[HTML]{000000} \textbf{35}} & {\color[HTML]{000000} \textbf{Xalan-2.7}} & {\color[HTML]{000000} 0.2697} & {\color[HTML]{000000} 0.6440} & {\color[HTML]{000000} \textbf{0.991}} & {\color[HTML]{000000} 0.616} & {\color[HTML]{000000} 0.4036} & {\color[HTML]{000000} 0.7248} \\
{\color[HTML]{000000} \textbf{36}} & {\color[HTML]{000000} \textbf{Xerces-1.2}} & {\color[HTML]{000000} 0.2345} & {\color[HTML]{000000} 0.2266} & {\color[HTML]{000000} 0.24} & {\color[HTML]{000000} 0.192} & {\color[HTML]{000000} 0.2367} & {\color[HTML]{000000} \textbf{0.8889}} \\
{\color[HTML]{000000} \textbf{37}} & {\color[HTML]{000000} \textbf{Xerces-1.3}} & {\color[HTML]{000000} 0.4027} & {\color[HTML]{000000} 0.3541} & {\color[HTML]{000000} 0.417} & {\color[HTML]{000000} 0.377} & {\color[HTML]{000000} 0.2297} & {\color[HTML]{000000} \textbf{0.8917}} \\
{\color[HTML]{000000} \textbf{38}} & {\color[HTML]{000000} \textbf{Xerces-1.4}} & {\color[HTML]{000000} 0.3012} & {\color[HTML]{000000} 0.4934} & {\color[HTML]{000000} \textbf{0.903}} & {\color[HTML]{000000} 0.69} & {\color[HTML]{000000} 0.7442} & {\color[HTML]{000000} 0.8163} \\ \hline
\multicolumn{1}{l}{{\color[HTML]{000000} }} & {\color[HTML]{000000} \textbf{Average}} & {\color[HTML]{000000} 0.3045} & {\color[HTML]{000000} 0.4691} & {\color[HTML]{000000} 0.5771} & {\color[HTML]{000000} 0.4692} & {\color[HTML]{000000} 0.3777} & {\color[HTML]{000000} \textbf{0.7825}} \\ \hline
\multicolumn{1}{l}{{\color[HTML]{000000} }} & {\color[HTML]{000000} \textbf{Improvement}} & {\color[HTML]{000000} \textbf{155.89}} & {\color[HTML]{000000} \textbf{66.11}} & {\color[HTML]{000000} \textbf{35.02}} & {\color[HTML]{000000} \textbf{66.07}} & {\color[HTML]{000000} \textbf{106.32}} & {\color[HTML]{000000} \textbf{-}} \\ \hline
\multicolumn{1}{l}{{\color[HTML]{000000} }} & {\color[HTML]{000000} \textbf{W/T/L}} & {\color[HTML]{000000} \textbf{37/0/1}} & {\color[HTML]{000000} \textbf{30/0/8}} & {\color[HTML]{000000} \textbf{26/0/12}} & {\color[HTML]{000000} \textbf{33/0/5}} & {\color[HTML]{000000} \textbf{34/0/4}} & {\color[HTML]{000000} \textbf{-}} \\ \hline
\multicolumn{1}{l}{{\color[HTML]{000000} }} & {\color[HTML]{000000} \textbf{p-value}} & \multicolumn{1}{r}{{\color[HTML]{000000} \textbf{1.55E-12}}} & \multicolumn{1}{r}{{\color[HTML]{000000} \textbf{1.44E-09}}} & \multicolumn{1}{r}{{\color[HTML]{000000} \textbf{1.40E-04}}} & \multicolumn{1}{r}{{\color[HTML]{000000} \textbf{1.35E-09}}} & \multicolumn{1}{r}{{\color[HTML]{000000} \textbf{3.28E-11}}} & {\color[HTML]{000000} \textbf{-}} \\ \hline
\multicolumn{1}{l}{{\color[HTML]{000000} }} & {\color[HTML]{000000} \textbf{Cliff's Delta}} & {\color[HTML]{000000} \textbf{0.9432}} & {\color[HTML]{000000} \textbf{0.8075}} & {\color[HTML]{000000} \textbf{0.5083}} & {\color[HTML]{000000} \textbf{0.8089}} & {\color[HTML]{000000} \textbf{0.8851}} & {\color[HTML]{000000} \textbf{-}} \\ \hline
\end{tabular}
\end{table*}

\begin{table*}[ht]
\caption{{\color[HTML]{000000} Comparison of the proposed HIEL model with the other published models in terms of {F-measure} on the NASA-MDP projects. The W/T/L indicates the proposed model won, tied, or lost to the other models}}
\label{FMeasure-NASA}
\begin{tabular}{clcccccc}
\hline
{\color[HTML]{000000} \textbf{S.No}} & {\color[HTML]{000000} \textbf{Project}} & {\color[HTML]{000000} \textbf{CODEP}} & {\color[HTML]{000000} \textbf{TCA+}} & {\color[HTML]{000000} \textbf{HYDRA}} & {\color[HTML]{000000} \textbf{TPTL}} & {\color[HTML]{000000} \textbf{TDS}} & {\color[HTML]{000000} \textbf{HIEL}} \\ \hline \hline
{\color[HTML]{000000} \textbf{1}} & {\color[HTML]{000000} \textbf{CM1}} & {\color[HTML]{000000} 0.0769} & {\color[HTML]{000000} 0.2182} & {\color[HTML]{000000} 0.7899} & {\color[HTML]{000000} \textbf{0.8978}} & {\color[HTML]{000000} 0.2338} & {\color[HTML]{000000} 0.8645} \\
{\color[HTML]{000000} \textbf{2}} & {\color[HTML]{000000} \textbf{JM1}} & {\color[HTML]{000000} 0.0334} & {\color[HTML]{000000} 0.3225} & {\color[HTML]{000000} 0.6897} & {\color[HTML]{000000} 0.7419} & {\color[HTML]{000000} 0.2138} & {\color[HTML]{000000} \textbf{0.8913}} \\
{\color[HTML]{000000} \textbf{3}} & {\color[HTML]{000000} \textbf{KC1}} & {\color[HTML]{000000} 0.0180} & {\color[HTML]{000000} 0.3951} & {\color[HTML]{000000} \textbf{0.9014}} & {\color[HTML]{000000} 0.6871} & {\color[HTML]{000000} 0.1901} & {\color[HTML]{000000} 0.882} \\
{\color[HTML]{000000} \textbf{4}} & {\color[HTML]{000000} \textbf{KC3}} & {\color[HTML]{000000} 0.0000} & {\color[HTML]{000000} 0.3051} & {\color[HTML]{000000} 0.7648} & {\color[HTML]{000000} 0.7889} & {\color[HTML]{000000} 0.2951} & {\color[HTML]{000000} \textbf{0.8162}} \\
{\color[HTML]{000000} \textbf{5}} & {\color[HTML]{000000} \textbf{MC1}} & {\color[HTML]{000000} 0.0494} & {\color[HTML]{000000} 0.0350} & {\color[HTML]{000000} 0.8987} & {\color[HTML]{000000} 0.6982} & {\color[HTML]{000000} 0.0271} & {\color[HTML]{000000} \textbf{0.9848}} \\
{\color[HTML]{000000} \textbf{6}} & {\color[HTML]{000000} \textbf{MC2}} & {\color[HTML]{000000} 0.1277} & {\color[HTML]{000000} 0.5301} & {\color[HTML]{000000} \textbf{0.7965}} & {\color[HTML]{000000} 0.7849} & {\color[HTML]{000000} 0.4198} & {\color[HTML]{000000} 0.7586} \\
{\color[HTML]{000000} \textbf{7}} & {\color[HTML]{000000} \textbf{MW1}} & {\color[HTML]{000000} 0.0000} & {\color[HTML]{000000} 0.1888} & {\color[HTML]{000000} \textbf{0.8733}} & {\color[HTML]{000000} 0.6911} & {\color[HTML]{000000} 0.1988} & {\color[HTML]{000000} 0.8611} \\
{\color[HTML]{000000} \textbf{8}} & {\color[HTML]{000000} \textbf{PC1}} & {\color[HTML]{000000} 0.1389} & {\color[HTML]{000000} 0.1629} & {\color[HTML]{000000} 0.8541} & {\color[HTML]{000000} \textbf{0.8922}} & {\color[HTML]{000000} 0.1520} & {\color[HTML]{000000} 0.8788} \\
{\color[HTML]{000000} \textbf{9}} & {\color[HTML]{000000} \textbf{PC2}} & {\color[HTML]{000000} 0.0000} & {\color[HTML]{000000} 0.0416} & {\color[HTML]{000000} 0.8765} & {\color[HTML]{000000} 0.9154} & {\color[HTML]{000000} 0.0377} & {\color[HTML]{000000} \textbf{0.9946}} \\
{\color[HTML]{000000} \textbf{10}} & {\color[HTML]{000000} \textbf{PC3}} & {\color[HTML]{000000} 0.0252} & {\color[HTML]{000000} 0.2422} & {\color[HTML]{000000} \textbf{0.9562}} & {\color[HTML]{000000} 0.7912} & {\color[HTML]{000000} 0.2956} & {\color[HTML]{000000} 0.9347} \\
{\color[HTML]{000000} \textbf{11}} & {\color[HTML]{000000} \textbf{PC4}} & {\color[HTML]{000000} 0.0220} & {\color[HTML]{000000} 0.2662} & {\color[HTML]{000000} 0.8451} & {\color[HTML]{000000} 0.8629} & {\color[HTML]{000000} 0.1498} & {\color[HTML]{000000} \textbf{0.9291}} \\
{\color[HTML]{000000} \textbf{12}} & {\color[HTML]{000000} \textbf{PC5}} & {\color[HTML]{000000} 0.1709} & {\color[HTML]{000000} 0.3876} & {\color[HTML]{000000} 0.8473} & {\color[HTML]{000000} 0.7696} & {\color[HTML]{000000} 0.1453} & {\color[HTML]{000000} \textbf{0.9847}} \\ \hline
\multicolumn{1}{l}{{\color[HTML]{000000} }} & {\color[HTML]{000000} \textbf{Average}} & {\color[HTML]{000000} \textbf{0.0552}} & {\color[HTML]{000000} \textbf{0.2579}} & {\color[HTML]{000000} \textbf{0.8411}} & {\color[HTML]{000000} \textbf{0.7934}} & {\color[HTML]{000000} \textbf{0.1966}} & {\color[HTML]{000000} \textbf{0.8984}} \\ \hline
\multicolumn{1}{l}{{\color[HTML]{000000} }} & {\color[HTML]{000000} \textbf{Improvement}} & {\color[HTML]{000000} \textbf{1527.54}} & {\color[HTML]{000000} \textbf{248.35}} & {\color[HTML]{000000} \textbf{6.81}} & {\color[HTML]{000000} \textbf{13.23}} & {\color[HTML]{000000} \textbf{356.97}} & {\color[HTML]{000000} \textbf{-}} \\ \hline
\multicolumn{1}{l}{{\color[HTML]{000000} }} & {\color[HTML]{000000} \textbf{W/T/L}} & {\color[HTML]{000000} \textbf{12/0/0}} & {\color[HTML]{000000} \textbf{12/0/0}} & {\color[HTML]{000000} \textbf{8/0/4}} & {\color[HTML]{000000} \textbf{9/0/3}} & {\color[HTML]{000000} \textbf{12/0/0}} & {\color[HTML]{000000} \textbf{-}} \\ \hline
\multicolumn{1}{l}{{\color[HTML]{000000} }} & {\color[HTML]{000000} \textbf{p-value}} & {\color[HTML]{000000} \textbf{3.60E-05}} & {\color[HTML]{000000} \textbf{7.40E-07}} & {\color[HTML]{000000} 0.06836} & {\color[HTML]{000000} \textbf{0.01}} & {\color[HTML]{000000} \textbf{7.40E-07}} & {\color[HTML]{000000} \textbf{-}} \\ \hline
\multicolumn{1}{l}{{\color[HTML]{000000} }} & {\color[HTML]{000000} \textbf{Cliff's Delta}} & {\color[HTML]{000000} \textbf{1}} & {\color[HTML]{000000} \textbf{1}} & {\color[HTML]{000000} \textbf{0.4444}} & {\color[HTML]{000000} \textbf{0.6111}} & {\color[HTML]{000000} \textbf{1}} & {\color[HTML]{000000} \textbf{-}} \\ \hline
\end{tabular}
\end{table*}

\begin{table*}[ht]
\centering
\caption{{\color[HTML]{000000}Comparison of the proposed HIEL model with the other published models in terms of {F-measure} on the AEEEM projects. The W/T/L indicates the proposed model won, tied, or lost to the other models}}
\label{FMeasure-AEEEM}
\begin{tabular}{clcccccc}
\hline
{\color[HTML]{000000} \textbf{S.No}} & \multicolumn{1}{l}{{\color[HTML]{000000} \textbf{Project}}} & {\color[HTML]{000000} \textbf{CODEP}} & {\color[HTML]{000000} \textbf{TCA+}} & {\color[HTML]{000000} \textbf{HYDRA}} & {\color[HTML]{000000} \textbf{TPTL}} & {\color[HTML]{000000} \textbf{CODEP}} & {\color[HTML]{000000} \textbf{HIEL}} \\ \hline \hline
{\color[HTML]{000000} \textbf{1}} & {\color[HTML]{000000} \textbf{Eclipse}} & {\color[HTML]{000000} 0.4709} & {\color[HTML]{000000} 0.3761} & {\color[HTML]{000000} 0.8612} & {\color[HTML]{000000} \textbf{0.8897}} & {\color[HTML]{000000} 0.3939} & {\color[HTML]{000000} 0.8716} \\
{\color[HTML]{000000} \textbf{2}} & {\color[HTML]{000000} \textbf{Equinox}} & {\color[HTML]{000000} 0.2282} & {\color[HTML]{000000} 0.6887} & {\color[HTML]{000000} 0.7998} & {\color[HTML]{000000} 0.7114} & {\color[HTML]{000000} 0.5094} & {\color[HTML]{000000} 0.7631} \\
{\color[HTML]{000000} \textbf{3}} & {\color[HTML]{000000} \textbf{Lucene}} & {\color[HTML]{000000} 0.2529} & {\color[HTML]{000000} 0.3247} & {\color[HTML]{000000} 0.9247} & {\color[HTML]{000000} 0.8462} & {\color[HTML]{000000} 0.3094} & {\color[HTML]{000000} \textbf{0.9509}} \\
{\color[HTML]{000000} \textbf{4}} & {\color[HTML]{000000} \textbf{Mylyn}} & {\color[HTML]{000000} 0.2186} & {\color[HTML]{000000} 0.3072} & {\color[HTML]{000000} 0.8618} & {\color[HTML]{000000} 0.7465} & {\color[HTML]{000000} 0.2530} & {\color[HTML]{000000} \textbf{0.9298}} \\
{\color[HTML]{000000} \textbf{5}} & {\color[HTML]{000000} \textbf{PDE}} & {\color[HTML]{000000} 0.2699} & {\color[HTML]{000000} 0.3021} & {\color[HTML]{000000} 0.8984} & {\color[HTML]{000000} 0.7249} & {\color[HTML]{000000} 0.2941} & {\color[HTML]{000000} \textbf{0.9159}} \\ \hline
\multicolumn{1}{l}{{\color[HTML]{000000} }} & {\color[HTML]{000000} \textbf{Average}} & {\color[HTML]{000000} 0.2881} & {\color[HTML]{000000} 0.3998} & {\color[HTML]{000000} 0.8692} & {\color[HTML]{000000} 0.7837} & {\color[HTML]{000000} 0.3520} & {\color[HTML]{000000} \textbf{0.8863}} \\ \hline
\multicolumn{1}{l}{{\color[HTML]{000000} }} & {\color[HTML]{000000} \textbf{Improvement}} & {\color[HTML]{000000} \textbf{207.76}} & {\color[HTML]{000000} \textbf{121.68}} & {\color[HTML]{000000} \textbf{1.96}} & {\color[HTML]{000000} \textbf{13.09}} & {\color[HTML]{000000} \textbf{151.79}} & {\color[HTML]{000000} \textbf{-}} \\ \hline
\multicolumn{1}{l}{{\color[HTML]{000000} }} & {\color[HTML]{000000} \textbf{W/T/L}} & {\color[HTML]{000000} \textbf{5/0/0}} & {\color[HTML]{000000} \textbf{5/0/0}} & {\color[HTML]{000000} \textbf{4/0/1}} & {\color[HTML]{000000} \textbf{4/0/1}} & {\color[HTML]{000000} \textbf{5/0/0}} & {\color[HTML]{000000} \textbf{-}} \\ \hline
\multicolumn{1}{l}{{\color[HTML]{000000} }} & {\color[HTML]{000000} \textbf{p-value}} & {\color[HTML]{000000} \textbf{0.0079}} & {\color[HTML]{000000} \textbf{0.0079}} & {\color[HTML]{000000} 0.402} & {\color[HTML]{000000} \textbf{0.036}} & {\color[HTML]{000000} \textbf{0.0079}} & {\color[HTML]{000000} \textbf{-}} \\ \hline
\multicolumn{1}{l}{{\color[HTML]{000000} }} & {\color[HTML]{000000} \textbf{Cliff's Delta}} & {\color[HTML]{000000} \textbf{1}} & {\color[HTML]{000000} \textbf{1}} & {\color[HTML]{000000} \textbf{0.36}} & {\color[HTML]{000000} \textbf{0.52}} & {\color[HTML]{000000} \textbf{1}} & {\color[HTML]{000000} \textbf{-}} \\ \hline
\end{tabular}
\end{table*}

\begin{table*}[ht]
\centering
\caption{\color[HTML]{000000} Comparison of the proposed HIEL model with the other published models in terms of {AUC} on the PROMISE projects. The W/T/L indicates the proposed model won, tied, or lost to the other models}
\label{AUC-PROMISE}
\begin{tabular}{clcccccc}
\hline
{\color[HTML]{000000} \textbf{S.No}} & {\color[HTML]{000000} \textbf{Project}} & {\color[HTML]{000000} \textbf{CODEP}} & {\color[HTML]{000000} \textbf{TCA+}} & {\color[HTML]{000000} \textbf{HYDRA}} & {\color[HTML]{000000} \textbf{TPTL}} & {\color[HTML]{000000} \textbf{TDS}} & {\color[HTML]{000000} \textbf{HIEL}} \\ \hline \hline
{\color[HTML]{000000} \textbf{1}} & {\color[HTML]{000000} \textbf{Ant-1.3}} & {\color[HTML]{000000} 0.5702} & {\color[HTML]{000000} 0.6702} & {\color[HTML]{000000} 0.8154} & {\color[HTML]{000000} 0.6545} & {\color[HTML]{000000} 0.5750} & {\color[HTML]{000000} \textbf{0.8329}} \\
{\color[HTML]{000000} \textbf{2}} & {\color[HTML]{000000} \textbf{Ant-1.4}} & {\color[HTML]{000000} 0.4966} & {\color[HTML]{000000} 0.5618} & {\color[HTML]{000000} 0.6512} & {\color[HTML]{000000} 0.6121} & {\color[HTML]{000000} 0.4342} & {\color[HTML]{000000} \textbf{0.6685}} \\
{\color[HTML]{000000} \textbf{3}} & {\color[HTML]{000000} \textbf{Ant-1.5}} & {\color[HTML]{000000} 0.6309} & {\color[HTML]{000000} \textbf{0.6986}} & {\color[HTML]{000000} 0.6698} & {\color[HTML]{000000} 0.5589} & {\color[HTML]{000000} 0.6164} & {\color[HTML]{000000} 0.5715} \\
{\color[HTML]{000000} \textbf{4}} & {\color[HTML]{000000} \textbf{Ant-1.6}} & {\color[HTML]{000000} 0.6721} & {\color[HTML]{000000} 0.6978} & {\color[HTML]{000000} \textbf{0.7145}} & {\color[HTML]{000000} 0.5846} & {\color[HTML]{000000} 0.6234} & {\color[HTML]{000000} 0.5925} \\
{\color[HTML]{000000} \textbf{5}} & {\color[HTML]{000000} \textbf{Ant-1.7}} & {\color[HTML]{000000} 0.6658} & {\color[HTML]{000000} 0.6966} & {\color[HTML]{000000} \textbf{0.7223}} & {\color[HTML]{000000} 0.5899} & {\color[HTML]{000000} 0.5858} & {\color[HTML]{000000} 0.5663} \\
{\color[HTML]{000000} \textbf{6}} & {\color[HTML]{000000} \textbf{Camel-1.0}} & {\color[HTML]{000000} 0.5294} & {\color[HTML]{000000} 0.6344} & {\color[HTML]{000000} 0.6541} & {\color[HTML]{000000} 0.5001} & {\color[HTML]{000000} 0.3992} & {\color[HTML]{000000} \textbf{0.681}} \\
{\color[HTML]{000000} \textbf{7}} & {\color[HTML]{000000} \textbf{Camel-1.2}} & {\color[HTML]{000000} 0.5242} & {\color[HTML]{000000} 0.5431} & {\color[HTML]{000000} 0.6248} & {\color[HTML]{000000} 0.6148} & {\color[HTML]{000000} 0.5639} & {\color[HTML]{000000} \textbf{0.6255}} \\
{\color[HTML]{000000} \textbf{8}} & {\color[HTML]{000000} \textbf{Camel-1.4}} & {\color[HTML]{000000} 0.5381} & {\color[HTML]{000000} 0.6418} & {\color[HTML]{000000} 0.8847} & {\color[HTML]{000000} 0.5877} & {\color[HTML]{000000} 0.5363} & {\color[HTML]{000000} \textbf{0.9197}} \\
{\color[HTML]{000000} \textbf{9}} & {\color[HTML]{000000} \textbf{Camel-1.6}} & {\color[HTML]{000000} 0.5301} & {\color[HTML]{000000} 0.5862} & {\color[HTML]{000000} 0.6088} & {\color[HTML]{000000} 0.5981} & {\color[HTML]{000000} 0.6033} & {\color[HTML]{000000} \textbf{0.6096}} \\
{\color[HTML]{000000} \textbf{10}} & {\color[HTML]{000000} \textbf{Jedit-3.2}} & {\color[HTML]{000000} 0.5814} & {\color[HTML]{000000} \textbf{0.7164}} & {\color[HTML]{000000} 0.6973} & {\color[HTML]{000000} 0.5549} & {\color[HTML]{000000} 0.5789} & {\color[HTML]{000000} 0.5808} \\
{\color[HTML]{000000} \textbf{11}} & {\color[HTML]{000000} \textbf{Jedit-4.0}} & {\color[HTML]{000000} 0.6221} & {\color[HTML]{000000} 0.6956} & {\color[HTML]{000000} \textbf{0.7111}} & {\color[HTML]{000000} 0.6154} & {\color[HTML]{000000} 0.6661} & {\color[HTML]{000000} 0.642} \\
{\color[HTML]{000000} \textbf{12}} & {\color[HTML]{000000} \textbf{Jedit-4.1}} & {\color[HTML]{000000} 0.6631} & {\color[HTML]{000000} \textbf{0.7003}} & {\color[HTML]{000000} 0.5122} & {\color[HTML]{000000} 0.5527} & {\color[HTML]{000000} 0.5816} & {\color[HTML]{000000} 0.5601} \\
{\color[HTML]{000000} \textbf{13}} & {\color[HTML]{000000} \textbf{Jedit-4.2}} & {\color[HTML]{000000} \textbf{0.7252}} & {\color[HTML]{000000} 0.6898} & {\color[HTML]{000000} 0.5009} & {\color[HTML]{000000} 0.5879} & {\color[HTML]{000000} 0.5387} & {\color[HTML]{000000} 0.5376} \\
{\color[HTML]{000000} \textbf{14}} & {\color[HTML]{000000} \textbf{Jedit-4.3}} & {\color[HTML]{000000} 0.5283} & {\color[HTML]{000000} 0.5209} & {\color[HTML]{000000} 0.5121} & {\color[HTML]{000000} 0.5262} & {\color[HTML]{000000} \textbf{0.5333}} & {\color[HTML]{000000} 0.513} \\
{\color[HTML]{000000} \textbf{15}} & {\color[HTML]{000000} \textbf{Log4j-1.0}} & {\color[HTML]{000000} 0.5488} & {\color[HTML]{000000} 0.6648} & {\color[HTML]{000000} 0.6712} & {\color[HTML]{000000} 0.6346} & {\color[HTML]{000000} 0.5504} & {\color[HTML]{000000} \textbf{0.8915}} \\
{\color[HTML]{000000} \textbf{16}} & {\color[HTML]{000000} \textbf{Log4j-1.1}} & {\color[HTML]{000000} 0.5738} & {\color[HTML]{000000} 0.7727} & {\color[HTML]{000000} 0.6939} & {\color[HTML]{000000} 0.6187} & {\color[HTML]{000000} 0.5439} & {\color[HTML]{000000} \textbf{0.8564}} \\
{\color[HTML]{000000} \textbf{17}} & {\color[HTML]{000000} \textbf{Log4j-1.2}} & {\color[HTML]{000000} 0.5428} & {\color[HTML]{000000} 0.5326} & {\color[HTML]{000000} \textbf{0.7856}} & {\color[HTML]{000000} 0.6909} & {\color[HTML]{000000} 0.4127} & {\color[HTML]{000000} 0.5006} \\
{\color[HTML]{000000} \textbf{18}} & {\color[HTML]{000000} \textbf{Lucene-2.0}} & {\color[HTML]{000000} 0.5577} & {\color[HTML]{000000} \textbf{0.6806}} & {\color[HTML]{000000} 0.6545} & {\color[HTML]{000000} 0.6155} & {\color[HTML]{000000} 0.6298} & {\color[HTML]{000000} 0.7326} \\
{\color[HTML]{000000} \textbf{19}} & {\color[HTML]{000000} \textbf{Lucene-2.2}} & {\color[HTML]{000000} 0.5195} & {\color[HTML]{000000} 0.6127} & {\color[HTML]{000000} 0.6333} & {\color[HTML]{000000} \textbf{0.6406}} & {\color[HTML]{000000} 0.5284} & {\color[HTML]{000000} 0.5853} \\
{\color[HTML]{000000} \textbf{20}} & {\color[HTML]{000000} \textbf{Lucene-2.4}} & {\color[HTML]{000000} 0.5497} & {\color[HTML]{000000} \textbf{0.6459}} & {\color[HTML]{000000} 0.5877} & {\color[HTML]{000000} 0.5641} & {\color[HTML]{000000} 0.5334} & {\color[HTML]{000000} 0.5899} \\
{\color[HTML]{000000} \textbf{21}} & {\color[HTML]{000000} \textbf{Poi-1.5}} & {\color[HTML]{000000} 0.5397} & {\color[HTML]{000000} 0.6521} & {\color[HTML]{000000} \textbf{0.6849}} & {\color[HTML]{000000} 0.6711} & {\color[HTML]{000000} 0.5020} & {\color[HTML]{000000} 0.5972} \\
{\color[HTML]{000000} \textbf{22}} & {\color[HTML]{000000} \textbf{Poi-2.0}} & {\color[HTML]{000000} 0.6098} & {\color[HTML]{000000} 0.5823} & {\color[HTML]{000000} 0.5588} & {\color[HTML]{000000} \textbf{0.6667}} & {\color[HTML]{000000} 0.6167} & {\color[HTML]{000000} 0.5784} \\
{\color[HTML]{000000} \textbf{23}} & {\color[HTML]{000000} \textbf{Poi-2.5}} & {\color[HTML]{000000} 0.5429} & {\color[HTML]{000000} \textbf{0.6669}} & {\color[HTML]{000000} 0.5501} & {\color[HTML]{000000} 0.6319} & {\color[HTML]{000000} 0.6395} & {\color[HTML]{000000} 0.5952} \\
{\color[HTML]{000000} \textbf{24}} & {\color[HTML]{000000} \textbf{Poi-3.0}} & {\color[HTML]{000000} 0.5926} & {\color[HTML]{000000} 0.7757} & {\color[HTML]{000000} \textbf{0.7889}} & {\color[HTML]{000000} 0.6244} & {\color[HTML]{000000} 0.6495} & {\color[HTML]{000000} 0.6125} \\
{\color[HTML]{000000} \textbf{25}} & {\color[HTML]{000000} \textbf{Redaktor}} & {\color[HTML]{000000} 0.5140} & {\color[HTML]{000000} 0.4892} & {\color[HTML]{000000} 0.5589} & {\color[HTML]{000000} 0.5445} & {\color[HTML]{000000} 0.6225} & {\color[HTML]{000000} \textbf{0.7091}} \\
{\color[HTML]{000000} \textbf{26}} & {\color[HTML]{000000} \textbf{Synapse-1.0}} & {\color[HTML]{000000} 0.7387} & {\color[HTML]{000000} 0.6602} & {\color[HTML]{000000} 0.6696} & {\color[HTML]{000000} 0.5654} & {\color[HTML]{000000} 0.4654} & {\color[HTML]{000000} \textbf{0.766}} \\
{\color[HTML]{000000} \textbf{27}} & {\color[HTML]{000000} \textbf{Synapse-1.1}} & {\color[HTML]{000000} 0.6117} & {\color[HTML]{000000} 0.5951} & {\color[HTML]{000000} 0.6562} & {\color[HTML]{000000} 0.5228} & {\color[HTML]{000000} 0.5634} & {\color[HTML]{000000} \textbf{0.7677}} \\
{\color[HTML]{000000} \textbf{28}} & {\color[HTML]{000000} \textbf{Synapse-1.2}} & {\color[HTML]{000000} 0.5781} & {\color[HTML]{000000} 0.6425} & {\color[HTML]{000000} 0.5941} & {\color[HTML]{000000} 0.6207} & {\color[HTML]{000000} 0.5319} & {\color[HTML]{000000} \textbf{0.6924}} \\
{\color[HTML]{000000} \textbf{29}} & {\color[HTML]{000000} \textbf{Tomcat}} & {\color[HTML]{000000} 0.7113} & {\color[HTML]{000000} \textbf{0.7135}} & {\color[HTML]{000000} 0.5846} & {\color[HTML]{000000} 0.6349} & {\color[HTML]{000000} 0.6907} & {\color[HTML]{000000} 0.5259} \\
{\color[HTML]{000000} \textbf{30}} & {\color[HTML]{000000} \textbf{Velocity-1.4}} & {\color[HTML]{000000} 0.5306} & {\color[HTML]{000000} 0.4796} & {\color[HTML]{000000} 0.6609} & {\color[HTML]{000000} \textbf{0.7129}} & {\color[HTML]{000000} 0.4352} & {\color[HTML]{000000} 0.5708} \\
{\color[HTML]{000000} \textbf{31}} & {\color[HTML]{000000} \textbf{Velocity-1.6}} & {\color[HTML]{000000} 0.5756} & {\color[HTML]{000000} 0.6165} & {\color[HTML]{000000} 0.5843} & {\color[HTML]{000000} 0.6848} & {\color[HTML]{000000} 0.3957} & {\color[HTML]{000000} \textbf{0.6873}} \\
{\color[HTML]{000000} \textbf{32}} & {\color[HTML]{000000} \textbf{Xalan-2.4}} & {\color[HTML]{000000} 0.6137} & {\color[HTML]{000000} \textbf{0.7017}} & {\color[HTML]{000000} 0.509} & {\color[HTML]{000000} 0.5678} & {\color[HTML]{000000} 0.5937} & {\color[HTML]{000000} 0.5708} \\
{\color[HTML]{000000} \textbf{33}} & {\color[HTML]{000000} \textbf{Xalan-2.5}} & {\color[HTML]{000000} 0.5654} & {\color[HTML]{000000} 0.5644} & {\color[HTML]{000000} 0.6967} & {\color[HTML]{000000} 0.6509} & {\color[HTML]{000000} 0.5157} & {\color[HTML]{000000} \textbf{0.73}} \\
{\color[HTML]{000000} \textbf{34}} & {\color[HTML]{000000} \textbf{Xalan-2.6}} & {\color[HTML]{000000} 0.5780} & {\color[HTML]{000000} 0.5477} & {\color[HTML]{000000} 0.5999} & {\color[HTML]{000000} 0.6333} & {\color[HTML]{000000} 0.5746} & {\color[HTML]{000000} \textbf{0.6717}} \\
{\color[HTML]{000000} \textbf{35}} & {\color[HTML]{000000} \textbf{Xalan-2.7}} & {\color[HTML]{000000} 0.5780} & {\color[HTML]{000000} 0.6923} & {\color[HTML]{000000} 0.7333} & {\color[HTML]{000000} 0.5147} & {\color[HTML]{000000} \textbf{0.8334}} & {\color[HTML]{000000} 0.5062} \\
{\color[HTML]{000000} \textbf{36}} & {\color[HTML]{000000} \textbf{Xerces-1.2}} & {\color[HTML]{000000} 0.5425} & {\color[HTML]{000000} 0.5143} & {\color[HTML]{000000} 0.5213} & {\color[HTML]{000000} 0.5411} & {\color[HTML]{000000} 0.4252} & {\color[HTML]{000000} \textbf{0.5859}} \\
{\color[HTML]{000000} \textbf{37}} & {\color[HTML]{000000} \textbf{Xerces-1.3}} & {\color[HTML]{000000} \textbf{0.6523}} & {\color[HTML]{000000} 0.6340} & {\color[HTML]{000000} 0.5708} & {\color[HTML]{000000} 0.5379} & {\color[HTML]{000000} 0.3544} & {\color[HTML]{000000} 0.5449} \\
{\color[HTML]{000000} \textbf{38}} & {\color[HTML]{000000} \textbf{Xerces-1.4}} & {\color[HTML]{000000} 0.5793} & {\color[HTML]{000000} 0.6021} & {\color[HTML]{000000} \textbf{0.6159}} & {\color[HTML]{000000} 0.5978} & {\color[HTML]{000000} 0.4093} & {\color[HTML]{000000} 0.5786} \\ \hline
\multicolumn{1}{l}{{\color[HTML]{000000} }} & {\color[HTML]{000000} \textbf{Average}} & {\color[HTML]{000000} 0.5848} & {\color[HTML]{000000} 0.6340} & {\color[HTML]{000000} 0.6386} & {\color[HTML]{000000} 0.6007} & {\color[HTML]{000000} 0.5488} & {\color[HTML]{000000} \textbf{0.6407}} \\ \hline
\multicolumn{1}{l}{{\color[HTML]{000000} }} & {\color[HTML]{000000} \textbf{Improvement}} & {\color[HTML]{000000} \textbf{9.56}} & {\color[HTML]{000000} \textbf{1.06}} & {\color[HTML]{000000} \textbf{0.33}} & {\color[HTML]{000000} \textbf{6.67}} & {\color[HTML]{000000} \textbf{16.76}} & {\color[HTML]{000000} \textbf{-}} \\ \hline
\multicolumn{1}{l}{{\color[HTML]{000000} }} & {\color[HTML]{000000} \textbf{W/T/L}} & {\color[HTML]{000000} \textbf{24/0/14}} & {\color[HTML]{000000} \textbf{18/0/20}} & {\color[HTML]{000000} \textbf{24/0/14}} & {\color[HTML]{000000} \textbf{26/0/12}} & {\color[HTML]{000000} \textbf{24/0/14}} & {\color[HTML]{000000} \textbf{-}} \\ \hline
\multicolumn{1}{l}{{\color[HTML]{000000} }} & {\color[HTML]{000000} \textbf{p-value}} & \multicolumn{1}{r}{{\color[HTML]{000000} \textbf{0.01322}}} & {\color[HTML]{000000} 0.5503} & {\color[HTML]{000000} 0.5195} & {\color[HTML]{000000} 0.2553} & {\color[HTML]{000000} \textbf{0.0007}} & {\color[HTML]{000000} \textbf{-}} \\ \hline
\multicolumn{1}{l}{{\color[HTML]{000000} }} & {\color[HTML]{000000} \textbf{Cliff's Delta}} & {\color[HTML]{000000} \textbf{0.3311}} & {\color[HTML]{000000} -0.0803} & {\color[HTML]{000000} -0.0865} & {\color[HTML]{000000} \textbf{0.1524}} & {\color[HTML]{000000} \textbf{0.4516}} & {\color[HTML]{000000} \textbf{-}} \\ \hline
\end{tabular}
\end{table*}

\begin{table*}[ht]
\caption{{\color[HTML]{000000}Comparison of the proposed HIEL model with the other published models in terms of {AUC} on the NASA projects.The W/T/L indicates the proposed model won, tied, or lost to the other models}}
\label{AUC-NASA}
\begin{tabular}{cccccccc}
\hline
{\color[HTML]{000000} \textbf{S.No}} & {\color[HTML]{000000} \textbf{Project}} & {\color[HTML]{000000} \textbf{CODEP}} & {\color[HTML]{000000} \textbf{TCA+}} & {\color[HTML]{000000} \textbf{HYDRA}} & {\color[HTML]{000000} \textbf{TPTL}} & {\color[HTML]{000000} \textbf{TDS}} & {\color[HTML]{000000} \textbf{HIEL}} \\ \hline \hline
{\color[HTML]{000000} \textbf{1}} & {\color[HTML]{000000} \textbf{CM1}} & {\color[HTML]{000000} 0.5106} & {\color[HTML]{000000} 0.5017} & {\color[HTML]{000000} 0.61} & {\color[HTML]{000000} \textbf{0.6218}} & {\color[HTML]{000000} 0.5678} & {\color[HTML]{000000} 0.5516} \\
{\color[HTML]{000000} \textbf{2}} & {\color[HTML]{000000} \textbf{JM1}} & {\color[HTML]{000000} 0.5081} & {\color[HTML]{000000} 0.5314} & {\color[HTML]{000000} 0.5212} & {\color[HTML]{000000} \textbf{0.5451}} & {\color[HTML]{000000} 0.4612} & {\color[HTML]{000000} 0.5151} \\
{\color[HTML]{000000} \textbf{3}} & {\color[HTML]{000000} \textbf{KC1}} & {\color[HTML]{000000} 0.5032} & {\color[HTML]{000000} \textbf{0.7057}} & {\color[HTML]{000000} 0.5146} & {\color[HTML]{000000} 0.512} & {\color[HTML]{000000} 0.575} & {\color[HTML]{000000} 0.5234} \\
{\color[HTML]{000000} \textbf{4}} & {\color[HTML]{000000} \textbf{KC3}} & {\color[HTML]{000000} 0.5000} & {\color[HTML]{000000} 0.5000} & {\color[HTML]{000000} 0.5226} & {\color[HTML]{000000} 0.5224} & {\color[HTML]{000000} \textbf{0.5750}} & {\color[HTML]{000000} 0.5335} \\
{\color[HTML]{000000} \textbf{5}} & {\color[HTML]{000000} \textbf{MC1}} & {\color[HTML]{000000} 0.5141} & {\color[HTML]{000000} \textbf{0.7621}} & {\color[HTML]{000000} 0.5349} & {\color[HTML]{000000} 0.5759} & {\color[HTML]{000000} 0.6016} & {\color[HTML]{000000} 0.5034} \\
{\color[HTML]{000000} \textbf{6}} & {\color[HTML]{000000} \textbf{MC2}} & {\color[HTML]{000000} 0.5341} & {\color[HTML]{000000} 0.5301} & {\color[HTML]{000000} \textbf{0.6265}} & {\color[HTML]{000000} 0.6117} & {\color[HTML]{000000} 0.5806} & {\color[HTML]{000000} \textbf{0.6265}} \\
{\color[HTML]{000000} \textbf{7}} & {\color[HTML]{000000} \textbf{MW1}} & {\color[HTML]{000000} 0.5000} & {\color[HTML]{000000} 0.5105} & {\color[HTML]{000000} \textbf{0.5601}} & {\color[HTML]{000000} 0.5159} & {\color[HTML]{000000} 0.5359} & {\color[HTML]{000000} 0.5502} \\
{\color[HTML]{000000} \textbf{8}} & {\color[HTML]{000000} \textbf{PC1}} & {\color[HTML]{000000} 0.5367} & {\color[HTML]{000000} 0.5509} & {\color[HTML]{000000} 0.5419} & {\color[HTML]{000000} 0.5264} & {\color[HTML]{000000} 0.5440} & {\color[HTML]{000000} \textbf{0.552}} \\
{\color[HTML]{000000} \textbf{9}} & {\color[HTML]{000000} \textbf{PC2}} & {\color[HTML]{000000} 0.5000} & {\color[HTML]{000000} \textbf{0.7648}} & {\color[HTML]{000000} 0.5058} & {\color[HTML]{000000} 0.6157} & {\color[HTML]{000000} 0.5430} & {\color[HTML]{000000} 0.5105} \\
{\color[HTML]{000000} \textbf{10}} & {\color[HTML]{000000} \textbf{PC3}} & {\color[HTML]{000000} 0.4985} & {\color[HTML]{000000} 0.5553} & {\color[HTML]{000000} 0.6578} & {\color[HTML]{000000} 0.6099} & {\color[HTML]{000000} 0.5985} & {\color[HTML]{000000} \textbf{0.7292}} \\
{\color[HTML]{000000} \textbf{11}} & {\color[HTML]{000000} \textbf{PC4}} & {\color[HTML]{000000} 0.5048} & {\color[HTML]{000000} 0.5962} & {\color[HTML]{000000} \textbf{0.6673}} & {\color[HTML]{000000} 0.6144} & {\color[HTML]{000000} 0.4588} & {\color[HTML]{000000} 0.6355} \\
{\color[HTML]{000000} \textbf{12}} & {\color[HTML]{000000} \textbf{PC5}} & {\color[HTML]{000000} 0.5494} & {\color[HTML]{000000} \textbf{0.9119}} & {\color[HTML]{000000} 0.5837} & {\color[HTML]{000000} 0.5431} & {\color[HTML]{000000} 0.5413} & {\color[HTML]{000000} 0.61} \\ \hline
\multicolumn{1}{l}{{\color[HTML]{000000} }} & {\color[HTML]{000000} \textbf{Average}} & {\color[HTML]{000000} 0.5133} & {\color[HTML]{000000} \textbf{0.6184}} & {\color[HTML]{000000} 0.5705} & {\color[HTML]{000000} {0.5679}} & {\color[HTML]{000000} {0.5486}} & {\color[HTML]{000000} {0.5701}} \\ \hline
\multicolumn{1}{l}{{\color[HTML]{000000} }} & {\color[HTML]{000000} \textbf{Improvement}} & {\color[HTML]{000000} \textbf{11.07}} & {\color[HTML]{000000} {-7.81}} & {\color[HTML]{000000} {-0.07}} & {\color[HTML]{000000} \textbf{0.39}} & {\color[HTML]{000000} \textbf{3.92}} & {\color[HTML]{000000} \textbf{-}} \\ \hline
\multicolumn{1}{l}{{\color[HTML]{000000} }} & {\color[HTML]{000000} \textbf{W/T/L}} & {\color[HTML]{000000} \textbf{11/0/1}} & {\color[HTML]{000000} \textbf{7/0/5}} & {\color[HTML]{000000} \textbf{6/1/5}} & {\color[HTML]{000000} \textbf{8/0/4}} & {\color[HTML]{000000} \textbf{7/0/5}} & {\color[HTML]{000000} \textbf{-}} \\ \hline
\multicolumn{1}{l}{{\color[HTML]{000000} }} & {\color[HTML]{000000} \textbf{p-value}} & {\color[HTML]{000000} \textbf{0.0035}} & {\color[HTML]{000000} 0.729} & {\color[HTML]{000000} 0.8852} & {\color[HTML]{000000} 0.9323} & {\color[HTML]{000000} 0.8398} & {\color[HTML]{000000} \textbf{-}} \\ \hline
\multicolumn{1}{l}{{\color[HTML]{000000} }} & {\color[HTML]{000000} \textbf{Cliff's Delta}} & {\color[HTML]{000000} \textbf{0.7083}} & {\color[HTML]{000000} -0.0903} & {\color[HTML]{000000} -0.0416} & {\color[HTML]{000000} -0.0278} & {\color[HTML]{000000} 0.0556} & {\color[HTML]{000000} \textbf{-}} \\ \hline
\end{tabular}
\end{table*}

\begin{table*}[ht]
\centering
\caption{\color[HTML]{000000} Comparison of the proposed HIEL model with the other published models in terms of {AUC} on the AEEEM projects. The W/T/L indicates the proposed model won, tied, or lost to the other models}
\label{AUC-AEEEM}
\begin{tabular}{clcccccc}
\hline
{\color[HTML]{000000} \textbf{S.No}} & {\color[HTML]{000000} \textbf{Project}} & {\color[HTML]{000000} \textbf{CODEP}} & {\color[HTML]{000000} \textbf{TCA+}} & {\color[HTML]{000000} \textbf{HYDRA}} & {\color[HTML]{000000} \textbf{TPTL}} & {\color[HTML]{000000} \textbf{TDS}} & {\color[HTML]{000000} \textbf{HIEL}} \\ \hline
{\color[HTML]{000000} \textbf{1}} & {\color[HTML]{000000} \textbf{Eclipse}} & {\color[HTML]{000000} \textbf{0.6636}} & {\color[HTML]{000000} 0.5713} & {\color[HTML]{000000} 0.5957} & {\color[HTML]{000000} 0.6259} & {\color[HTML]{000000} 0.6378} & {\color[HTML]{000000} 0.6173} \\
{\color[HTML]{000000} \textbf{2}} & {\color[HTML]{000000} \textbf{Equinox}} & {\color[HTML]{000000} 0.5582} & {\color[HTML]{000000} \textbf{0.7262}} & {\color[HTML]{000000} 0.6641} & {\color[HTML]{000000} 0.6241} & {\color[HTML]{000000} 0.6391} & {\color[HTML]{000000} 0.6945} \\
{\color[HTML]{000000} \textbf{3}} & {\color[HTML]{000000} \textbf{Lucene}} & {\color[HTML]{000000} 0.5764} & {\color[HTML]{000000} 0.7138} & {\color[HTML]{000000} 0.6412} & {\color[HTML]{000000} 0.6887} & {\color[HTML]{000000} 0.6485} & {\color[HTML]{000000} \textbf{0.7273}} \\
{\color[HTML]{000000} \textbf{4}} & {\color[HTML]{000000} \textbf{Mylyn}} & {\color[HTML]{000000} 0.5595} & {\color[HTML]{000000} 0.6108} & {\color[HTML]{000000} \textbf{0.6514}} & {\color[HTML]{000000} 0.5679} & {\color[HTML]{000000} 0.5702} & {\color[HTML]{000000} 0.5723} \\
{\color[HTML]{000000} \textbf{5}} & {\color[HTML]{000000} \textbf{PDE}} & {\color[HTML]{000000} 0.5769} & {\color[HTML]{000000} \textbf{0.6177}} & {\color[HTML]{000000} 0.5521} & {\color[HTML]{000000} 0.5659} & {\color[HTML]{000000} 0.6046} & {\color[HTML]{000000} 0.5364} \\ \hline
\multicolumn{1}{l}{{\color[HTML]{000000} }} & {\color[HTML]{000000} \textbf{Average}} & {\color[HTML]{000000} 0.5869} & {\color[HTML]{000000} \textbf{0.6480}} & {\color[HTML]{000000} 0.6209} & {\color[HTML]{000000} 0.6145} & {\color[HTML]{000000} 0.6200} & {\color[HTML]{000000} 0.6296} \\ \hline
\multicolumn{1}{l}{{\color[HTML]{000000} }} & {\color[HTML]{000000} \textbf{Improvement}} & {\color[HTML]{000000} \textbf{7.27}} & {\color[HTML]{000000} -2.84} & {\color[HTML]{000000} \textbf{1.39}} & {\color[HTML]{000000} \textbf{1.04}} & {\color[HTML]{000000} \textbf{1.53}} & {\color[HTML]{000000} \textbf{-}} \\ \hline
\multicolumn{1}{l}{{\color[HTML]{000000} }} & {\color[HTML]{000000} \textbf{W/T/L}} & {\color[HTML]{000000} \textbf{3/0/2}} & {\color[HTML]{000000} 2/0/3} & {\color[HTML]{000000} \textbf{3/0/2}} & {\color[HTML]{000000} 2/0/3} & {\color[HTML]{000000} \textbf{3/0/2}} & {\color[HTML]{000000} \textbf{-}} \\ \hline
\multicolumn{1}{l}{{\color[HTML]{000000} }} & {\color[HTML]{000000} \textbf{p-value}} & {\color[HTML]{000000} 0.5476} & {\color[HTML]{000000} 0.8413} & {\color[HTML]{000000} 1} & {\color[HTML]{000000} 1} & {\color[HTML]{000000} 1} & {\color[HTML]{000000} \textbf{-}} \\ \hline
\multicolumn{1}{l}{{\color[HTML]{000000} }} & {\color[HTML]{000000} \textbf{Cliff's Delta}} & {\color[HTML]{000000} \textbf{0.28}} & {\color[HTML]{000000} -0.12} & {\color[HTML]{000000} 0.04} & {\color[HTML]{000000} 0.04} & {\color[HTML]{000000} 0.04} & {\color[HTML]{000000} \textbf{-}} \\ \hline
\end{tabular}
\end{table*}

After selecting the value of $\beta$, we have {compared the proposed HIEL (which is constructed using a set of 60 diverse classifiers) with the other baseline models} on the target projects {selected from the publicly available repositories such as PROMISE, NASA, and AEEEM. The comparative analysis is conducted with works such as TDS, TCA+, CODEP, HYDRA, and TPTL. The comparative analysis is conducted using measures such as {F-measure} and {AUC}. In \cite{herbold2017comparative}, Herbold et al. conducted a large scale empirical survey on the cross project defect prediction studies and replicated 24 approaches that were published between 2008 and 2015. Amongst the replicated models, we have utilised the {F-measure} and {AUC} values directly from the replicated works of TDS, TCA+, and CODEP on PROMISE, NASA, and AEEEM datasets.} Because, irrespective of utilising the training set, each of the above models computes the final performance on the same target project. {And for the works HYDRA and TPTL, we used publicly available source code provided by the respective authors}.

{Tables \ref{FMeasure-PROMISE}, \ref{FMeasure-NASA}, and \ref{FMeasure-AEEEM} represent the {F-measure} values of various models on the PROMISE, NASA, and AEEEM projects, respectively.} From table \ref{FMeasure-PROMISE}, it is observed that, {on many target PROMISE projects,} the proposed HIEL model shows a substantial improvement over the other base-line models, significantly. {On an average across 38 PROMISE projects, the proposed HIEL model achieved an improvement of 155.89\%, 66.11\%, 35.02\%, 66.07\%, and 106.32\% on the base-line models such as CODEP, TCA+, HYDRA, TPTL, and TDS, respectively. Similar results have been observed in NASA projects also. From the table \ref{FMeasure-NASA}, it is observed that, on an average across 12 NASA projects, the proposed HIEL model achieved an improvement of 1527.54\%, 248.35\%, 6.81\%, 13.23\%, and 356.97\% on the base-line models such as CODEP, TCA+, HYDRA, TPTL, and TDS, respectively. Similarly, from table \ref{FMeasure-AEEEM}, it is observed that, on the 5 AEEEM projects, on an average, the proposed HIEL model achieved an improvement of 207.76\%, 121.68\%, 1.96\%, 13.09\%, and 151.79\% on the base-line models such as CODEP, TCA+, HYDRA, TPTL, and TDS, respectively.}

{Tables \ref{AUC-PROMISE}, \ref{AUC-NASA}, and \ref{AUC-AEEEM} represent the {AUC} values of various models on the PROMISE, NASA, and AEEEM projects, respectively. The performances (in terms of {AUC}) of all the models on NASA and AEEEM projects indicate the competitive outcomes. From table \ref{AUC-PROMISE}, it is observed that, on many target PROMISE projects, the proposed HIEL model shows a little improvement over the other base-line models. On an average across 38 PROMISE projects, the proposed HIEL model achieved an improvement of 9.56\%, 1.06\%, 0.33\%, 6.67\%, and 16.76\% on the base-line models such as CODEP, TCA+, HYDRA, TPTL, and TDS respectively. From table \ref{AUC-NASA}, it is observed that, on an average across 12 NASA projects, the proposed HIEL model achieved an improvement of 11.06\%, 0.39\%, and 3.92\% over the base-line models such as CODEP, TPTL, and TDS respectively. But the models TCA+ and HYDRA are showing their better performances over the proposed HIEL model. From table \ref{AUC-AEEEM}, it is observed that, on the 5 AEEEM projects, on an average, the proposed HIEL model achieved an improvement of 7.27\%, 1.39\%, 1.04\%, and 1.53\% on the base-line models such as CODEP, HYDRA, TPTL, and TDS, respectively. In this case, the model  TCA+ is showing its better performance over the proposed HIEL model.}

The tables \ref{FMeasure-PROMISE}, \ref{FMeasure-NASA}, and \ref{FMeasure-AEEEM} also present the results of the  significant tests {(conducted using {F-measure})} such as the one-sample Wilcoxson signed rank test and Cliff's delta effect size test. It is observed from the table \ref{FMeasure-PROMISE} that, the proposed model shows its significant improvement in terms of {F-measure} over all the other models. {The values of Cliff's delta also indicate the large effect among the models. On NASA projects (from the table \ref{FMeasure-NASA}), except on HYDRA, the proposed model achieved significantly better results. Despite the fact that the Wilcoxson signed rank test finds no significant difference between HYDRA and HIEL ($p = 0.0684 > 0.05$), the Cliff's delta effect size test finds a medium effect for the HIEL ($delta = 0.444$) when compared to the HYDRA model. The table \ref{FMeasure-AEEEM} also indicate the similar results. On AEEEM projects, except on HYDRA, the proposed model achieved significantly better results in terms of {F-measure}. Even though there is no significant difference between HYDRA and HIEL performances, ($p=0.402>0.05$) is observed according to the Wilcoxson signed rank test, the Cliff's delta effect size test shows that the HIEL has a medium effect ($\delta=0.36$) when compared with the HYDRA model.}

{Similarly, the tables \ref{AUC-PROMISE}, \ref{AUC-NASA}, and \ref{AUC-AEEEM} also present the results of the significant tests (conducted using {AUC}) such as the one-sample Wilcoxson signed rank test and Cliff's delta effect size test. It is observed from the table \ref{AUC-PROMISE}, that the proposed model shows a significant improvement in the performances in terms of {AUC} on the models such as CODEP and TDS. The values of Cliff's delta also indicate the large effect among these models. On the other models, the HIEL does not show its significance difference in terms of both the tests.}

{From the table \ref{AUC-NASA}, on NASA projects, it is observed that, when compared with the CODEP model, the proposed model achieved significantly better average performance in terms of the Wilcoxson signed rank test and Cliff's delta effect size test. On the other models such as TCA+, HYDRA, TPTL, and TDS, the HIEL does not show a significant difference in the performances in terms of both the tests. From table \ref{AUC-AEEEM}, we observe on AEEEM projects that, even though the proposed model achieved better average {AUC}s, it does not achieve significantly better results in terms of either of the tests.}

{In summary, on some models, even though the proposed model is incapable of producing significant results in terms of {AUC}, it is showing a significantly better improvement in the average performance (in terms of {F-measure}) than the other models.}
\subsection{The Saved {Budget}, Remaining Service Time and Failure Analysis}
\label{Cost,Service Time and Failure Analysis}
\begin{figure*}[]
    \centering
    \includegraphics[width = \textwidth]{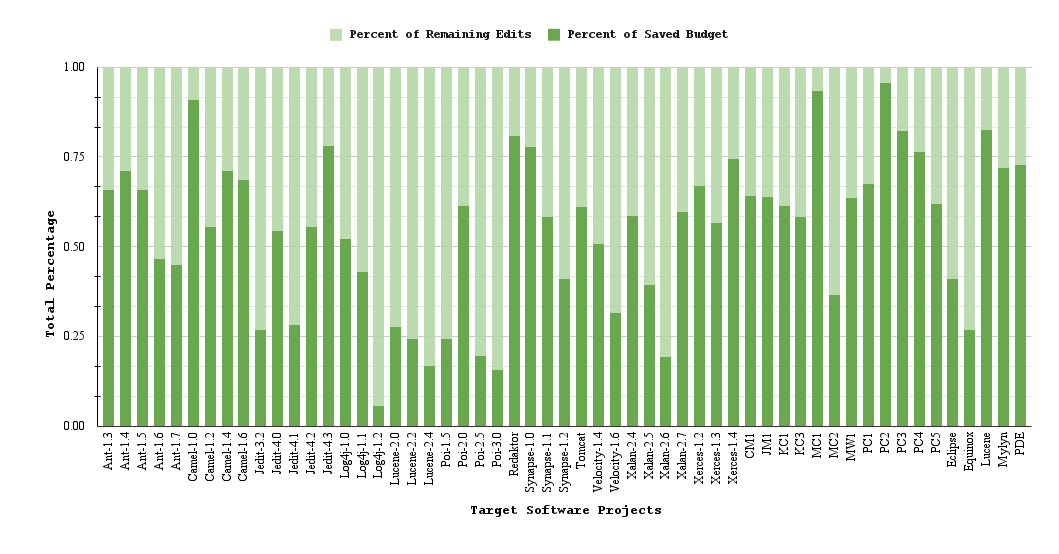}
    \caption{{The percent of saved budget (PSB) and the percent of remaining edits (PRE), obtained from the HIEL model on all the target projects}}
    \label{PSB-RST}
\end{figure*}
\begin{figure*}[!ht]
    \centering
    \includegraphics[width = \textwidth]{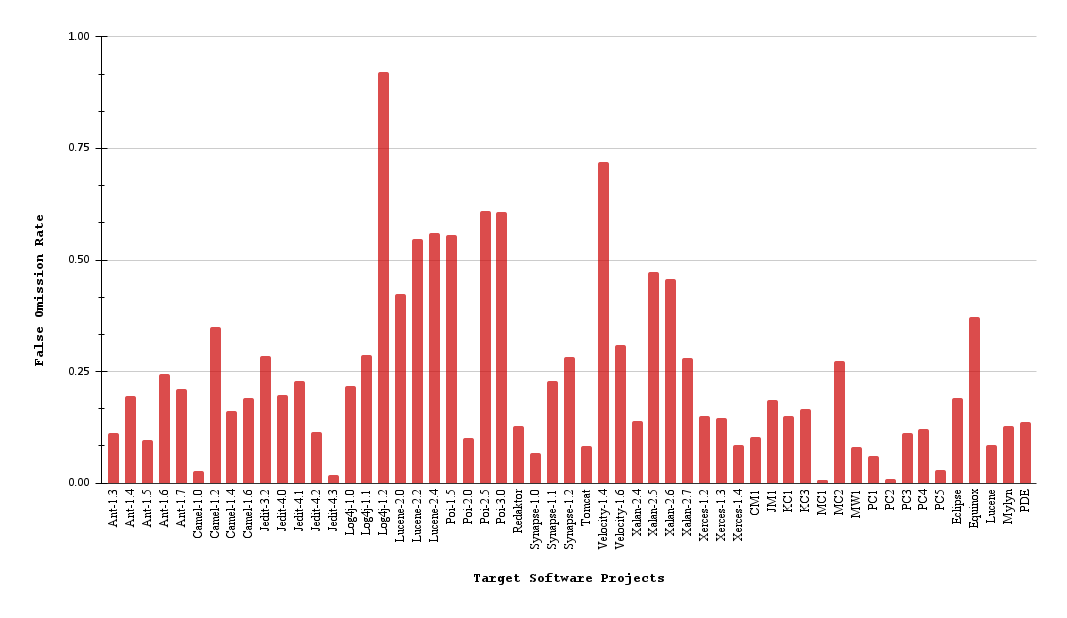}
    \caption{{The false omission rates ({FOR}) obtained from the HIEL model on all the target projects.}}
    \label{FailureRate}
\end{figure*}

\begin{table*}[ht]
\centering
\caption{\color[HTML]{000000} The supplementary details of the proposed measures}
\label{SupplementaryDetails}
\resizebox{\textwidth}{!}{
\begin{tabular}{clrrrrrrrr}
\hline
\multicolumn{10}{c}{{\color[HTML]{000000} \textbf{PROMISE Projects}}} \\ \hline
{\color[HTML]{000000} \textbf{S.No}} & \multicolumn{1}{c}{{\color[HTML]{000000} \textbf{Project}}} & \multicolumn{1}{c}{{\color[HTML]{000000} \textbf{Total LoC}}} & \multicolumn{1}{c}{{\color[HTML]{000000} \textbf{Defects}}} & \multicolumn{1}{c}{{\color[HTML]{000000} \textbf{Saved Budget}}} & \multicolumn{1}{c}{{\color[HTML]{000000} \textbf{\begin{tabular}[c]{@{}c@{}}Remaining\\ Service\\ Time\end{tabular}}}} & \multicolumn{1}{c}{{\color[HTML]{000000} \textbf{Project Hours}}} & \multicolumn{1}{c}{{\color[HTML]{000000} \textbf{\begin{tabular}[c]{@{}c@{}}Original\\ Editing\\ Rate\end{tabular}}}} & \multicolumn{1}{c}{{\color[HTML]{000000} \textbf{Editing Rate}}} & \multicolumn{1}{c}{{\color[HTML]{000000} \textbf{\begin{tabular}[c]{@{}c@{}}Decreased\\ Editing\\ Rate\end{tabular}}}} \\ \hline
{\color[HTML]{000000} \textbf{1}} & {\color[HTML]{000000} \textbf{Ant-1.3}} & {\color[HTML]{000000} 37699} & {\color[HTML]{000000} 20} & {\color[HTML]{000000} 24755} & {\color[HTML]{000000} 12944} & {\color[HTML]{000000} 129.44} & {\color[HTML]{000000} 1884.95} & {\color[HTML]{000000} 647.20} & {\color[HTML]{000000} 1237.75} \\
{\color[HTML]{000000} \textbf{2}} & {\color[HTML]{000000} \textbf{Ant-1.4}} & {\color[HTML]{000000} 54195} & {\color[HTML]{000000} 40} & {\color[HTML]{000000} 38528} & {\color[HTML]{000000} 15667} & {\color[HTML]{000000} 156.67} & {\color[HTML]{000000} 1354.88} & {\color[HTML]{000000} 391.68} & {\color[HTML]{000000} 963.20} \\
{\color[HTML]{000000} \textbf{3}} & {\color[HTML]{000000} \textbf{Ant-1.5}} & {\color[HTML]{000000} 87047} & {\color[HTML]{000000} 32} & {\color[HTML]{000000} 57345} & {\color[HTML]{000000} 29702} & {\color[HTML]{000000} 297.02} & {\color[HTML]{000000} 2720.22} & {\color[HTML]{000000} 928.19} & {\color[HTML]{000000} 1792.03} \\
{\color[HTML]{000000} \textbf{4}} & {\color[HTML]{000000} \textbf{Ant-1.6}} & {\color[HTML]{000000} 113246} & {\color[HTML]{000000} 92} & {\color[HTML]{000000} 52707} & {\color[HTML]{000000} 60539} & {\color[HTML]{000000} 605.39} & {\color[HTML]{000000} 1230.93} & {\color[HTML]{000000} 658.03} & {\color[HTML]{000000} 572.90} \\
{\color[HTML]{000000} \textbf{5}} & {\color[HTML]{000000} \textbf{Ant-1.7}} & {\color[HTML]{000000} 208653} & {\color[HTML]{000000} 166} & {\color[HTML]{000000} 93929} & {\color[HTML]{000000} 114724} & {\color[HTML]{000000} 1147.24} & {\color[HTML]{000000} 1256.95} & {\color[HTML]{000000} 691.11} & {\color[HTML]{000000} 565.84} \\
{\color[HTML]{000000} \textbf{6}} & {\color[HTML]{000000} \textbf{Camel-1.0}} & {\color[HTML]{000000} 33721} & {\color[HTML]{000000} 13} & {\color[HTML]{000000} 30679} & {\color[HTML]{000000} 3042} & {\color[HTML]{000000} 30.42} & {\color[HTML]{000000} 2593.92} & {\color[HTML]{000000} 234.00} & {\color[HTML]{000000} 2359.92} \\
{\color[HTML]{000000} \textbf{7}} & {\color[HTML]{000000} \textbf{Camel-1.2}} & {\color[HTML]{000000} 66302} & {\color[HTML]{000000} 216} & {\color[HTML]{000000} 36764} & {\color[HTML]{000000} 29538} & {\color[HTML]{000000} 295.38} & {\color[HTML]{000000} 306.95} & {\color[HTML]{000000} 136.75} & {\color[HTML]{000000} 170.20} \\
{\color[HTML]{000000} \textbf{8}} & {\color[HTML]{000000} \textbf{Camel-1.4}} & {\color[HTML]{000000} 98080} & {\color[HTML]{000000} 145} & {\color[HTML]{000000} 69685} & {\color[HTML]{000000} 28395} & {\color[HTML]{000000} 283.95} & {\color[HTML]{000000} 676.41} & {\color[HTML]{000000} 195.83} & {\color[HTML]{000000} 480.59} \\
{\color[HTML]{000000} \textbf{9}} & {\color[HTML]{000000} \textbf{Camel-1.6}} & {\color[HTML]{000000} 113055} & {\color[HTML]{000000} 188} & {\color[HTML]{000000} 77510} & {\color[HTML]{000000} 35545} & {\color[HTML]{000000} 355.45} & {\color[HTML]{000000} 601.36} & {\color[HTML]{000000} 189.07} & {\color[HTML]{000000} 412.29} \\
{\color[HTML]{000000} \textbf{10}} & {\color[HTML]{000000} \textbf{Jedit-3.2}} & {\color[HTML]{000000} 128883} & {\color[HTML]{000000} 90} & {\color[HTML]{000000} 34674} & {\color[HTML]{000000} 94209} & {\color[HTML]{000000} 942.09} & {\color[HTML]{000000} 1432.03} & {\color[HTML]{000000} 1046.77} & {\color[HTML]{000000} 385.27} \\
{\color[HTML]{000000} \textbf{11}} & {\color[HTML]{000000} \textbf{Jedit-4.0}} & {\color[HTML]{000000} 144803} & {\color[HTML]{000000} 75} & {\color[HTML]{000000} 78584} & {\color[HTML]{000000} 66219} & {\color[HTML]{000000} 662.19} & {\color[HTML]{000000} 1930.71} & {\color[HTML]{000000} 882.92} & {\color[HTML]{000000} 1047.79} \\
{\color[HTML]{000000} \textbf{12}} & {\color[HTML]{000000} \textbf{Jedit-4.1}} & {\color[HTML]{000000} 153087} & {\color[HTML]{000000} 79} & {\color[HTML]{000000} 43047} & {\color[HTML]{000000} 110040} & {\color[HTML]{000000} 1100.40} & {\color[HTML]{000000} 1937.81} & {\color[HTML]{000000} 1392.91} & {\color[HTML]{000000} 544.90} \\
{\color[HTML]{000000} \textbf{13}} & {\color[HTML]{000000} \textbf{Jedit-4.2}} & {\color[HTML]{000000} 170683} & {\color[HTML]{000000} 48} & {\color[HTML]{000000} 94791} & {\color[HTML]{000000} 75892} & {\color[HTML]{000000} 758.92} & {\color[HTML]{000000} 3555.90} & {\color[HTML]{000000} 1581.08} & {\color[HTML]{000000} 1974.81} \\
{\color[HTML]{000000} \textbf{14}} & {\color[HTML]{000000} \textbf{Jedit-4.3}} & {\color[HTML]{000000} 202363} & {\color[HTML]{000000} 11} & {\color[HTML]{000000} 158076} & {\color[HTML]{000000} 44287} & {\color[HTML]{000000} 442.87} & {\color[HTML]{000000} 18396.64} & {\color[HTML]{000000} 4026.09} & {\color[HTML]{000000} 14370.55} \\
{\color[HTML]{000000} \textbf{15}} & {\color[HTML]{000000} \textbf{Log4j-1.0}} & {\color[HTML]{000000} 21549} & {\color[HTML]{000000} 34} & {\color[HTML]{000000} 11220} & {\color[HTML]{000000} 10329} & {\color[HTML]{000000} 103.29} & {\color[HTML]{000000} 633.79} & {\color[HTML]{000000} 303.79} & {\color[HTML]{000000} 330.00} \\
{\color[HTML]{000000} \textbf{16}} & {\color[HTML]{000000} \textbf{Log4j-1.1}} & {\color[HTML]{000000} 19938} & {\color[HTML]{000000} 37} & {\color[HTML]{000000} 8537} & {\color[HTML]{000000} 11401} & {\color[HTML]{000000} 114.01} & {\color[HTML]{000000} 538.86} & {\color[HTML]{000000} 308.14} & {\color[HTML]{000000} 230.73} \\
{\color[HTML]{000000} \textbf{17}} & {\color[HTML]{000000} \textbf{Log4j-1.2}} & {\color[HTML]{000000} 38191} & {\color[HTML]{000000} 189} & {\color[HTML]{000000} 2103} & {\color[HTML]{000000} 36088} & {\color[HTML]{000000} 360.88} & {\color[HTML]{000000} 202.07} & {\color[HTML]{000000} 190.94} & {\color[HTML]{000000} 11.13} \\
{\color[HTML]{000000} \textbf{18}} & {\color[HTML]{000000} \textbf{Lucene-2.0}} & {\color[HTML]{000000} 50596} & {\color[HTML]{000000} 91} & {\color[HTML]{000000} 13920} & {\color[HTML]{000000} 36676} & {\color[HTML]{000000} 366.76} & {\color[HTML]{000000} 556.00} & {\color[HTML]{000000} 403.03} & {\color[HTML]{000000} 152.97} \\
{\color[HTML]{000000} \textbf{19}} & {\color[HTML]{000000} \textbf{Lucene-2.2}} & {\color[HTML]{000000} 63571} & {\color[HTML]{000000} 144} & {\color[HTML]{000000} 15384} & {\color[HTML]{000000} 48187} & {\color[HTML]{000000} 481.87} & {\color[HTML]{000000} 441.47} & {\color[HTML]{000000} 334.63} & {\color[HTML]{000000} 106.83} \\
{\color[HTML]{000000} \textbf{20}} & {\color[HTML]{000000} \textbf{Lucene-2.4}} & {\color[HTML]{000000} 102859} & {\color[HTML]{000000} 203} & {\color[HTML]{000000} 17128} & {\color[HTML]{000000} 85731} & {\color[HTML]{000000} 857.31} & {\color[HTML]{000000} 506.69} & {\color[HTML]{000000} 422.32} & {\color[HTML]{000000} 84.37} \\
{\color[HTML]{000000} \textbf{21}} & {\color[HTML]{000000} \textbf{Poi-1.5}} & {\color[HTML]{000000} 55428} & {\color[HTML]{000000} 141} & {\color[HTML]{000000} 13422} & {\color[HTML]{000000} 42006} & {\color[HTML]{000000} 420.06} & {\color[HTML]{000000} 393.11} & {\color[HTML]{000000} 297.91} & {\color[HTML]{000000} 95.19} \\
{\color[HTML]{000000} \textbf{22}} & {\color[HTML]{000000} \textbf{Poi-2.0}} & {\color[HTML]{000000} 93171} & {\color[HTML]{000000} 37} & {\color[HTML]{000000} 57128} & {\color[HTML]{000000} 36043} & {\color[HTML]{000000} 360.43} & {\color[HTML]{000000} 2518.14} & {\color[HTML]{000000} 974.14} & {\color[HTML]{000000} 1544.00} \\
{\color[HTML]{000000} \textbf{23}} & {\color[HTML]{000000} \textbf{Poi-2.5}} & {\color[HTML]{000000} 119731} & {\color[HTML]{000000} 248} & {\color[HTML]{000000} 21104} & {\color[HTML]{000000} 98627} & {\color[HTML]{000000} 986.27} & {\color[HTML]{000000} 482.79} & {\color[HTML]{000000} 397.69} & {\color[HTML]{000000} 85.10} \\
{\color[HTML]{000000} \textbf{24}} & {\color[HTML]{000000} \textbf{Poi-3.0}} & {\color[HTML]{000000} 129327} & {\color[HTML]{000000} 281} & {\color[HTML]{000000} 20275} & {\color[HTML]{000000} 109052} & {\color[HTML]{000000} 1090.52} & {\color[HTML]{000000} 460.24} & {\color[HTML]{000000} 388.09} & {\color[HTML]{000000} 72.15} \\
{\color[HTML]{000000} \textbf{25}} & {\color[HTML]{000000} \textbf{Redaktor}} & {\color[HTML]{000000} 59280} & {\color[HTML]{000000} 27} & {\color[HTML]{000000} 47876} & {\color[HTML]{000000} 11404} & {\color[HTML]{000000} 114.04} & {\color[HTML]{000000} 2195.56} & {\color[HTML]{000000} 422.37} & {\color[HTML]{000000} 1773.19} \\
{\color[HTML]{000000} \textbf{26}} & {\color[HTML]{000000} \textbf{Synapse-1.0}} & {\color[HTML]{000000} 28806} & {\color[HTML]{000000} 16} & {\color[HTML]{000000} 22391} & {\color[HTML]{000000} 6415} & {\color[HTML]{000000} 64.15} & {\color[HTML]{000000} 1800.38} & {\color[HTML]{000000} 400.94} & {\color[HTML]{000000} 1399.44} \\
{\color[HTML]{000000} \textbf{27}} & {\color[HTML]{000000} \textbf{Synapse-1.1}} & {\color[HTML]{000000} 42302} & {\color[HTML]{000000} 60} & {\color[HTML]{000000} 24610} & {\color[HTML]{000000} 17692} & {\color[HTML]{000000} 176.92} & {\color[HTML]{000000} 705.03} & {\color[HTML]{000000} 294.87} & {\color[HTML]{000000} 410.17} \\
{\color[HTML]{000000} \textbf{28}} & {\color[HTML]{000000} \textbf{Synapse-1.2}} & {\color[HTML]{000000} 53500} & {\color[HTML]{000000} 86} & {\color[HTML]{000000} 21963} & {\color[HTML]{000000} 31537} & {\color[HTML]{000000} 315.37} & {\color[HTML]{000000} 622.09} & {\color[HTML]{000000} 366.71} & {\color[HTML]{000000} 255.38} \\
{\color[HTML]{000000} \textbf{29}} & {\color[HTML]{000000} \textbf{Tomcat}} & {\color[HTML]{000000} 300674} & {\color[HTML]{000000} 77} & {\color[HTML]{000000} 183474} & {\color[HTML]{000000} 117200} & {\color[HTML]{000000} 1172.00} & {\color[HTML]{000000} 3904.86} & {\color[HTML]{000000} 1522.08} & {\color[HTML]{000000} 2382.78} \\
{\color[HTML]{000000} \textbf{30}} & {\color[HTML]{000000} \textbf{Velocity-1.4}} & {\color[HTML]{000000} 51713} & {\color[HTML]{000000} 147} & {\color[HTML]{000000} 26258} & {\color[HTML]{000000} 25455} & {\color[HTML]{000000} 254.55} & {\color[HTML]{000000} 351.79} & {\color[HTML]{000000} 173.16} & {\color[HTML]{000000} 178.63} \\
{\color[HTML]{000000} \textbf{31}} & {\color[HTML]{000000} \textbf{Velocity-1.6}} & {\color[HTML]{000000} 57012} & {\color[HTML]{000000} 78} & {\color[HTML]{000000} 17946} & {\color[HTML]{000000} 39066} & {\color[HTML]{000000} 390.66} & {\color[HTML]{000000} 730.92} & {\color[HTML]{000000} 500.85} & {\color[HTML]{000000} 230.08} \\
{\color[HTML]{000000} \textbf{32}} & {\color[HTML]{000000} \textbf{Xalan-2.4}} & {\color[HTML]{000000} 225088} & {\color[HTML]{000000} 110} & {\color[HTML]{000000} 131695} & {\color[HTML]{000000} 93393} & {\color[HTML]{000000} 933.93} & {\color[HTML]{000000} 2046.25} & {\color[HTML]{000000} 849.03} & {\color[HTML]{000000} 1197.23} \\
{\color[HTML]{000000} \textbf{33}} & {\color[HTML]{000000} \textbf{Xalan-2.5}} & {\color[HTML]{000000} 304860} & {\color[HTML]{000000} 387} & {\color[HTML]{000000} 119902} & {\color[HTML]{000000} 184958} & {\color[HTML]{000000} 1849.58} & {\color[HTML]{000000} 787.75} & {\color[HTML]{000000} 477.93} & {\color[HTML]{000000} 309.82} \\
{\color[HTML]{000000} \textbf{34}} & {\color[HTML]{000000} \textbf{Xalan-2.6}} & {\color[HTML]{000000} 411737} & {\color[HTML]{000000} 411} & {\color[HTML]{000000} 78959} & {\color[HTML]{000000} 332778} & {\color[HTML]{000000} 3327.78} & {\color[HTML]{000000} 1001.79} & {\color[HTML]{000000} 809.68} & {\color[HTML]{000000} 192.11} \\
{\color[HTML]{000000} \textbf{35}} & {\color[HTML]{000000} \textbf{Xalan-2.7}} & {\color[HTML]{000000} 428555} & {\color[HTML]{000000} 898} & {\color[HTML]{000000} 256233} & {\color[HTML]{000000} 172322} & {\color[HTML]{000000} 1723.22} & {\color[HTML]{000000} 477.23} & {\color[HTML]{000000} 191.90} & {\color[HTML]{000000} 285.34} \\
{\color[HTML]{000000} \textbf{36}} & {\color[HTML]{000000} \textbf{Xerces-1.2}} & {\color[HTML]{000000} 159254} & {\color[HTML]{000000} 71} & {\color[HTML]{000000} 106330} & {\color[HTML]{000000} 52924} & {\color[HTML]{000000} 529.24} & {\color[HTML]{000000} 2243.01} & {\color[HTML]{000000} 745.41} & {\color[HTML]{000000} 1497.61} \\
{\color[HTML]{000000} \textbf{37}} & {\color[HTML]{000000} \textbf{Xerces-1.3}} & {\color[HTML]{000000} 167095} & {\color[HTML]{000000} 69} & {\color[HTML]{000000} 94588} & {\color[HTML]{000000} 72507} & {\color[HTML]{000000} 725.07} & {\color[HTML]{000000} 2421.67} & {\color[HTML]{000000} 1050.83} & {\color[HTML]{000000} 1370.84} \\
{\color[HTML]{000000} \textbf{38}} & {\color[HTML]{000000} \textbf{Xerces-1.4}} & {\color[HTML]{000000} 141180} & {\color[HTML]{000000} 437} & {\color[HTML]{000000} 105066} & {\color[HTML]{000000} 36114} & {\color[HTML]{000000} 361.14} & {\color[HTML]{000000} 323.07} & {\color[HTML]{000000} 82.64} & {\color[HTML]{000000} 240.43} \\ \hline
{\color[HTML]{000000} \textbf{}} & {\color[HTML]{000000} \textbf{Total}} & {\color[HTML]{000000} \textbf{4737234}} & {\color[HTML]{000000} \textbf{5494}} & {\color[HTML]{000000} \textbf{2308586}} & {\color[HTML]{000000} \textbf{2428648}} & {\color[HTML]{000000} \textbf{24286.48}} & {\color[HTML]{000000} \textbf{66224.21}} & {\color[HTML]{000000} \textbf{24910.68}} & {\color[HTML]{000000} \textbf{41313.53}} \\ \hline
\multicolumn{10}{c}{{\color[HTML]{000000} \textbf{NASA Projects}}} \\ \hline
{\color[HTML]{000000} \textbf{1}} & {\color[HTML]{000000} \textbf{CM1}} & {\color[HTML]{000000} 15486} & {\color[HTML]{000000} 42} & {\color[HTML]{000000} 9925} & {\color[HTML]{000000} 5561} & {\color[HTML]{000000} 55.61} & {\color[HTML]{000000} 368.71} & {\color[HTML]{000000} 132.40} & {\color[HTML]{000000} 236.31} \\
{\color[HTML]{000000} \textbf{2}} & {\color[HTML]{000000} \textbf{JM1}} & {\color[HTML]{000000} 376794} & {\color[HTML]{000000} 1759} & {\color[HTML]{000000} 240094} & {\color[HTML]{000000} 136700} & {\color[HTML]{000000} 1367.00} & {\color[HTML]{000000} 214.21} & {\color[HTML]{000000} 77.71} & {\color[HTML]{000000} 136.49} \\
{\color[HTML]{000000} \textbf{3}} & {\color[HTML]{000000} \textbf{KC1}} & {\color[HTML]{000000} 42706} & {\color[HTML]{000000} 325} & {\color[HTML]{000000} 26221} & {\color[HTML]{000000} 16485} & {\color[HTML]{000000} 164.85} & {\color[HTML]{000000} 131.40} & {\color[HTML]{000000} 50.72} & {\color[HTML]{000000} 80.68} \\
{\color[HTML]{000000} \textbf{4}} & {\color[HTML]{000000} \textbf{KC3}} & {\color[HTML]{000000} 6399} & {\color[HTML]{000000} 36} & {\color[HTML]{000000} 3734} & {\color[HTML]{000000} 2665} & {\color[HTML]{000000} 26.65} & {\color[HTML]{000000} 177.75} & {\color[HTML]{000000} 74.03} & {\color[HTML]{000000} 103.72} \\
{\color[HTML]{000000} \textbf{5}} & {\color[HTML]{000000} \textbf{MC1}} & {\color[HTML]{000000} 66583} & {\color[HTML]{000000} 68} & {\color[HTML]{000000} 62137} & {\color[HTML]{000000} 4446} & {\color[HTML]{000000} 44.46} & {\color[HTML]{000000} 979.16} & {\color[HTML]{000000} 65.38} & {\color[HTML]{000000} 913.78} \\
{\color[HTML]{000000} \textbf{6}} & {\color[HTML]{000000} \textbf{MC2}} & {\color[HTML]{000000} 5503} & {\color[HTML]{000000} 44} & {\color[HTML]{000000} 2006} & {\color[HTML]{000000} 3497} & {\color[HTML]{000000} 34.97} & {\color[HTML]{000000} 125.07} & {\color[HTML]{000000} 79.48} & {\color[HTML]{000000} 45.59} \\
{\color[HTML]{000000} \textbf{7}} & {\color[HTML]{000000} \textbf{MW1}} & {\color[HTML]{000000} 6905} & {\color[HTML]{000000} 27} & {\color[HTML]{000000} 4390} & {\color[HTML]{000000} 2515} & {\color[HTML]{000000} 25.15} & {\color[HTML]{000000} 255.74} & {\color[HTML]{000000} 93.15} & {\color[HTML]{000000} 162.59} \\
{\color[HTML]{000000} \textbf{8}} & {\color[HTML]{000000} \textbf{PC1}} & {\color[HTML]{000000} 23020} & {\color[HTML]{000000} 61} & {\color[HTML]{000000} 15499} & {\color[HTML]{000000} 7521} & {\color[HTML]{000000} 75.21} & {\color[HTML]{000000} 377.38} & {\color[HTML]{000000} 123.30} & {\color[HTML]{000000} 254.08} \\
{\color[HTML]{000000} \textbf{9}} & {\color[HTML]{000000} \textbf{PC2}} & {\color[HTML]{000000} 17834} & {\color[HTML]{000000} 16} & {\color[HTML]{000000} 17069} & {\color[HTML]{000000} 765} & {\color[HTML]{000000} 7.65} & {\color[HTML]{000000} 1114.63} & {\color[HTML]{000000} 47.81} & {\color[HTML]{000000} 1066.81} \\
{\color[HTML]{000000} \textbf{10}} & {\color[HTML]{000000} \textbf{PC3}} & {\color[HTML]{000000} 33016} & {\color[HTML]{000000} 140} & {\color[HTML]{000000} 27151} & {\color[HTML]{000000} 5865} & {\color[HTML]{000000} 58.65} & {\color[HTML]{000000} 235.83} & {\color[HTML]{000000} 41.89} & {\color[HTML]{000000} 193.94} \\
{\color[HTML]{000000} \textbf{11}} & {\color[HTML]{000000} \textbf{PC4}} & {\color[HTML]{000000} 30055} & {\color[HTML]{000000} 178} & {\color[HTML]{000000} 22981} & {\color[HTML]{000000} 7074} & {\color[HTML]{000000} 70.74} & {\color[HTML]{000000} 168.85} & {\color[HTML]{000000} 39.74} & {\color[HTML]{000000} 129.11} \\
{\color[HTML]{000000} \textbf{12}} & {\color[HTML]{000000} \textbf{PC5}} & {\color[HTML]{000000} 161695} & {\color[HTML]{000000} 516} & {\color[HTML]{000000} 99849} & {\color[HTML]{000000} 61846} & {\color[HTML]{000000} 618.46} & {\color[HTML]{000000} 313.36} & {\color[HTML]{000000} 119.86} & {\color[HTML]{000000} 193.51} \\ \hline
{\color[HTML]{000000} \textbf{}} & {\color[HTML]{000000} \textbf{Total}} & {\color[HTML]{000000} \textbf{785996}} & {\color[HTML]{000000} \textbf{3212}} & {\color[HTML]{000000} \textbf{531056}} & {\color[HTML]{000000} \textbf{254940}} & {\color[HTML]{000000} \textbf{2549.40}} & {\color[HTML]{000000} \textbf{4462.09}} & {\color[HTML]{000000} \textbf{945.48}} & {\color[HTML]{000000} \textbf{3516.61}} \\ \hline
\multicolumn{10}{c}{{\color[HTML]{000000} \textbf{AEEEM Projects}}} \\ \hline
{\color[HTML]{000000} \textbf{1}} & {\color[HTML]{000000} \textbf{Eclipse}} & {\color[HTML]{000000} 224055} & {\color[HTML]{000000} 206} & {\color[HTML]{000000} 91682} & {\color[HTML]{000000} 132373} & {\color[HTML]{000000} 1323.73} & {\color[HTML]{000000} 1087.65} & {\color[HTML]{000000} 642.59} & {\color[HTML]{000000} 445.06} \\
{\color[HTML]{000000} \textbf{2}} & {\color[HTML]{000000} \textbf{Equinox}} & {\color[HTML]{000000} 39534} & {\color[HTML]{000000} 129} & {\color[HTML]{000000} 10541} & {\color[HTML]{000000} 28993} & {\color[HTML]{000000} 289.93} & {\color[HTML]{000000} 306.47} & {\color[HTML]{000000} 224.75} & {\color[HTML]{000000} 81.71} \\
{\color[HTML]{000000} \textbf{3}} & {\color[HTML]{000000} \textbf{Lucene}} & {\color[HTML]{000000} 73184} & {\color[HTML]{000000} 64} & {\color[HTML]{000000} 60449} & {\color[HTML]{000000} 12735} & {\color[HTML]{000000} 127.35} & {\color[HTML]{000000} 1143.50} & {\color[HTML]{000000} 198.98} & {\color[HTML]{000000} 944.52} \\
{\color[HTML]{000000} \textbf{4}} & {\color[HTML]{000000} \textbf{Mylyn}} & {\color[HTML]{000000} 156102} & {\color[HTML]{000000} 245} & {\color[HTML]{000000} 112333} & {\color[HTML]{000000} 43769} & {\color[HTML]{000000} 437.69} & {\color[HTML]{000000} 637.15} & {\color[HTML]{000000} 178.65} & {\color[HTML]{000000} 458.50} \\
{\color[HTML]{000000} \textbf{5}} & {\color[HTML]{000000} \textbf{PDE}} & {\color[HTML]{000000} 146952} & {\color[HTML]{000000} 209} & {\color[HTML]{000000} 106900} & {\color[HTML]{000000} 40052} & {\color[HTML]{000000} 400.52} & {\color[HTML]{000000} 703.12} & {\color[HTML]{000000} 191.64} & {\color[HTML]{000000} 511.48} \\ \hline
{\color[HTML]{000000} } & {\color[HTML]{000000} \textbf{Total}} & {\color[HTML]{000000} 639827} & {\color[HTML]{000000} 853} & {\color[HTML]{000000} 381905} & {\color[HTML]{000000} 257922} & {\color[HTML]{000000} 2579.22} & {\color[HTML]{000000} 3877.88} & {\color[HTML]{000000} \textbf{1436.61}} & {\color[HTML]{000000} 2441.27} \\ \hline
\end{tabular}}
\end{table*}

{This section describes the proposed HIEL model's performance in terms of the proposed measures}. 

{The figure \ref{PSB-RST} depicts the percent of saved budget and the percent of remaining edit ratios, which are obtained from the HIEL model on all the datasets. It is observed from the figure \ref{PSB-RST} that, out of 55 target projects, 36 target projects account for more than 50\% of the budget savings from the HIEL model. Within the 36 projects, 9 projects such as \textit{Camel-1.0, JEdit-4.3, Redaktor, Synapse-1.0, MC1, PC2, PC3, PC4}, and \textit{Lucene} account for more than 75\% of the budget savings from the original allocated budget. Surprisingly, the projects \textit{Camel-1.0, MC1}, and \textit{PC2} greatly benefit from the use of the HIEL model, as these projects account for more than 90\% of the budget savings from the original allocated budget.}

{The figure \ref{PSB-RST} also represents the amount of service time required to remove the defect content. The vertical bars on top of the dark green bars indicate the percent of remaining edits that are present in each target project. From the figure \ref{PSB-RST}, it is observed that the testers need to conduct more than 50\% code-walk on 19 projects to observe the total defects in the respective projects. Of which, only 7 projects, such as \textit{Log4j-1.2, Lucene-2.2, Lucene-2.4, Poi-1.5, Poi-2.5, Poi-3.0}, and \textit{Xalan-2.6} require more than 75\% of the code walk to observe the total defects. Among the 7 projects, \textit{Log4j-1.2} requires more than 90\% of the code-walk to observe the total defects.}

{Table \ref{SupplementaryDetails}\footnote[1]{A long table is provided in the following link that gives complete experimental results: \href{https://github.com/ekamnit/CPDP-HIEL}{https://github.com/ekamnit/CPDP-HIEL}} presents the supplementary details of the proposed measures. These details provide more insights on the predictions. From the table \ref{SupplementaryDetails}, it is observed that, on an average, the HIEL on the PROMISE projects contributes to 49.70\% of the total savings (nearly 2.3 million cost units out of 4.7 million allocated cost units). The remaining 50.30\% lines of code (nearly 2.4 million LoC) need to be investigated by the testers to observe the defects. This is equivalent to the testers spending 24,286 hours of time (assume, from equation \ref{ProjectHours}, a group of testers service 100 LoC ($\delta=100$) in an hour) on the respective modules. In the process of conducting the code-walk, on an average, for every 24,911 LoC, the testers will find one defect. If the group of testers would involved in testing the same developed code, then the original editing rate to observe one defect is approximately 66,225 LoC/defect. Hence, by using the HIEL model, on an average, the testers can avoid nearly 41,314 edits.}

{On NASA and AEEEM projects (from the table \ref{SupplementaryDetails}), the HIEL model attributes to an improved budget savings of 68.70\% cost units (531,056 cost units out of 785,996 allocated cost units). On the remaining 31.30\% of the LoC (254,940 time units), the testers need to investigate such code to observe the defects. This is equivalent to the testers spending 2549 hours of time on the respective modules. In the process of conducting the code-walk, on an average, for every 945 LoC, the testers will find one defect. If the group of testers were involved in testing the same developed code, then the original editing rate to observe one defect is approximately 4462 LoC/defect. Hence, by using the HIEL model, on an average, the testers can avoid nearly 3517 edits.}

{Similarly, on the AEEEM projects, the HIEL tries to save 58.98\% of the budget (which is equivalent to 381,905 cost units) and leaves the testers about 41.02\% of the LoC (which is equivalent to 257,922 time units) for the investigation of the defects. In this case, for the complete code-walk, the testers have to engage with the code for nearly 2579 hours of time to detect the complete defects. While conducting the code-walk, on an average, the testers may observe one defect for every 1437 LoC. If the group of testers were involved in testing the same developed code, then the original editing rate to observe one defect is approximately 3878 LoC/defect. Hence, by using the HIEL model, on an average, the testers can avoid nearly 2441 edits. However, to achieve the desired goal of the defect prediction task, the prediction model should have to reduce the remaining service time (consequently, it improves the savings in the total allocated budget) to the maximum possible extent.}

{On the other hand, the figure \ref{FailureRate} represents the failure rate (in terms of {FOR}) in each target project. From the figure \ref{FailureRate}, it is observed that, among all the target projects, the projects such as \textit{Log4j-1.2, Lucene-2.2, Lucene-2.4, Poi-1.5, Poi-2.5, Poi-3.0}, and \textit{Velocity-1.4} may experience more than 50\% of the failure incidents, due to dormant defects. Among which, the project \textit{Log4j-1.2} has more than 90\% failure incidents. Note that, because each project is different in size (in terms of modules and the number of defects), the above mentioned percentages are proportional to the failure condition of the individual project. In total, the HIEL results in 5494, 3212, and 853 hidden defects (dormant defects) in the PROMISE, NASA, and AEEEM projects, respectively. On an average, this is equivalent to 28.44\%, 10.87\%, and 18.30\% of the failure chances in the PROMISE, NASA, and AEEEM projects, respectively. However, to achieve a safe system, the prediction model needs to nullify the number of failure conditions (that is, {FOR}) from the predictions.}
\subsection{{Comparative Analysis Using Proposed Measures}}
\label{Comparative Analysis Using PM}
\begin{table*}
\caption{{Average performances of the models such as HIEL, HYDRA, TPTL in terms of the proposed measures.}}
\label{Table: Comparative Performances Using Proposed Measures}
\resizebox{\textwidth}{!}{
\begin{tabular}{c|c|ccc|ccc|ccc|ccc|ccc}
\hline
{ } & { } & \multicolumn{3}{c|}{{ \textbf{PPC}}} & \multicolumn{3}{c|}{{ \textbf{PSB}}} & \multicolumn{3}{c|}{{ \textbf{PNPC}}} & \multicolumn{3}{c|}{{ \textbf{PRE}}} & \multicolumn{3}{c}{{ \textbf{FOR}}} \\ \cline{3-17} 
\multirow{-2}{*}{{ \textbf{S.No}}} & \multirow{-2}{*}{{ \textbf{Repository}}} & \multicolumn{1}{c|}{{ \textbf{HIEL}}} & \multicolumn{1}{c|}{{ \textbf{HYDRA}}} & { \textbf{TPTL}} & \multicolumn{1}{c|}{{ \textbf{HIEL}}} & \multicolumn{1}{c|}{{ \textbf{HYDRA}}} & { \textbf{TPTL}} & \multicolumn{1}{c|}{{ \textbf{HIEL}}} & \multicolumn{1}{c|}{{ \textbf{HYDRA}}} & { \textbf{TPTL}} & \multicolumn{1}{c|}{{ \textbf{HIEL}}} & \multicolumn{1}{c|}{{ \textbf{HYDRA}}} & { \textbf{TPTL}} & \multicolumn{1}{c|}{{ \textbf{HIEL}}} & \multicolumn{1}{c|}{{ \textbf{HYDRA}}} & { \textbf{TPTL}} \\ \hline
{ \textbf{1}} & { \textbf{PROMISE}} & \multicolumn{1}{c|}{{ \textbf{0.6455}}} & \multicolumn{1}{c|}{{ 0.4547}} & { 0.3517} & \multicolumn{1}{c|}{{ \textbf{0.4970}}} & \multicolumn{1}{c|}{{ 0.3260}} & { 0.3129} & \multicolumn{1}{c|}{{ \textbf{0.3545}}} & \multicolumn{1}{c|}{{ 0.5453}} & { 0.6483} & \multicolumn{1}{c|}{{ \textbf{0.5030}}} & \multicolumn{1}{c|}{{ 0.6740}} & { 0.6871} & \multicolumn{1}{c|}{{ 28.44}} & \multicolumn{1}{c|}{{ \textbf{28.29}}} & { 30.66} \\ 
{ \textbf{2}} & { \textbf{NASA}} & \multicolumn{1}{c|}{{ \textbf{0.8023}}} & \multicolumn{1}{c|}{{ 0.7089}} & { 0.6720} & \multicolumn{1}{c|}{{ \textbf{0.6870}}} & \multicolumn{1}{c|}{{ 0.6099}} & { 0.5329} & \multicolumn{1}{c|}{{ \textbf{0.1977}}} & \multicolumn{1}{c|}{{ 0.2911}} & { 0.3280} & \multicolumn{1}{c|}{{ \textbf{0.3130}}} & \multicolumn{1}{c|}{{ 0.3901}} & { 0.4671} & \multicolumn{1}{c|}{{ 10.87}} & \multicolumn{1}{c|}{{ \textbf{10.73}}} & { 12.90} \\ 
{ \textbf{3}} & { \textbf{AEEEM}} & \multicolumn{1}{c|}{{ \textbf{0.7826}}} & \multicolumn{1}{c|}{{ 0.7495}} & { 0.6650} & \multicolumn{1}{c|}{{ \textbf{0.5898}}} & \multicolumn{1}{c|}{{ 0.5355}} & { 0.5198} & \multicolumn{1}{c|}{{ \textbf{0.2174}}} & \multicolumn{1}{c|}{{ 0.2505}} & { 0.3350} & \multicolumn{1}{c|}{{ \textbf{0.4102}}} & \multicolumn{1}{c|}{{ 0.4645}} & { 0.4802} & \multicolumn{1}{c|}{{ 18.30}} & \multicolumn{1}{c|}{{ \textbf{17.96}}} & { 20.17} \\ \hline
{ } & { \textbf{Average}} & \multicolumn{1}{c|}{{ \textbf{0.7435}}} & \multicolumn{1}{c|}{{ 0.6377}} & { 0.5629} & \multicolumn{1}{c|}{{ \textbf{0.5913}}} & \multicolumn{1}{c|}{{ 0.4905}} & { 0.4552} & \multicolumn{1}{c|}{{ \textbf{0.2565}}} & \multicolumn{1}{c|}{{ 0.3623}} & { 0.4371} & \multicolumn{1}{c|}{{ \textbf{0.4087}}} & \multicolumn{1}{c|}{{ 0.5095}} & { 0.5448} & \multicolumn{1}{c|}{{ 19.20}} & \multicolumn{1}{c|}{{ \textbf{18.99}}} & { 21.24} \\ \hline
{ } & { \textbf{Improvement}} & \multicolumn{1}{c|}{{ -}} & \multicolumn{1}{c|}{{ \textbf{16.59\% {\color[HTML]{036400}$\Uparrow$}}}} & { \textbf{32.08\%{\color[HTML]{036400}$\Uparrow$}}} & \multicolumn{1}{c|}{{ -}} & \multicolumn{1}{c|}{{ \textbf{20.55\%{\color[HTML]{036400}$\Uparrow$}}}} & { \textbf{29.90\%{\color[HTML]{036400}$\Uparrow$}}} & \multicolumn{1}{c|}{{ -}} & \multicolumn{1}{c|}{{ \textbf{41.25\%{\color[HTML]{036400}$\Downarrow$}}}} & { \textbf{70.41\%{\color[HTML]{036400}$\Downarrow$}}} & \multicolumn{1}{c|}{{ -}} & \multicolumn{1}{c|}{{ \textbf{24.66\%{\color[HTML]{036400}$\Downarrow$}}}} & { \textbf{33.30\%{\color[HTML]{036400}$\Downarrow$}}} & \multicolumn{1}{c|}{{ -}} & \multicolumn{1}{c|}{{ 1.11\%\textcolor{red}{$\Downarrow$}}} & { \textbf{10.63\%{\color[HTML]{036400}$\Downarrow$}}} \\ \hline
\multicolumn{17}{l}{{{\large\textit{The green-coloured arrows (both up and down arrows) indicates that the proposed model achieved better performances than the other models. The red-}}}}\\ 
\multicolumn{17}{l}{{{\large\textit{coloured arrow indicates that the proposed model does not achieve better performance over the other model.}}}}
\end{tabular}
}
\end{table*}
{In section \ref{Comparative Analysis Using TM}, we have provided the comparative analysis using the traditional measures such as F-measure and AUC. In this section, we present a comparative analysis using the proposed measures. As an example, the comparative analysis is provided between the three approaches, such as TPTL, HYDRA, and the proposed HIEL model. We have chosen two models such as HYDRA and TPTL, because from the empirical evaluations (in section \ref{Comparative Analysis Using TM}), it is observed that these two models are the strongest competitors to our proposed approach. Nonetheless, numerous approaches may utilise the proposed measures to evaluate the predictions.}

{Table \ref{Table: Comparative Performances Using Proposed Measures}\footnote[1]{{A long table is provided in the following link that gives complete experimental results of the proposed HIEL model, along with HYDRA and TPTL, respectively, on all the target projects from the three repositories such as PROMISE, NASA, and AEEEM, respectively, using the proposed measures: \href{https://github.com/ekamnit/CPDP-HIEL/blob/main/CPDP-Review-Results.xlsx}{https://github.com/ekamnit/CPDP-HIEL/blob/main/CPDP-Review-Results.xlsx}}} presents the average performances of the proposed HIEL model, along with HYDRA and TPTL, respectively, on the three repositories such as PROMISE, NASA, and AEEEM, using the proposed measures. When compared with the traditional measures, in the majority of cases, the proposed model shows a clear improvement in terms of the proposed measures on the base-lines such as HYDRA and TPTL. From table \ref{Table: Comparative Performances Using Proposed Measures}, it is observed that, on an average across all the repositories, the proposed HIEL model achieved an improvement of 16.59\% and 32.08\%, respectively, over the base-line models such as HYDRA and TPTL, respectively, in terms of the percent of perfect cleans. This is an indication of accurate predictions in terms of actual clean modules. Hence, consequently, the proposed model is showing an improvement of 20.55\% and 29.90\%, respectively, over the base-line models such as HYDRA and TPTL, in terms of the percent of saved budget. This indicates that, with the use of the proposed model, on an average, the target projects may benefit in terms of savings in the total allocated budget.} 

{Similarly, from table \ref{Table: Comparative Performances Using Proposed Measures}, it is observed that, on an average across all the repositories, the proposed HIEL model achieved an improvement of 41.25\% and 70.41\%, respectively, over the base-line models such as HYDRA and TPTL, in terms of the percent of non-perfect cleans. This indicates that, when compared with the other models, on average, the proposed model is more accurate in predicting the defective instances. Hence, consequently, the proposed model is showing an improvement of 24.66\% and 33.30\%, respectively, over the base-line models such as HYDRA and TPTL, in terms of the percent of remaining edits. This indicates that, when compared with the other models, on average, the testers may benefit in terms of savings in the remaining edits on the target projects.}

{In terms of average failure cases (that is, FOR) across all the target projects, the HYDRA model outperformed the other models. However, the HYDRA model when compared with the proposed model, achieved a negligible improvement of 1.11\% on the average of all the target projects. This shows that, on an average, the HYDRA model is faintly better at reducing false negative instances on the majority of the target projects. But when compared with the TPTL model, the proposed model shows an improvement of 10.63\% on the average of all the target projects. In summary, on an average across all utilised projects, when compared in terms of savings in the allocated budget and the remaining edits, the proposed model is clearly showing its improvement over the base-lines such as HYDRA and TPTL.}

\subsection{Discussion}
\label{discussion}
\subsubsection{On the Performance of HIEL using PWMV}
In section \ref{Comparative Analysis Using TM}, we have seen the comparative analysis of the HIEL over the other models. In this section, we discuss the change in the performance with the use of different sets of diverse inducers in the proposed HIEL model.

\begin{figure*}[ht]
 \subfloat[Performances of PWMV at $\beta$ = 0.1]{
  \label{Beta0.1}
	\begin{minipage}[c][1\width]{
	   0.3\textwidth}
	   \centering
	   \includegraphics[width=1.1\textwidth]{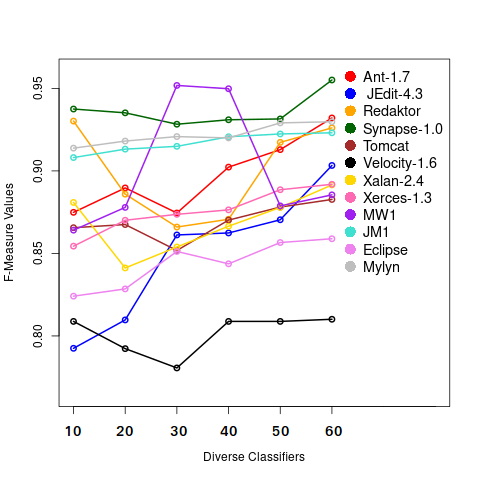}
	\end{minipage}}
 \hfill 	
 \subfloat[Performances of PWMV at $\beta$ = 0.2]{
  \label{Beta0.2}
	\begin{minipage}[c][1\width]{
	   0.3\textwidth}
	   \centering
	   \includegraphics[width=1.1\textwidth]{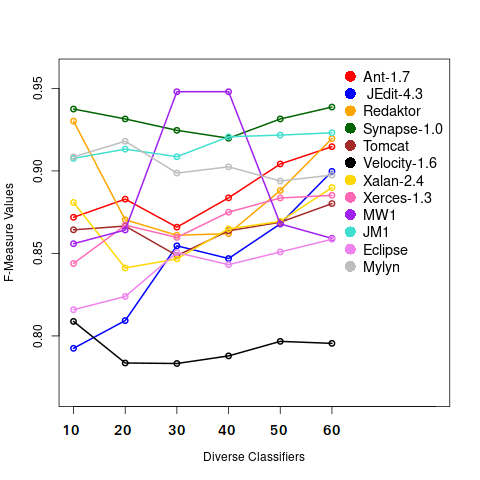}
	\end{minipage}}
 \hfill	
 \subfloat[Performances of PWMV at $\beta$ = 0.3]{
  \label{Beta0.3}
	\begin{minipage}[c][1\width]{
	   0.3\textwidth}
	   \centering
	   \includegraphics[width=1.1\textwidth]{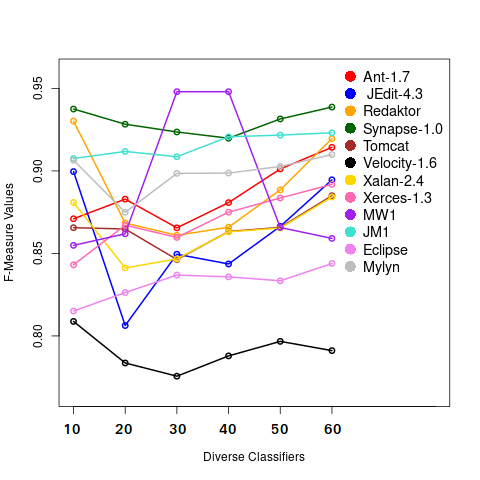}
	\end{minipage}}
\caption{Variation in the {F-measure} values of HIEL using PWMV on 12 randomly selected projects (from three repositories) at different values of $\beta$. For each project, the performances are represented in a line of points for each set of utilised diverse classifiers}
\label{MistakesFigure}
\end{figure*}
From the figure \ref{MistakesFigure} it is observed that the performances on the target projects are proportional to the {utilisation of a} diverse set of classifiers. This satisfies the definition of ensemble learning, as the number of diverse classifiers increases infinitely, then the final decision on the test observation approaches to 1. From the experiments, it is observed that, except for the few projects, the final performances are increased as a result of an increase in the diverse set of classifiers. For example, on the project {\textit{Redaktor}, at $\beta=0.1$}, the {F-measure} value obtained is {0.9302} when utilising the set of 10 diverse classifiers, {which is higher than the value of {F-measure} when} utilising the set of 60 diverse classifiers. From the corollary \ref{corollary1}, we derive an analysis that, if the expert (taken from any set of inducers) is continuously making mistakes, then it is penalised with the decremented weights. Hence, in very few cases like {the} above, few decision makers may participate in obtaining the final decision. 

Similarly, at $\beta=0.2$, {the performances on all the target projects are proportional to the utilised diverse set of classifiers except for the projects \textit{Redaktor, Veocity-1.6}, and \textit{Mylyn}.} When using the set of 60 different classifiers, projects such as {\textit{JEdit-4.3, Redaktor, Veocity-1.6}, and \textit{Mylyn}} obtained decremented {F-measure} values at $\beta=0.3$. In other cases, the projects {show} better performances with the use of more diverse classifiers. This indicates the advantage of using more diverse classifiers to obtain generalizable solutions on the target datasets.

\subsubsection{{On the Comparison Between the Performance Measures}}
\begin{figure}
    \centering
    \includegraphics[width = \linewidth]{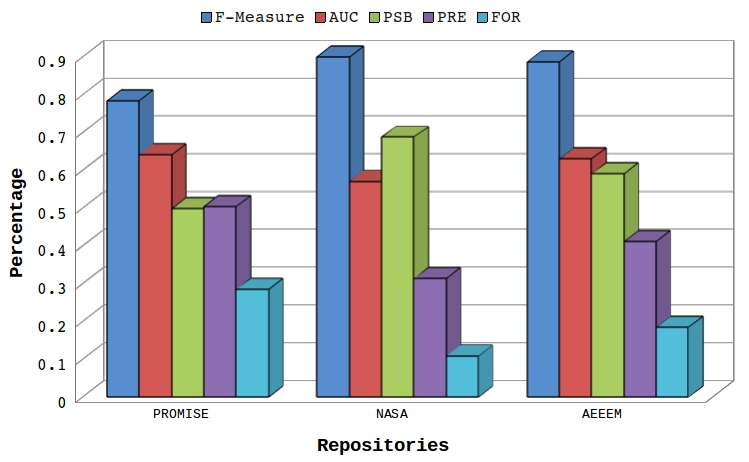}
    \caption{The traditional and proposed measures on all the target projects}
    \label{AllMeasures}
\end{figure}
In section \ref{Cost,Service Time and Failure Analysis}, we have provided the {experimental evaluations in terms of savings in the allocated total budget (using PSB/Saved budget)} and the editing rates to achieve a failure-free software system. {In this section, we discuss the utilised performance measures.}

Figure \ref{AllMeasures} represents the {average performance of HIEL (in terms of all the utilised measures), observed on the projects of PROMISE, NASA, and AEEEM repositories. From the table \ref{AllMeasures}, it is observed that, even though the HIEL model exhibits the better {F-measure} and {AUC} values on all the repositories, the measures such as PSB and PRE indicate nearly half the benefits on the PROMISE projects. That is, on PROMISE projects, the testers still need to put effort into more than half the original code to remove the defects. In terms of PSB and PRE, the HIEL produces the best results on NASA and AEEEM projects. In addition to that, on an average, the failure rates in NASA are low when compared with the failure rates in the PROMISE and AEEEM repositories.}

{However, the above results show that, while traditional measures produce better results, these measures also provide an additional analysis of the obtained performances. Hence, in a critical application like SDP (in this case, CPDP), to understand the true benefits of the prediction model, we suggest using the proposed measures such as PSB, PRE, and FOR in addition to the traditional measures.}
\section{Threats to Validity}
\label{threats}
Since this model builds up on the machine learning model, it is inevitable to observe variation in the final performance at various implementation conditions. Consequently, this prediction result may have an impact on the saved {budget}, remaining service time, and the percent of failures in the target project. For this, the threats to the validity of the proposed model, such as construct validity in \ref{constructValidity}, internal validity in \ref{internalValidity}, and external validity in \ref{externalValidity}, are discussed.
\subsection{Construct Validity}
\label{constructValidity}
The projects in PROMISE, {NASA, and AEEEM have} limited metrics. Inclusion of a greater number of metrics may affect the final performance of HIEL. From the definition of ensemble learning, the model correctly predicts a class label for the test example if the number of classifiers approaches infinity. But this work utilises up to 60 diverse classifiers (for different experiments) to observe the final decision. However, variation in the performance of the proposed model may be observed by using {a} number of diverse classifiers other than 60. But, if the number of diverse classifiers increases in the ensemble, then the computational complexity of the ensemble algorithm (for instance, HIEL) also increases; hence, the time required to observe the decision increases. Therefore, {a trade-off} between the diverse classifiers and the computational complexity needs to be achieved to obtain better performances.
\subsection{Internal Validity}
\label{internalValidity}
Before training the proposed HIEL model, an optional feature reduction stage can be incorporated to select the relevant features from the original defect data. The studies such as \cite{balogun2019performance, ni2019empirical} suggests to incorporate this optional step to analyse the variation in the performance of the proposed model. Hence, there is a chance of getting improved results for the proposed model that uses any one of the feature reduction algorithms as the preprocessing stage. In contrast to incorporating the feature reduction schemes, selecting specific features such as combination of size, complexity, and object-oriented metrics \cite{laradji2015software} can also be evaluated prior to {training} the proposed model to observe the variation in the final prediction performance. Since each classification model differs {from} the implementation of the other, it is important to investigate the use of implementing feature reduction algorithms as well as the use of specific metrics such as object-oriented metrics and some size metrics as features in the proposed model.

To generate the diverse classifiers, the proposed model uses six inducers such as LR, SVM, DT, NB, K-NN, and NN. Nonetheless, the model can be built with the use of other inducers. As decision-makers increase (preferably weak inducers) in the ensemble model, then, {the} probability of getting {the} correct decision approaches 1. Hence, there is a chance of getting better performance when the proposed model {is} tested along with the other inducers.

To train the HIEL model, we have incorporated two diversity generation mechanisms, such as bootstrapping and hybrid-inducers, to generate the bag of classifiers. Nonetheless, the prediction performance of HIEL may be enhanced by including other diversity generation mechanisms such as {optimising} the learning parameters, feature subspace sampling, and changing the output representation methods. But {the} inclusion of these schemes will increase the computational overhead. 
\subsection{External Validity}
\label{externalValidity}
To know the effectiveness of the proposed model, we compared it with the limited {exiting works for} CPDP. The robustness of the proposed model will be determined when {it is compared with the large} number of base-line classifiers and various classic ensemble models. 

To find the transferability of the proposed approach, in the case of PWMV in HIEL, it has to {be tested} with the other {categories} of the training data such as defect severity data and heterogeneous defect data.
\section{Conclusion and Future Works}
\label{conclusion}
This paper presents a novel hybrid-inducer ensemble learning (HIEL) method for cross-project defect prediction, that uses bootstrap aggregation {of} the source project's defects data. Using the proposed method, a bag of 60 classifiers {is} generated by employing six inducers, each with ten training samples. Later, the probabilistic weighted majority voting technique (PWMV) is used as a combiner method on the generated classifiers to get the final decision for the test instance. Using PWMV, during {the} classification stage, a tight upper bound on the number of mistakes made by the best expert (classifier) taken from the proposed HIEL model is derived. We also conducted an empirical evaluation to analyse the variation in the performance of PWMV on the {publicly available defect repositories such as PROMISE, NASA, and AEEEM}. To validate the HIEL model, a comparative analysis {was} conducted with the {recently} published models such as TDS, TCA+, HYDRA, TPTL, and CODEP.
Using the empirical {analysis} on the datasets, the following vital findings are observed:
\begin{itemize}
    \item On the many PROMISE, {NASA, and AEEEM projects}, the HIEL model with PWMV achieved an improved average {in-terms of both} {F-measure} and {AUC} when compared with the other published CPDP models. Among the published models, {in terms of AUC, the TCA+ and HYDRA stand as strong competitors to the HIEL model.}
    \item {The empirical analysis indicates the advantage of utilising the diversity-based ensemble models for this problem context.} From the experimentation on utilising the different sets of diverse classifiers in the PWMV model, it is observed that the obtained performances are proportional to the utilised diverse classifiers on the target projects. {Hence, building complex diverse classifiers has good potential for classifying the defect-proneness of the software module.}
\end{itemize}
This work also aims to fill the gap in analysing the results of the CPDP model on developing software. While the traditional performance measures provide information about the working of the machine learning model, in this work, by making use of the LoC of the test module along with the information from the confusion matrix, the new perspectives from the predictions have been investigated through three performance measures such as {PPC, PNPC}, and {FOR}. Since LoC acts as an additional attribute along with the information from the confusion matrix, these performance metrics are treated as application-dependent measures. That is, these are applicable only to the binary classification of defect prediction models. 

Among the proposed measures, the {PPC} {helps to} provide information {about} the percent of the budget saved, where as the {PNPC} measure {helps to calculate} the amount of work {still remaining} in the newly developed software project, {by utilising the CPDP models}. The {FOR} {measure} is used to estimate the number of failures that may occur in the developed software. As {maximising the} budget savings and minimising the remaining service time and failures {are} the primary goals in introducing SDP research, we recommend future research on SDP (any category, such as WPDP or CPDP) to {focus on providing such} model-specific outcome analysis.
\subsection{Future Directions}
\label{future directions}
{The possible future directions in this research include:} finding the defect associations from the {predicted} result of our proposed CPDP model, {}{which} is the immediate future objective. Finding the severity of the defective module once it is identified as defective by the SDP model, estimating the reliability of the software system after the deployment of {the} SDP model, etc. {are still open challenges in this research field, on which we will focus our future research directions.}
\bibliographystyle{elsarticle-num} 				
\bibliography{CPDP}
\vskip6pt
\end{document}